\documentclass[a4paper, UKenglish, cleveref, autoref, thm-restate, numberwithinsect]{lipics-v2021}

\pdfoutput=1 %
\hideLIPIcs  %

\usepackage{ifthen}

\usepackage{refcount}

\usepackage{amsfonts}

\usepackage{mathtools}

\usepackage[only, llbracket, rrbracket]{stmaryrd}

\usepackage{amsthm}
\usepackage{thmtools}

\usepackage{tikz}
\usetikzlibrary{arrows, calc, shapes, snakes, automata, fit, positioning, intersections}
\tikzset{>=stealth'} %

\tikzstyle{every picture} = [style=semithick]
\tikzstyle{every node}    = [font=\small]
\tikzstyle{every state}   = [thick, minimum size=1mm, inner sep=2pt]

\tikzstyle{initial}   = [initial   by arrow, initial   text=, initial   distance=4mm]
\tikzstyle{accepting} = [accepting by arrow, accepting text=, accepting distance=4mm]

\usepackage[linesnumbered, vlined, ruled]{algorithm2e}
\SetKwInput{Input}{Input}
\SetKwInput{Output}{Output}
\usepackage{subcaption}

\usepackage{hyperref}

\theoremstyle{remark}
\newtheorem{fact}[theorem]{Fact}

\newcommand{\setN}{\mathbb{N}}
\newcommand{\setZ}{\mathbb{Z}}
\newcommand{\setQ}{\mathbb{Q}}

\newcommand{\partie}[1]{\mathbb{P}({#1})}
\newcommand{\pow}[1]{\partie{#1}}

\newcommand{\relimg}[2]{{#1}[{#2}]}
\renewcommand{\vec}[1]{{\mathbf #1}}

\newcommand{\stab}[2]{\operatorname{Stab}_{#1}(#2)}

\newcommand{\con}[1]{\operatorname{Con}(#1)}
\newcommand{\conP}[1]{\operatorname{Con}_{\geq 0}(#1)}
\newcommand{\per}[1]{\operatorname{Per}(#1)}
\newcommand{\perP}[1]{\operatorname{Per}_{\geq 0}(#1)}
\newcommand{\multiset}[1]{\{\!\!\{ #1 \}\!\!\}}
\newcommand{\bvass}{\mathcal{B}}
\newcommand{\bvassexample}{\mathcal{E}}
\newcommand{\vass}{\mathcal{V}}
\newcommand{\post}[1][]{\operatorname{Post}_{#1}}
\newcommand{\ReachSet}[1]{\llbracket #1 \rrbracket}
\newcommand{\step}[1][]{\xRightarrow{#1}}
\newcommand{\instantiate}[2]{#1 \langle #2 \rangle}

\newcommand{\accel}[1]{\operatorname{Add}_{#1}^*}
\newcommand{\clo}[2]{\operatorname{Clo}_{#1, #2}}

\newcommand{\graph}{\mathcal{G}}
\newcommand{\graphbis}{\mathcal{H}}
\newcommand{\explo}{\graph}
\newcommand{\explobis}{\graphbis}
\newcommand{\anc}[2][]{\operatorname{Anc}^{#1}(#2)}
\newcommand{\des}[2][]{\operatorname{Des}^{#1}(#2)}

\newcommand{\WorkList}{W}
\newcommand{\Redundant}{R}

\title{On the Reachability Problem for Two-Dimensional Branching VASS} %

\titlerunning{Reachability for Two-Dimensional Branching VASS} %

\author{Clotilde Bizière}{LaBRI, Univ. Bordeaux, CNRS, Bordeaux INP, Talence, France}{}{}{}
\author{Thibault Hilaire}{LaBRI, Univ. Bordeaux, CNRS, Bordeaux INP, Talence, France}{}{}{}
\author{Jérôme Leroux}{LaBRI, Univ. Bordeaux, CNRS, Bordeaux INP, Talence, France}{}{}{}
\author{Grégoire Sutre}{LaBRI, Univ. Bordeaux, CNRS, Bordeaux INP, Talence, France}{}{}{}

\authorrunning{C. Bizière, T. Hilaire, J. Leroux, and G. Sutre}

\Copyright{Clotilde Bizière, Thibault Hilaire, Jérôme Leroux, and Grégoire Sutre}



\nolinenumbers %

\begin{document}

\maketitle

\begin{abstract}
  Vectors addition systems with states (VASS), or equivalently Petri nets,
are arguably one of the most studied formalisms for the
modeling and analysis of concurrent systems.
A central decision problem for VASS is reachability: whether there exists a run from an initial configuration to a final one.
This problem has been known to be decidable for over forty years, and its complexity has recently been precisely characterized.
Our work concerns the reachability problem for BVASS, a branching generalization of VASS.
In dimension one,
the exact complexity of this problem is known.
In this paper, we prove that the reachability problem for 2-dimensional BVASS is decidable.
In fact, we even show that the reachability set admits a computable semilinear presentation.
The decidability status of the reachability problem for BVASS remains open in higher dimensions.

\end{abstract}

\section{Introduction}\label{sec:introduction}
Vectors addition systems with states (VASS), or equivalently Petri nets,
are arguably one of the most studied formalisms for the
modeling and analysis of concurrent systems. A central decision problem for VASS is reachability: whether there exists a run from an initial configuration to a final one. This problem was shown decidable more than forty years ago~\cite{Mayr84} but its precise complexity was only established a few years ago~\cite{DBLP:conf/lics/LerouxS19,DBLP:journals/jacm/CzerwinskiLLLM21,DBLP:conf/focs/Leroux21}.
Several VASS extensions have been introduced and studied, most notably unordered data nets~\cite{DBLP:journals/fuin/LazicNORW08}, pushdown VASS~\cite{DBLP:conf/fsttcs/AtigG11,Lazic2013}, and branching VASS~\cite{GGS04,DBLP:journals/dmtcs/VermaG05}. But so far, the reachability problem is still open for these models.

\smallskip

One of the first subclasses of VASS for which reachability was shown to be decidable is the class of $2$-dimensional VASS. For this class, Hopcroft and Pansiot devised an algorithm that computes a finite description (more precisely, a semilinear presentation) of the reachability set~\cite{DBLP:journals/tcs/HopcroftP79}. As an immediate consequence, they obtained that reachability is decidable for this class. In fact, the algorithm of Hopcroft and Pansiot can be viewed as a refinement of the classical Karp-Miller algorithm \cite{DBLP:journals/jcss/KarpM69} where the abstract pumping of cycles (putting $\omega$ in some components) is replaced by an exact acceleration of cycles (adding new vectors to the current set of periods).

\smallskip

In this paper, we investigate the reachability problem for branching VASS (shortly called BVASS in the sequel), a branching generalization of VASS. More precisely, BVASS extend VASS with special branching transitions that merge configurations (by summing their vectors). This model has gained a lot of interest recently due to strong links with several fields in computer science such as cryptographic protocols~\cite{DBLP:journals/dmtcs/VermaG05}, linear logic~\cite{GGS04,LazicS15}, recursively parallel programs~\cite{DBLP:journals/toplas/BouajjaniE13}, timed pushdown systems~\cite{DBLP:conf/lics/ClementeLLM17}, computational linguistic~\cite{DBLP:conf/acl/Rambow94,DBLP:conf/acl/Schmitz10}, game semantics~\cite{DBLP:conf/esop/Cotton-BarrattM17}, equational tree automata~\cite{DBLP:conf/csl/Ohsaki01,DBLP:conf/fossacs/Lugiez03} and data logics~\cite{DBLP:journals/corr/JacquemardSD16,DBLP:conf/pods/BojanczykDMSS06}. For instance, provability in the multiplicative exponential fragment of linear logic (MELL) is inter-reducible with the reachability problem in BVASS~\cite{GGS04}.
As mentioned before, the reachability problem is still open in arbitrary dimension for BVASS. In dimension one, the reachability problem is decidable and the exact complexity is known~\cite{GollerHLT16,DBLP:conf/icalp/FigueiraLLMS17}. The objective of our work was to investigate the decidability of the reachability problem for $2$-dimensional BVASS.

\subparagraph*{Contributions.}
In this paper, we prove that the reachability problem for $2$-dimensional BVASS is decidable. In fact, we even show that the reachability set admits a computable semilinear presentation.
We propose an algorithm that essentially performs
a forward symbolic exploration of a $2$-dimensional BVASS given as input.
Our algorithm is inspired from Hopcroft and Pansiot's algorithm for
classical $2$-dimensional VASS~\cite{DBLP:journals/tcs/HopcroftP79}.
The latter computes a symbolic reachability tree,
but in our case we need an acyclic graph, that we call \emph{exploration}, because of branching transition rules.

\smallskip

Compared to Hopcroft and Pansiot's algorithm where pumped cycles are computed statically, in our case pumped cycles are computed dynamically since they exploit configurations that are discovered during the exploration. This feature complicates the proof of soundness of our algorithm. But the main challenge is the proof of termination. As usual, we proceed by contradiction and assume that the algorithm constructs an infinite exploration. A first source of difficulty in order to obtain a contradiction is that the set of pumped cycles is potentially infinite. A second source is the fact that we cannot consider any infinite path of the exploration. In fact, there are mutual dependencies between paths since the exploration is not necessarily a tree. We need to consider an infinite path that ultimately does not depend on the other paths, and we show that such a path always exist. We believe that our proof techniques could be applied to other algorithms that construct potentially infinite acyclic graphs.

\subparagraph*{Related Work.}
In general dimension, the coverability problem (a weak version of the reachability problem) and the boundedness problem are decidable for BVASS~\cite{DBLP:journals/dmtcs/VermaG05}, and their precise complexity is known~\cite{DBLP:journals/jcss/DemriJLL13,LazicS15}. The complexity of the reachability problem for bounded BVASS was established in~\cite{DBLP:conf/concur/MazowieckiP19}. The reachability problem for BVASS and pushdown VASS is still open. For pushdown VASS the problem is known to be decidable in the bidirected case~\cite{DBLP:conf/icalp/GanardiMPSZ22}.
In small dimensions, the above-mentioned idea of pumping cycles was successfully applied to the analysis of several VASS generalizations, and in particular to solve reachability for $2$-dimensional VASS and extensions~\cite{DBLP:journals/jacm/BlondinEFGHLMT21,DBLP:conf/focs/0001CMOSW24,DBLP:conf/fsttcs/FinkelLS18}, coverability for 1-dimensional pushdown VASS~\cite{DBLP:conf/icalp/LerouxST15}, and reachability for 1-dimensional BVASS~\cite{GollerHLT16,DBLP:conf/icalp/FigueiraLLMS17}.

\subparagraph*{Outline.}
Some preliminary background and notations are provided in \cref{sec:preliminaries}.
We define in \cref{sec:BVASS} the model of BVASS and their semantics.
\cref{sec:algorithm} introduces a class of acyclic graphs,
called \emph{explorations},
where nodes are labeled by sets of configurations, and
presents our reachability algorithm for $2$-BVASS.
We show in \cref{sec:correctness} that the explorations constructed by our algorithm,
called \emph{algorithmic} explorations,
are sound and complete
(for the reachability set).
This shows the partial correctness of our algorithm, and we
then focus on its termination.
We prove in \cref{sec:branches} that any infinite graph that admits a finitely-branching spanning forest contains a ``core'' witness of infinity defined as a so-called directed and primary graph.
\cref{sec:termination} provides the proof of termination of our algorithm, and
is decomposed into four subsections.
First,
we study the stabilization of cones, a form of acceleration for cones through the so-called notion of \emph{modes}.
Second,
we show how to decompose the effect of paths in an exploration into a sum of
so-called \emph{elementary} vectors corresponding to previously mentioned \emph{pumped cycles}, and
\emph{consecutive} vectors.
Third,
we prove that the periodic set associated to a primary infinite set of nodes in an algorithmic exploration is finitely-generated.
Fourth, we assemble the results from the previous sections to deduce the termination of our algorithm.
\cref{sec:conclusion} concludes the paper.
\emph{Due to space limitations, detailed proofs are deferred to the appendix.}

\section{Preliminaries}\label{sec:preliminaries}
We denote by $\setZ$ the set of integers, $\setN$ the set of natural numbers, by $\setQ$ the set of rational numbers, and by $\setQ_{\geq 0}$ the set of non-negative rational numbers. We also introduce $\setN_{>0}$ and $\setQ_{>0}$ defined as $\setN\setminus \{0\}$ and $\setQ_{\geq 0}\setminus\{0\}$, respectively. The powerset of a set $S$ is written $\pow{S}$.

\subparagraph*{Vectors.}
Vectors are typeset in bold face. Given $\vec{c}\in\setQ^d$, we let $(\vec{c}(1),\ldots, \vec{c}(d))$ denote the vector of rational numbers defining $\vec{c}$. We write $\vec{x}\leq \vec{y}$ for two vectors $\vec{x},\vec{y}\in\setQ^d$ if $\vec{x}(i)\leq \vec{y}(i)$ for every $i\in\{1,\ldots,d\}$. The sum of two vectors $\vec{x}+\vec{y}$ is defined component-wise.
The sum operator over vectors is extended over sets $\vec{X},\vec{Y}\subseteq\setQ^d$ by $\vec{X}+\vec{Y}=\{\vec{x}+\vec{y}\mid \vec{x}\in\vec{X}\wedge\vec{y}\in\vec{Y}\}$. Given $\vec{X}\subseteq \setQ^d$ and $\vec{x}\in\setQ^d$, we define $\vec{x}+\vec{X}$ as $\{\vec{x}\}+\vec{X}$. The set $\vec{X}+\vec{x}$ is defined similarly. We also write $\setQ_{\geq 0}\vec{X}$ the set $\{\lambda\vec{x} \mid \lambda\in \setQ_{\geq 0}\wedge \vec{x}\in\vec{X}\}$.

\subparagraph*{Periodic sets and semilinear sets.}
A set $\vec{P}\subseteq \setQ^d$ is said to be \emph{periodic} if $\vec{0}\in\vec{P}$ and $\vec{P}+\vec{P}\subseteq\vec{P}$. Given a set $\vec{A}\subseteq\setQ^d$, we denote by $\per{\vec{A}}$ the set of finite sums $\vec{a}_1+\cdots+\vec{a}_k$ where $k\in\setN$, and $\vec{a}_1,\ldots,\vec{a}_k\in\vec{A}$. This periodic set is called the \emph{periodic set spanned} by $\vec{A}$. A periodic set $\vec{P}$ is said to be \emph{finitely-generated} if $\vec{P}=\per{\vec{A}}$ for some finite set $\vec{A}\subseteq \setQ^d$.
The sum of two periodic sets is a periodic set. Given a sequence $(\vec{P}_i)_{i\in I}$ of periodic sets indexed by a finite or infinite set $I$, we denote by $\sum_{i\in I}\vec{P}_i$ the periodic set $\per{\bigcup_{i\in I}\vec{P}_i}$. Notice that that if $I$ is finite, this definition of sum coincides with the previously introduced finite sum of subsets of $\setQ^d$.
A set $\vec{L}\subseteq \setN^d$ is said to be \emph{linear} if $\vec{L}=\vec{b}+\vec{P}$ where $\vec{b} \in \setN^d$ and $\vec{P}\subseteq\setN^d$ is a finitely-generated periodic set. A set $\vec{S} \subseteq \setN^d$ is \emph{semilinear} if $\vec{S}$ is a finite union of linear subsets of~$\setN^d$.

\subparagraph*{Cones.}
A \emph{cone} $\vec{C}$ of $\setQ^d$ is a periodic subset of $\setQ^d$ such that $\setQ_{\geq 0}\vec{C}\subseteq \vec{C}$. The \emph{cone spanned} by a set $\vec{A}\subseteq\setQ^d$ is the cone denoted by $\con{\vec{A}}$ and defined as $\setQ_{\geq 0}\per{\vec{A}}$. %
A cone $\vec{C}\subseteq\setQ^d$ is said to be \emph{finitely-generated} if $\vec{C}=\con{\vec{A}}$ for some finite set $\vec{A}\subseteq \setQ^d$.
\begin{lemma}[{\cite[Lemma 1.2]{DBLP:journals/tcs/HopcroftP79}}]\label{cor:HPcone}
  Let $\vec{P} \subseteq \setZ^d$ be a periodic set. Then $\vec{P}$ is a finitely-generated periodic set if, and only if, $\con{\vec{P}}$ is a finitely-generated cone.
\end{lemma}

\subparagraph*{Graphs.}
A \emph{graph} is a pair $\graph=(N,\rightarrow)$ where $N$ is a set of \emph{nodes}, and $\rightarrow$ is a binary relation on $N$ called the \emph{edge relation}. The graph is said to be \emph{empty} (resp. \emph{finite}, \emph{infinite}) when its set of nodes is empty (resp. finite, infinite). We denote by $\xrightarrow{+}$ the transitive closure of $\rightarrow$, and by $\xrightarrow{*}$ the reflexive closure of $\xrightarrow{+}$.
The graph is called \emph{acyclic} when $\xrightarrow{+}$ is irreflexive.
We associate with a node $n\in N$ the set of \emph{ancestors} $\anc[\graph]{n}=\{m\in N \mid m\xrightarrow{*}n\}$, and the set of \emph{descendants} $\des[\graph]{n}=\{m\in N\mid n\xrightarrow{*}m\}$.
Ancestors and descendant are extended over sets of nodes $X \subseteq N$ as expected,
by $\anc[\graph]{X} = \bigcup_{n \in X} \anc[\graph]{n}$ and $\des[\graph]{X} = \bigcup_{n \in X} \des[\graph]{n}$.
A set of nodes $X \subseteq N$ verifying $X = \anc[\graph]{X}$ is said to be \emph{ancestor-closed}.
A node $n\in N$ is called a \emph{leaf} if there does not exist a node $m\in M$ satisfying $n\rightarrow m$. A node $n$ is called a \emph{source} if there is no node $m$ such that $m\rightarrow n$.
The \emph{restriction} of a graph $\graph=(N,\rightarrow)$ to a set of nodes $X \subseteq N$ is the graph $(X, {\rightarrow} \cap (X \times X))$.
A \emph{node-labeled graph} is a triple $(N, \rightarrow, \lambda)$ where
$(N, \rightarrow)$ is a graph and
$\lambda$ is a function with domain $N$.
The notions defined above for graphs naturally carry over to node-labeled graphs.

\section{Branching VASS}\label{sec:BVASS}
A $d$-dimensional \emph{branching vector addition system with states}
(\emph{$d$-BVASS} for short)
is a pair $\bvass = (Q, \Delta)$ where
$Q$ is a finite non-empty set of \emph{states} and
$\Delta \subseteq (\pow{Q} \times \setZ^d \times Q)$ is a finite set of \emph{transition rules}.
A transition rule $\delta = (S, \vec{a}, q)$ in $\Delta$ consists in
a set $S \subseteq Q$ of \emph{input states},
a \emph{displacement} $\vec{a} \in \setZ^d$, and
a single \emph{output state} $q \in Q$.
Intuitively,
assuming that $S = \{q_1, \ldots, q_k\}$,
this transition rule can be seen as the rewriting rule
$q_1(\mathtt{x}_1), \ldots, q_k(\mathtt{x}_k) \rightarrow q(\vec{a} + \mathtt{x}_1 + \cdots + \mathtt{x}_k)$
with formal parameters $\mathtt{x}_1, \ldots, \mathtt{x}_k$.
Note that our definition forbids a state from occurring twice on the left-hand side of a transition rule
(as this left-hand side is given by a set of states).
This restriction is only a matter of technical convenience.
A transition rule $\delta = (S, \vec{a}, q)$ is called
\emph{initial} when $S = \emptyset$,
\emph{unary} when $|S| = 1$, and
\emph{branching} when $|S| \geq 2$.
A $d$-dimensional \emph{vector addition system with states}
(\emph{$d$-VASS} for short)
is a $d$-BVASS $\vass = (Q, \Delta)$ such that $|S| \leq 1$
for every transition rule $(S, \vec{a}, q)$ in $\Delta$.

\smallskip

We formulate the semantics of a $d$-BVASS $\bvass = (Q, \Delta)$ in terms of a configuration-set transformer $\post$.
A \emph{configuration} of $\bvass$ is a pair $(q, \vec{x})$ in $Q \times \setN^d$,
also written as $q(\vec{x})$ in the sequel.
By extension,
given a set $\vec{X} \subseteq \setN^d$,
we let $q(\vec{X})$ denote the set of configurations $\{q\} \times \vec{X}$.
The set of \emph{initial} configurations of $\bvass$ is
$\{q(\vec{a}) \mid (\emptyset, \vec{a}, q) \in \Delta \text{ and } \vec{a} \geq \vec{0}\}$.
For each transition rule $\delta = (S, \vec{a}, q)$ in $\Delta$,
we introduce the function
$\post[\delta] : \pow{Q \times \setN^d} \rightarrow \pow{Q \times \setN^d}$
defined by\footnote{%
  We use double braces $\multiset{\cdots}$ to denote multisets.
  Here,
  the condition
  $S = \multiset{r \mid r(\vec{z}) \in D}$
  means that,
  firstly,
  the set $S$ is equal to the set $\{r \mid r(\vec{z}) \in D\}$, and,
  secondly,
  $r_1 \neq r_2$ for every two distinct configurations $r_1(\vec{z}_1), r_2(\vec{z}_2)$ in $D$.
}
\begin{align*}
  \post[\delta](C)
  =
  q\left(\left\{
    \vec{y} \in \setN^d
    \mid
    \exists D \subseteq C :
    S = \multiset{r \mid r(\vec{z}) \in D}
    \text{ and }
    \textstyle
    \vec{y} = \vec{a} + \sum_{r(\vec{z}) \in D} \vec{z}
  \right\}\right)
\end{align*}
for every set $C \subseteq Q \times \setN^d$.
We also introduce
$\post[\bvass] : \pow{Q \times \setN^d} \rightarrow \pow{Q \times \setN^d}$,
defined by
$\post[\bvass](C) = \bigcup_{\delta \in \Delta} \post[\delta](C)$.
Note that $\post[\bvass]$ is $\subseteq$-nondecreasing and that
$\post[\bvass](\emptyset)$ coincides with the set of initial configurations of $\bvass$.
The \emph{reachability set} of $\bvass$,
written $\ReachSet{\bvass}$,
is the $\subseteq$-least set $C \subseteq Q \times \setN^d$
such that $\post[\bvass](C) \subseteq C$.

\begin{figure}[t]
  \centering
  \begin{tikzpicture}[node distance=2cm, text centered, ->, bend angle=25]
    \node[state] (r) {$r$};
    \node[state, above left of=r, yshift=2mm]  (p) {$p$};
    \node[state, above right of=r, yshift=2mm] (q) {$q$};
    \node[below of=p, node distance=1.2cm] (m) {};
    \node[state, above right of=p, yshift=2mm] (s) {$s$};
    \draw[->]
    (m) edge             node[left, yshift=-1mm]       {$(4, 4)$} (p)
    (p) edge[]  node[above, yshift = 1.5pt]       {$(-1, 0)$} (q)
    (r) edge[bend right] node[below right] {$(-1, 1)$} (q)
    (p) edge[bend left]  coordinate[pos=0.8] (pr) (r)
    (q) edge[bend right] coordinate[pos=0.8] (qr) (r)
    (r) edge[loop below, looseness=10, in=-120, out=-60] node[below] {$(1,3)$} (r)
    (q) edge[bend right] node[above right] {$(0,0)$} (s)
    (s) edge[bend right] node[above left] {$(0,0)$} (p);
    \draw[-]
    (pr) edge[bend left] node[above, yshift=-0.5mm] {$+$} (qr);

    \coordinate (Orig) at (6.5,3.15);
    \node[draw, rectangle, label=left:$n_0$]  (N0) at ($(Orig) + (0,0)$) {$p, (4,4), \{(0,0)\}$};
    \node[draw, rectangle, label=left:$n_1$]  (N1) at ($(Orig) + (0,-1)$) {$q, (3,4), \{(0,0)\}$};
    \node[draw, rectangle, label=left:$n_2$]  (N2) at ($(Orig) + (0,-2)$) {$s, (3,4), \per{\{(2,5)\}}$};
    \node[draw, rectangle, label=left:$n_3$]  (N3) at ($(Orig) + (0,-3)$) {$p, (3,4), \per{\{(2,5)\}}$};
    \node[draw, rectangle, label=right:$n_4$] (N4) at ($(Orig) + (2,-4)$) {$q, (2,4), \per{\{(2,5)\}}$};
    \node[draw, rectangle, label=left:$n_5$]  (N5) at ($(Orig) + (-2, -4)$) {$r, (6,8), \per{\{(2,5)\}}$};
    \node[draw, rectangle, label=right:$n_6$] (N6) at ($(Orig) + (3.25,-2)$) {$r, (7,8), \{(0,0)\}$};

    \draw[->] (N0) to (N1);
    \draw[->] (N1) to (N2);
    \draw[->] (N2) to (N3);
    \draw[->] (N3) to (N4);
    \draw[->, out=180, in=90] (N1.south west) to (N5.160);
    \draw[->] (N3) to (N5);
    \draw[->] (N1.south east) to (N6);
    \draw[->] (N0.south east) to (N6);
  \end{tikzpicture}
  \caption{%
    The $2$-BVASS $\bvassexample$ from \cref{exa:BVASS} (left) and
    an execution of $\mathtt{Explore}(\bvassexample)$ (right).
  }
  \label{fig:BVASS}
\end{figure}

\begin{example}
  \label{exa:BVASS}
  Consider the $2$-BVASS $\bvassexample$ depicted in \cref{fig:BVASS}.
  It has four states $p, q, r, s$ and
  seven transition rules,
  namely
  $(\emptyset, (4, 4), p)$,
  $(\{p, q\}, (0, 0), r)$ and
  five unary transition rules (see \cref{fig:BVASS}).
  We have $p(4, 4) \in \ReachSet{\bvassexample}$
  since $(\emptyset, (4, 4), p) \in \Delta$.
  From $(\{p\}, (-1, 0), q) \in \Delta$ and $p(4, 4) \in \ReachSet{\bvassexample}$
  we get that $q(3, 4) \in \ReachSet{\bvassexample}$.
  From $(\{p, q\}, (0, 0), r) \in \Delta$ and $p(4, 4), q(3, 4) \in \ReachSet{\bvassexample}$
  we get that $r(7, 8) \in \ReachSet{\bvassexample}$.
  \lipicsEnd
\end{example}

A set of configurations $C \subseteq Q \times \setN^d$ is said to be \emph{semilinear} if $C$ is
a finite union of sets of the form $q(\vec{L})$ where $q \in Q$ and $\vec{L} \subseteq \setN^d$ is linear, i.e. a set of the form $\vec{b}+\per{\vec{A}}$ for some $\vec{b}\in\setN^d$ and some finite subset $\vec{A}$ of $\setN^d$.
A \emph{presentation} of a semilinear set of configurations $C \subseteq Q \times \setN^d$
is a finite set $\{(q_1, \vec{b}_1, \vec{A}_1), \ldots, (q_k, \vec{b}_k, \vec{A}_k)\}$,
where $q_j \in Q$, $\vec{b}_j \in \setN^d$, and $\vec{A}_j$ is a finite subset of $\setN^d$,
such that $C = \bigcup_{j=1}^k q_j(\vec{b}_j+\per{\vec{A}_j})$.
The main contribution of this paper is the generalization to $2$-BVASS of the following theorem.

\begin{theorem}[\cite{DBLP:journals/tcs/HopcroftP79}]
  \label{thm:hp79}
  For every $2$-VASS $\vass$,
  the reachability set $\ReachSet{\vass}$ of $\vass$ is semilinear and
  a presentation of $\ReachSet{\vass}$ is computable from $\vass$.
\end{theorem}

In order to generalize \cref{thm:hp79} to 2-BVASS,
we will extend to BVASS some techniques developed for classical VASS.
A linear (i.e., non-branching) view of BVASS behaviors is required to do so.
Intuitively,
this view is obtained by instantiating transition rules $\delta = (S, \vec{a}, q)$ with $|S| \geq 2$
into unary transition rules.
This instantiation can be performed at the semantic level or at the syntactic level.
Let us make these ideas more concrete.
We assume that $\bvass = (Q, \Delta)$ is a $d$-BVASS for the remainder of this section.

\smallskip

For each transition rule $\delta = (S, \vec{a}, q)$ in $\Delta$,
we introduce the binary relation $\step[\delta]$ on $Q \times \setN^d$ defined as
the set of pairs $(p(\vec{x}), q(\vec{y})) \in (Q \times \setN^d)^2$ such that
$p \in S$ and
there exists a set $D \subseteq \ReachSet{\bvass}$ verifying
$(S \setminus \{p\}) = \multiset{r \mid r(\vec{z}) \in D}$ and
$\vec{y} = \vec{a} + \vec{x} + \sum_{r(\vec{z}) \in D} \vec{z}$.
The binary relation $\step[\delta]$ instantiates the transition rule $\delta$
at the semantic level as
it relies on the reachability set $\ReachSet{\bvass}$ of $\bvass$.
We also introduce the binary \emph{step} relation $\step$ on $Q \times \setN^d$
defined as the union $\bigcup_{\delta \in \Delta} \step[\delta]$.
The reflexive-transitive closure of $\step$ is denoted by $\step[*]$.
It is readily seen that $\step$ and $\step[*]$ are diagonal.\footnote{
  \label{footnote:diagonal}
  A binary relation $\bowtie$ on $Q \times \setN^d$ is called \emph{diagonal} if
  $p(\vec{x}) \bowtie q(\vec{y})$ implies $p(\vec{x} + \vec{u}) \bowtie q(\vec{y} + \vec{u})$,
  for every configurations $p(\vec{x}), q(\vec{y}) \in Q \times \setN^d$ and vector $\vec{u} \in \setN^d$.
}
Using square brackets to denote relational images\footnote{%
  Given a binary relation $\bowtie$ on a set $S$ and a subset $X$ of $S$,
  we let $\relimg{\bowtie}{X}$ denote the \emph{relational image} of $X$ under $\bowtie$,
  defined by $\relimg{\bowtie}{X} = \{y \in S \mid \exists x \in X : x \bowtie y\}$.
}, we have
$\post[\bvass](C) \subseteq \relimg{\step}{C} \subseteq \relimg{\step[*]}{C} \subseteq \ReachSet{\bvass}$
for every set of configurations $C \subseteq \ReachSet{\bvass}$.

\smallskip

Let us now instantiate transition rules at the syntactic level.
Given a finite set $F \subseteq Q \times \setN^d$,
the \emph{instantiation of $\bvass$ with $F$},
written $\instantiate{\bvass}{F}$,
is the $d$-VASS $\instantiate{\bvass}{F} = (Q, \Delta')$ where
$\Delta'$ is the set of triples
$(\{p\}, \vec{a}', q)$ such that
there exist a transition rule $(S, \vec{a}, q) \in \Delta$ with $p \in S$ and
a set $D \subseteq F$ verifying
$(S \setminus \{p\}) = \multiset{r \mid r(\vec{z}) \in D}$ and
$\vec{a}' = \vec{a} + \sum_{r(\vec{z}) \in D} \vec{z}$.
We observe that if $F \subseteq \ReachSet{\bvass}$ then
the step relation of $\instantiate{\bvass}{F}$ is contained in
the step relation of $\bvass$.

\begin{remark}
  In the definition of the instantiation $\instantiate{\bvass}{F}$,
  we require $F$ to be finite solely to ensure that $\Delta'$ is finite.
  We could drop this requirement and obtain an ``infinite $d$-VASS'',
  meaning that its set of transition rules is potentially infinite.
  The step relations of $\bvass$ and of the resulting ``infinite $d$-VASS''
  $\instantiate{\bvass}{\ReachSet{\bvass}}$ coincide.
\end{remark}

\smallskip

We conclude this section with notions that are specific to classical VASS.
Consider a $d$-VASS $\vass = (Q, \Delta)$.
From now on,
unary transition rules $(\{p\}, \vec{a}, q) \in \Delta$ will be written $(p, \vec{a}, q)$ for short.
A \emph{path} of $\vass$ is a non-empty sequence
$\theta = (p_1, \vec{a}_1, q_1) \cdots (p_k, \vec{a}_k, q_k)$
of unary transition rules $(p_i, \vec{a}_i, q_i) \in \Delta$
such that
$q_i = p_{i+1}$ for all $i \in \{1, \ldots, k-1\}$.
We call $p_1$ and $q_k$ the \emph{start} and the \emph{end} of $\theta$, respectively.
The \emph{displacement} of $\theta$ is $\sum_{i=1}^k \vec{a}_i$.
We say that $\theta$ is a \emph{cycle} if $p_1 = q_k$.
It is an \emph{elementary cycle} if
$p_1 = q_k$ and $p_1, \ldots, p_k$ are pairwise distinct.

\begin{fact}
  \label{fact:iteration-cycles-of-instantiation}
  Consider a $d$-BVASS $\bvass = (Q, \Delta)$ and finite set $F \subseteq \ReachSet{\bvass}$.
  Let $q \in Q$, $\vec{x} \in \setN^d$ and
  let $\theta$ be an elementary cycle of $\instantiate{\bvass}{F}$ with displacement $\vec{v}$ and with
  start (and end) $q$.
  If $\vec{x}\geq (c,\ldots,c)$ where $c = |Q| \max_{i\in\{1,\ldots,d\}} \max_{(S, \vec{a}, q) \in \Delta} -\vec{a}(i)$ then $q(\vec{x}) \step[*] q(\vec{x} + \vec{v})$.
\end{fact}

\section{Reachability Set Computation for $2$-BVASS}\label{sec:algorithm}
We present in this section an algorithm to compute the reachability set for $2$-BVASS.
More precisely,
given a $2$-BVASS $\bvass$,
our algorithm returns a finite exploration of $\bvass$ that is both sound and complete.
We start by defining what we mean by sound and complete exploration.

\begin{definition}
  \label{def:exploration}
  An \emph{exploration} of a $2$-BVASS $\bvass = (Q, \Delta)$ is
  a node-labeled acyclic graph $\explo = (N, \rightarrow, \lambda)$ such that
  \begin{enumerate}
  \item
    \label{explo:wf}
    the edge relation $\rightarrow$ is well-founded,
    i.e.,
    there is no infinite sequence $n_0, n_1, \ldots$ of nodes in $N$ such that
    $n_{i+1} \rightarrow n_i$ for all $i \in \setN$,
    and
  \item
    \label{explo:label}
    each node $n \in N$ is labeled with
    $\lambda(n) = (\vec{a}_n, q_n, \vec{z}_n, \vec{P}_n)$
    where
    $\vec{a}_n \in \setZ^2$,
    $q_n \in Q$,
    $\vec{z}_n \in \setN^2$, and
    $\vec{P}_n$ is a periodic subset of $\setN^2$.
  \end{enumerate}
\end{definition}

Intuitively,
the label $\lambda(n) = (\vec{a}_n, q_n, \vec{z}_n, \vec{P}_n)$ of a node $n$ provides,
firstly,
the displacement $\vec{a}_n$ of the transition rule used to create $n$
(this will be made clear later on and can be ignored for now),
and,
secondly,
the set of configurations $q_n (\vec{z}_n + \vec{P}_n)$ associated with the node $n$.
Recall that $\ReachSet{\bvass}$ denotes the reachability set of a $2$-BVASS $\bvass$.
Similarly,
we associate to an exploration $\explo = (N, \rightarrow, \lambda)$ of $\bvass$
the set of configurations
$\ReachSet{\explo} = \bigcup_{n \in N} q_n(\vec{z}_n + \vec{P}_n)$.
We say that $\explo$ is \emph{sound} when $\ReachSet{\explo} \subseteq \ReachSet{\bvass}$ and
that it is \emph{complete} when $\ReachSet{\explo} \supseteq \ReachSet{\bvass}$.
A node $n \in N$ is called \emph{redundant} if
there exists $s \in N$ verifying $s \xrightarrow{+} n$ and $q_n(\vec{z}_n + \vec{P}_n) \subseteq q_s(\vec{z}_s + \vec{P}_s)$.
We say that $\explo$ is \emph{non-redundant} when every redundant node is a leaf.

\begin{example}
  \label{exa:exploration}
  The node-labeled acyclic graph depicted on the right-hand side of \cref{fig:BVASS}
  is an exploration of the $2$-BVASS $\bvassexample$ depicted on the left-hand side
  (see also \cref{exa:BVASS}).
  For instance,
  the set of configurations associated with the node $n_4$ is
  $\{q(2 + 2k, 4 + 5k) \mid k \in \setN\}$.
  The first component $\vec{a}_n$ of $\lambda(n)$ is omitted in the figure
  to reduce clutter.
  As mentioned above,
  the vectors $\vec{a}_n$ can be ignored for now,
  see \cref{exa:algorithmic-exploration} for actual values.
  \lipicsEnd
\end{example}

As in Hopcroft and Pansiot's algorithm for
classical $2$-VASS~\cite{DBLP:journals/tcs/HopcroftP79},
a crucial ingredient of our algorithm is the \emph{acceleration} of cycles.
The purpose of cycle acceleration is to make the periodic sets $\vec{P}_n$ grow.
We will utilize three kinds of cycles.
Consider an exploration $\explo = (N, \rightarrow, \lambda)$ of a $2$-BVASS $\bvass = (Q, \Delta)$.
We associate to each node $n \in N$ the $2$-VASS $\mathcal{V}_n$ defined as
the instantiation $\instantiate{\bvass}{F}$ of $\bvass$ with the finite set of configurations
$F = \{q_s(\vec{z}_s) \mid s \in N, s \xrightarrow{+} n\}$.
In particular,
if $n$ is a source then
$\mathcal{V}_n = \instantiate{\bvass}{\emptyset} = (Q, \Delta')$
where $\Delta' = \{(S, \vec{a}, q) \in \Delta \mid |S| = 1\}$
is the set of unary transition rules in $\Delta$.
We introduce three finite subsets of $\setZ^2$,
namely $\vec{E}_n$, $\vec{C}_n$ and $\overline{\vec{C}}_n$,
that correspond to the three kinds of cycles mentioned above.
Let us define
$c = |Q| \max_{i \in \{1, 2\}} \max_{(S, \vec{a}, q) \in \Delta} -\vec{a}(i)$.
We call $c$ the \emph{constant of iteration} of $\bvass$. Observe that $c$ is the same constant as the one in \cref{fact:iteration-cycles-of-instantiation}.
\begin{itemize}
\item
  $\vec{E}_n$ is the set of vectors $\vec{v} \in \setZ^2$
  such that there exist an elementary cycle $\theta$ of $\mathcal{V}_n$ with displacement $\vec{v}$ and
  a node $s \in \anc{n}$ verifying
  $q_s$ is the start of $\theta$, $\vec{z}_s \geq (c, c)$ and
  $s \neq n$ implies $\vec{v} \geq (0, 0)$.
\item
  $\vec{C}_n$ is the set of vectors $\vec{v} \in \setZ^2$
  such that there exists a node $s \in \anc{n}$ verifying
  $q_s = q_n$, $\vec{v} = \vec{z}_n - \vec{z}_s$,
  $\vec{z}_n \not \geq (c, c)$ and $\vec{z}_s \not \geq (c, c)$.
\item
  $\overline{\vec{C}}_n$ is the set of vectors $\vec{v} \in \setZ^2$
  such that there exists a node $s \in \anc{n}$ verifying
  $q_s = q_n$ and $\vec{v} = \vec{z}_n - \vec{z}_s$.
\end{itemize}
Note that $\vec{C}_n \subseteq \overline{\vec{C}}_n$.
We also introduce the finite set $\vec{I}_n$ defined by
$\vec{I}_n = \vec{E}_n \cup \overline{\vec{C}}_n$ if $\sum_{m \rightarrow n} \vec{P}_m = \{(0, 0)\}$ and
$\vec{I}_n = \vec{E}_n \cup \vec{C}_n$ otherwise.
Vectors in $\vec{E}_n$, $\vec{C}_n$, and $\vec{I}_n$ are respectively called
\emph{$n$-elementary}, \emph{$n$-consecutive} and \emph{$n$-iterable}.

\begin{example}
  \label{exa:iterable-vectors}
  Let us continue \cref{exa:exploration}.
  The constant of iteration of $\bvassexample$ is $c = 4$.
  By definition,
  $\vec{I}_{n_i} = \vec{E}_{n_i} \cup \overline{\vec{C}}_{n_i}$ for $i \in \{0, 1, 2, 6\}$, and
  $\vec{I}_{n_i} = \vec{E}_{n_i} \cup \vec{C}_{n_i}$ for $i \in \{3, 4, 5\}$.
  We first discuss $n$-consecutive vectors.
  We have $\vec{C}_{n_0} = \overline{\vec{C}}_{n_0} = \{(0, 0)\}$
  since $n_0$ has no ancestor except itself.
  The set $\vec{C}_{n_1}$ is empty because
  the first component of $\vec{z}_{n_1} = (3, 4)$ is strictly below $c$,
  hence,
  $\vec{z}_{n_1} \not\geq (c, c)$.
  Similarly,
  $\vec{C}_{n_2} = \vec{C}_{n_3} = \emptyset$.
  It is readily seen that
  $\overline{\vec{C}}_{n_1} = \overline{\vec{C}}_{n_2} = \{(0, 0)\}$ and that
  $\overline{\vec{C}}_{n_3} = \{(0, 0), (-1, 0)\}$.
  We now discuss $n$-elementary vectors.
  The $2$-VASS $\mathcal{V}_{n_0} = \instantiate{\bvassexample}{\emptyset}$ contains exactly
  one elementary cycle (up to rotation),
  namely\footnote{%
    We use the notation $p \xhookrightarrow{\vec{a}} q$ to
    denote unary transition rules $(p, \vec{a}, q)$ of a $d$-VASS.
  }
  $\theta_0 = p \xhookrightarrow{(-1,0)} q \xhookrightarrow{(0,0)} s \xhookrightarrow{(0,0)} p$.
  It follows that $\vec{E}_{n_0} = \{(-1, 0)\}$.
  In addition to the cycle $\theta_0$,
  the $2$-VASS $\mathcal{V}_{n_1} = \instantiate{\bvassexample}{\{p(4, 4)\}}$
  also contains the elementary cycle
  $\theta_1 = q \xhookrightarrow{(4,4)} r \xhookrightarrow{(-1,1)} q$.
  Still,
  the set $\vec{E}_{n_1}$ is empty,
  because none of these two cycles contributes to $\vec{E}_{n_1}$.
  Indeed,
  even though the elementary cycle $\theta_0$ contains the states $q_{n_0}$ and $q_{n_1}$,
  its displacement $\vec{v}_0 = (-1, 0)$ is not in $\vec{E}_{n_1}$ because
  $\vec{z}_{n_1} \not\geq (c, c)$ and $\vec{v}_0 \not\geq (0, 0)$.
  Analogously,
  the displacement of $\theta_1$ is not in $\vec{E}_{n_1}$ because $\vec{z}_{n_1} \not\geq (c, c)$.
  The elementary cycles of $\mathcal{V}_{n_2}$ are $\theta_0$, $\theta_1$ and $\theta_2$,
  where
  $\theta_2 = p \xhookrightarrow{(3,4)} r \xhookrightarrow{(-1,1)} q \xhookrightarrow{(0,0)} s \xhookrightarrow{(0,0)} p$.
  The displacement $\vec{v}_2 = (2, 5)$ of $\theta_2$ is in $\vec{E}_{n_2}$
  since $\theta_2$ contains the state $q_{n_0}$, $\vec{z}_{n_0} \geq (c, c)$ and $\vec{v}_2 \geq (0, 0)$.
  We get that $\vec{E}_{n_2} = \{(2, 5)\}$.
  Last,
  we observe that $\mathcal{V}_{n_2} = \mathcal{V}_{n_3}$ since the state $q_{n_2}$ is not part of any branching transition.
  So $\mathcal{V}_{n_3}$ has the same elementary cycles as $\mathcal{V}_{n_2}$, and
  we get that $\vec{E}_{n_3} = \vec{E}_{n_2} = \{(2, 5)\}$.
  Indeed,
  even though the elementary cycle $\theta_0$ contains the state $q_{n_3}$,
  its displacement $\vec{v}_0 = (-1, 0)$ is not in $\vec{E}_{n_3}$ because
  $\vec{z}_{n_3} \not\geq (c, c)$ and $\vec{v}_0 \not\geq (0, 0)$.
  \lipicsEnd
\end{example}

\smallskip

Given a set $\vec{I}\subseteq \setZ^2$ and a periodic set $\vec{P}\subseteq\setN^2$, we introduce the periodic set $\accel{\vec{I}}(\vec{P})$ defined as the set of vectors $\vec{p}+\vec{v}_1+\cdots+\vec{v}_k$ where $\vec{p}\in\vec{P}$, $k\in\setN$, and $\vec{v}_1,\ldots,\vec{v}_k$ are vectors in $\vec{I}$ such that $\vec{p}+\vec{v}_1+\cdots+\vec{v}_\ell\geq (0, 0)$ for every $\ell\in\{1,\ldots,k\}$. 
\begin{restatable}{lemma}{lemAccel}\label{lem:accel-finitely-generated-dim-2-bis}
   For every finite sets $\vec{G}\subseteq \setN^2$ and $\vec{I}\subseteq \setZ^2$, we can effectively compute a finite set $\vec{H}\subseteq \setN^2$ such that $\accel{\vec{I}}(\per{\vec{G}})=\per{\vec{H}}$.
\end{restatable}
\begin{proof}[Proof sketch.]
  First of all, notice that \cref{thm:hp79} is not sufficient for proving that result since there exist semilinear periodic sets, like $\{(0,0)\}\cup ((1,1)+\setN^2)$ that are not finitely-generated. From \cref{cor:HPcone}, we deduce that $\accel{\vec{I}}(\per{\vec{G}})$ is finitely-generated periodic set if $\con{\accel{\vec{I}}(\per{\vec{G}})}$ is a finitely-generated cone. We observe that this cone is spanned by $\vec{G}\cup (\vec{I}\cap\setN^2)\cup\vec{U}$ where $\vec{U}$ is a set of axis, i.e. a subset of $\{(1,0),(0,1)\}$. It follows that there exists a finite set $\vec{H}\subseteq \setN^2$ such that $\accel{\vec{I}}(\per{\vec{G}})=\per{\vec{H}}$. Finally, with a step-by-step algorithm computing increasing finite subsets of $\accel{\vec{I}}(\per{\vec{G}})$ we eventually reach a set $\vec{H}$ satisfying the lemma.
\end{proof}

\begin{algorithm}[t]
  \DontPrintSemicolon
  \Input{%
    A $2$-BVASS $\bvass = (Q, \Delta)$.
  }
  \Output{%
    A sound and complete finite exploration of $\bvass$.
  }
  $(N, {\rightarrow}, \lambda, \Redundant, \WorkList) := (\emptyset, \emptyset, \emptyset, \emptyset, \emptyset)$\;
  \label{line:empty-exploration}
  \ForEach{$(S, \vec{a}, q) \in \Delta$ with $S = \emptyset$ and $\vec{a} \geq (0, 0)$}{
    \label{line:init-start}
    create a new node $n$ (with $n \not\in N$)\;
    $N := N \cup \{n\}$\;
    $\lambda(n) := (\vec{a}, q, \vec{a}, \{(0, 0)\})$\;
    \label{line:lambda-init}
    $\WorkList := \WorkList \cup \{n\}$\;
    \label{line:init-end}
  }
  \While{$\WorkList \neq \emptyset$}{
    \label{line:while-start}
    let $n$ be a node in $\WorkList$\;
    $\WorkList := \WorkList \setminus \{n\}$\;
    $\vec{P}_n := \accel{\vec{I}_n}(\vec{P}_n)$\;
    \label{line:acceleration}
    \If{there exists $s \in N$ verifying $s \xrightarrow{+} n$ and $q_n(\vec{z}_n + \vec{P}_n) \subseteq q_s(\vec{z}_s + \vec{P}_s)$}{
      \label{line:cover-check}
      $\Redundant := \Redundant \cup \{n\}$\;
      \label{line:redundant-node}
    }
    \Else{
      \ForEach{$M \subseteq N$ with $n \in M$ and $M \cap (\Redundant \cup \WorkList) = \emptyset$}{
        \label{line:expansion-start}
        \ForEach{$(S, \vec{a}, q) \in \Delta$ with $S = \multiset{q_m \mid m \in M}$}{
          $(\vec{z}', \vec{P}') := (\sum_{m \in M} \vec{z}_m, \sum_{m \in M} \vec{P}_m)$\;
          \label{line:z'-P'}
          let $\vec{B}$ be a finite subset of $\setN^2$ such that $(\vec{B} + \vec{P}') = (\vec{a} + \vec{z}' + \vec{P}') \cap \setN^2$\;
          \label{line:basis-computation}
          \ForEach{$\vec{b} \in \vec{B}$}{
            \label{line:children-start}
            create a new node $n'$ (with $n' \not\in N$)\;
            $(N, {\rightarrow}) := (N \cup \{n'\}, {\rightarrow} \cup \{(m, n') \mid m \in M\})$\;
            \label{line:child-add-to-N}
            $\lambda(n') := (\vec{a}, q, \vec{b}, \vec{P}')$\;
            $\WorkList := \WorkList \cup \{n'\}$\;
            \label{line:child-add-to-W}
            \label{line:children-end}
          }
        }
      }
      \label{line:expansion-end}
    }
    \label{line:while-end}
  }
  \Return{$(N, {\rightarrow}, \lambda)$}\;
  \label{line:return}
  \caption{$\mathtt{Explore}(\bvass)$}
  \label{algo:main}
\end{algorithm}

\smallskip

Our algorithm, dubbed $\mathtt{Explore}$, is defined in \cref{algo:main}.
We use an abstract pseudocode to simplify the presentation.
This raises implementability issues that are addressed in \cref{rem:algo-implementability-issues}.
Given a $2$-BVASS $\bvass = (Q, \Delta)$ as input,
$\mathtt{Explore}(\bvass)$ iteratively computes an exploration of $\bvass$ and then returns it.
This exploration is maintained in the variables $N$, $\rightarrow$ and $\lambda$, and
is initially empty (see line~\ref{line:empty-exploration}).
As in \cref{def:exploration},
we let $q_n$, $\vec{z}_n$ and $\vec{P}_n$ denote
the second, third and fourth components of $\lambda(n)$.
The set of redundant nodes of the exploration is tracked in the variable $\Redundant$.
The set of unprocessed nodes, called the \emph{worklist}, is maintained in the variable $\WorkList$.
Both variables $\Redundant$ and $\WorkList$ remain disjoint subsets of $N$ during the execution of the algorithm.
For each initial configuration $q(\vec{a})$ of $\bvass$,
a new node is created and put in the worklist
(see lines~\ref{line:init-start}--\ref{line:init-end}).
After this initialization phase,
$\mathtt{Explore}(\bvass)$
repeatedly selects a node from the worklist and processes it,
as long as the worklist is non-empty
(see lines~\ref{line:while-start}--\ref{line:while-end}).
The processing of a node $n$ consists in four steps.
First,
the node $n$ is removed from the worklist
(so $n$ is considered processed afterwards).
Second,
the periodic set $\vec{P}_n$ is enlarged using the set $\vec{I}_n$ of $n$-iterable vectors
(see line~\ref{line:acceleration}).
This is the cycle \emph{acceleration} step.
Note that the set $\vec{I}_n$ is finite and
implicitly depends on the constant of iteration of $\bvass$ and on
the current exploration $(N, \rightarrow, \lambda)$ of $\bvass$.
The assignment at line~\ref{line:acceleration} actually means that
the label $\lambda(n)$ of the node $n$ is modified
(in fact,
only the fourth component of the label is modified).
Third,
the algorithm tests whether the node $n$ is ``covered'' by one of its ancestors
(see line~\ref{line:cover-check}).
If that is the case then $n$ is added to the set $\Redundant$ of redundant nodes and
the processing of $n$ stops.
Otherwise,
as a fourth step,
the node $n$ is expanded
(see lines~\ref{line:expansion-start}--\ref{line:expansion-end}),
meaning that
for each set of processed and non-redundant nodes $M$ containing $n$ and
for each transition rule $(S, \vec{a}, q)$ that applies to $M$,
finitely many children of $n$ are created and put in the worklist.
At first glance,
one might want to create a single child labeled with
$(\vec{a}, q, \vec{a} + \vec{z}', \vec{P}')$.
This would be fine if $\vec{a} + \vec{z}' \geq (0, 0)$,
however that is not the case in general.
This is the reason why we need the finite set $\vec{B}$ at line~\ref{line:basis-computation}
and the \textbf{foreach}-loop at lines~\ref{line:children-start}--\ref{line:children-end}.
The existence of such a finite set is explained in \cref{rem:algo-implementability-issues}.
The observant reader will notice that the modification of the variables $N$ and $\WorkList$
at lines~\ref{line:child-add-to-N} and~\ref{line:child-add-to-W}
has no impact on the \textbf{foreach}-loop iteration at line~\ref{line:expansion-start},
since $N \setminus \WorkList$ remains constant.
When the worklist becomes empty,
the constructed exploration is returned at line~\ref{line:return}.
The reachability set of $\bvass$ is then easily obtained from this exploration,
provided that it is sound and complete.

\begin{remark}
  \label{rem:algo-implementability-issues}
  The abstract pseudocode used in \cref{algo:main} raises some implementability issues.
  The first issue is that
  we liberally use periodic sets in the pseudocode,
  but these periodic sets need to admit a finite and computable representation.
  To address this issue,
  we observe that each periodic set defined in \cref{algo:main} is finitely-generated and
  admits a computable finite spanning set.
  The only line where this property is non-trivial is line~\ref{line:acceleration}, but \cref{lem:accel-finitely-generated-dim-2-bis} provides the result.
  The second issue is the existence and computation of the finite set $\vec{B}$ at line~\ref{line:basis-computation}.
  To address this issue,
  we recall the following well-known fact
  (see for instance \cite[Lemma~1.1]{DBLP:journals/tcs/HopcroftP79}).
  Given a vector $\vec{v} \in \setZ^d$ and
  a finite subset $\vec{A} = \{\vec{a}_1, \ldots, \vec{a}_k\}$ of $\setN^d$,
  the set $(\vec{v} + \per{\vec{A}}) \cap \setN^d$ is equal to
  $(\vec{B} + \per{\vec{A}})$ where
  $\vec{B} \subseteq \setN^d$ is the finite set of vectors
  $\vec{v} + \alpha_1\vec{a}_1 + \cdots + \alpha_k\vec{a}_k$ such that
  $(\alpha_1, \ldots, \alpha_k)$ is a minimal vector in $\setN^k$ satisfying
  $\vec{v} + \alpha_1\vec{a}_1 + \cdots + \alpha_k\vec{a}_k \geq \vec{0}$.
\end{remark}

\begin{example}
  \label{exa:algo-execution}
  To illustrate \cref{algo:main},
  we apply it on the $2$-BVASS $\bvassexample$ from \cref{exa:BVASS}.
  The resulting exploration is depicted in~\cref{fig:BVASS}.
  There is only one initial configuration,
  so the exploration has only one source $n_0$.
  The nodes are added to ($+$) and removed from ($-$) the worklist in the order
  $+n_0$, $-n_0$, $+n_1$, $-n_1$, $+n_2$, $+n_6$, $-n_2$, $+n_3$, $-n_3$, $+n_4$, $+n_5$.
  So the nodes $n_4, n_5$ and $n_6$ are still in the worklist.
  The labels $q_n$ and $\vec{z}_n$ are straightforward.
  For the periodic sets $\vec{P}_n$,
  we first recall from \cref{exa:iterable-vectors} that
  $\vec{I}_{n_0} = \{(0, 0), (-1, 0)\}$,
  $\vec{I}_{n_1} = \{(0, 0)\}$,
  $\vec{I}_{n_2} = \{(0, 0), (2, 5)\}$ and
  $\vec{I}_{n_3} = \{(2, 5)\}$.
  The periodic set $\vec{P}_{n_0}$ is $\accel{\vec{I}_{n_0}}(\{(0, 0)\}) = \{(0, 0)\}$.
  The node $n_1$ also has $\{(0, 0)\}$ as periodic set since
  $\vec{I}_{n_1} = \{(0, 0)\}$.
  The periodic set $\vec{P}_{n_2}$ is $\accel{\vec{I}_{n_2}}(\{(0, 0)\}) = \per{\{(2,5)\}}$.
  Similarly,
  $\vec{P}_{n_3} = \accel{\vec{I}_{n_3}}(\per{\{(2,5)\}}) = \per{\{(2,5)\}}$.
  \lipicsEnd
\end{example}

\section{Soundness and Completeness of Algorithmic Explorations}\label{sec:correctness}
This section is devoted to the partial correctness of our algorithm $\mathtt{Explore}$ (see \cref{algo:main}).
Its termination is much more involved and will be the subject of the next sections.
We first refine the notion of exploration to account for the behavior of our algorithm.

\begin{definition}
  \label{def:algorithmic-exploration}
  An exploration $\explo = (N, \rightarrow, \lambda)$ of a $2$-BVASS $\bvass = (Q, \Delta)$
  is \emph{algorithmic} if it satisfies,
  for every node $n \in N$,
  the three following conditions:
  \begin{enumerate}
  \item
    \label{explo:v+q}
    the multiset $S = \multiset{q_m \mid m \rightarrow n}$ is a set and verifies
    $(S, \vec{a}_n, q_n) \in \Delta$,
  \item
    \label{explo:z}
    the vector $\vec{z}_n$ is in $\vec{a}_n + \sum_{m \rightarrow n} (\vec{z}_m + \vec{P}_m)$, and
  \item
    \label{explo:P}
    the periodic set $\vec{P}_n$ verifies
    $\vec{P}_n = \accel{\vec{I}_n}(\sum_{m \rightarrow n} \vec{P}_m)$.
  \end{enumerate}
\end{definition}

Intuitively,
Conditions~\ref{explo:v+q} and~\ref{explo:z} ensure that the exploration
conforms to the semantics of $2$-BVASS.
As hinted before,
the vector $\vec{a}_n$ is the displacement of the transition rule leading to $n$
(see Condition~\ref{explo:v+q}).
Notice these two conditions entail in particular that
$q_n(\vec{z}_n)$ is an initial configuration for every source $n$.
Condition~\ref{explo:P} corresponds to the previously-mentioned cycle acceleration step
(see line~\ref{line:acceleration} of \cref{algo:main}).

\begin{example}
  \label{exa:algorithmic-exploration}
  Let us get back to the exploration of \cref{exa:exploration} and
  give the actual values of the vectors $\vec{a}_n$.
  We take
  $\vec{a}_{n_0} = (4, 4)$,
  $\vec{a}_{n_1} = \vec{a}_{n_4} = (-1, 0)$ and
  $\vec{a}_{n_2} = \vec{a}_{n_3} = \vec{a}_{n_5} = \vec{a}_{n_6} = (0, 0)$.
  The restriction to $\{n_0, n_1, n_2, n_3\}$ of the resulting exploration is algorithmic.
  \lipicsEnd
\end{example}

\begin{remark}
  \label{rem:finite-anc}
  For every algorithmic exploration $\explo = (N, \rightarrow, \lambda)$ and every node $n \in N$,
  the set $\anc{n}$ is finite.
  Indeed,
  by Condition~\ref{explo:v+q},
  every node $n \in N$ has finite in-degree (i.e., the set of nodes $m \in N$ such that $m \rightarrow n$ is finite).
  As $\rightarrow$ is well-founded by Condition~\ref{explo:wf} of \cref{def:exploration},
  it follows from König's Lemma that $\anc{n}$ is finite for every $n \in N$.
\end{remark}

In the rest of this section,
we show,
firstly,
that every algorithmic exploration is sound,
and secondly,
that our algorithm constructs explorations that are algorithmic and complete.

\subsection{Soundness of algorithmic explorations}
\label{subsec:soundness-algorithmic-explorations}

The main difficulty to establish the partial correctness of our algorithm $\mathtt{Explore}$
comes from the cycle acceleration step
(see line~\ref{line:acceleration} of \cref{algo:main}),
which translates to Condition~\ref{explo:P} of \cref{def:algorithmic-exploration}.
Contrary to usual cycle acceleration techniques,
our cycle acceleration step is ``retroactive'' since
the set $\vec{I}_n$ used to accelerate a node $n$
accounts for elementary cycles of $\mathcal{V}_n$ that
apply to an ancestor $s \xrightarrow{+} n$.
Thus, to show that a given algorithmic exploration $\explo$ is sound,
we transform $\explo$ into an alternative exploration $\explobis$
whose periodic sets $\vec{Q}_n$ are closed
in the sense that each $\vec{Q}_n$ already accounts for all potential cycles applicable to $n$, and
we show that $\explobis$ is sound.
Let us make these ideas more precise.

\smallskip

Consider a $2$-BVASS $\bvass = (Q, \Delta)$.
For every configuration $q(\vec{x})$ of $\bvass$,
we define the function $\clo{q}{\vec{x}} : \pow{\setN^2} \rightarrow \pow{\setN^2}$ by
$
\clo{q}{\vec{x}}(\vec{P})
=
\{\vec{w} \in \setN^2 \mid \exists \vec{u} \in \vec{P} : q(\vec{x} + \vec{u}) \step[*] q(\vec{x} + \vec{w})\}
$.
Observe that $\clo{q}{\vec{x}}$ is an upper closure operator\footnote{%
  An \emph{upper closure operator} on a partially-ordered set $(S, \leq)$ is
  any function $f : S \rightarrow S$ that is
  $\leq$-nondecreasing ($x \leq y$ implies $f(x) \leq f(y)$),
  extensive ($x \leq f(x)$), and
  idempotent ($f \circ f = f$).
} on the partially-ordered set $(\pow{\setN^2}, \subseteq)$.
In fact,
for every subset $\vec{P} \subseteq \setN^2$,
the set $\clo{q}{\vec{x}}(\vec{P})$ is
the $\subseteq$-greatest subset $\vec{Q} \subseteq \setN^2$ such that
$q(\vec{x} + \vec{Q}) \subseteq \relimg{\step[*]}{q(\vec{x} + \vec{P})}$.
Note also that $\clo{q}{\vec{x}}$ preserves periodicity,
meaning that
$\clo{q}{\vec{x}}(\vec{P})$ is periodic for every periodic subset $\vec{P} \subseteq \setN^2$.
The following technical lemma will allow us to relate
$\accel{\vec{I}_n}$ and $\clo{q_n}{\vec{z}_n}$.

\begin{restatable}{lemma}{lemCorrectnessAccelerationIterableVectors}
  \label{lem:correctness-of-acceleration-via-iterable-vectors}
  For every configuration $q(\vec{x})$ of $\bvass$ and periodic subset $\vec{P}$ of $\setN^2$,
  it holds that
  $\accel{\vec{J}}(\vec{P}) = \clo{q}{\vec{x}}(\vec{P})$ where
  $\vec{J}$ is the set of vectors $\vec{v} \in \setZ^2$ such that
  there exists $\vec{y} \in \setN^2$ and $\vec{u} \in \vec{P}$ verifying
  $(\vec{v} = \vec{y} - \vec{x} \text{ and } q(\vec{x} + \vec{u}) \step[*] q(\vec{y}))$
  or
  $(\vec{v} = \vec{x} - \vec{y} \text{ and } q(\vec{y} + \vec{u}) \step[*] q(\vec{x}))$.
\end{restatable}

\begin{restatable}{lemma}{lemSoundnessAlgorithmicExplorations}
  \label{lem:soundness-of-algorithmic-explorations}
  Every algorithmic exploration of a $2$-BVASS is sound.
\end{restatable}
\begin{proof}[Proof sketch.]
  Let $\explo = (N, \rightarrow, \lambda)$ be an algorithmic exploration of $2$-BVASS $\bvass$.
  We introduce the family $\left(\vec{Q}_n\right)_{n \in N}$ of subsets of $\setN^2$
  defined,
  by well-founded recursion over $\rightarrow$,
  by
  $
  \vec{Q}_n = \clo{q_n}{\vec{z}_n}(\sum_{m \rightarrow n} \vec{Q}_m)
  $
  for every node $n \in N$.
  Observe that each $\vec{Q}_n$ is periodic and
  that $m \rightarrow n$ implies $\vec{Q}_m \subseteq \vec{Q}_n$.
  We show by well-founded induction over $\rightarrow$ that,
  for all $n \in N$,
  we have
  $\vec{P}_n \subseteq \vec{Q}_n$ and $q_n(\vec{z}_n + \vec{Q}_n) \subseteq \ReachSet{\bvass}$.
  Let $n \in N$ and
  assume that
  $\vec{P}_s \subseteq \vec{Q}_s$ and $q_s(\vec{z}_s + \vec{Q}_s) \subseteq \ReachSet{\bvass}$
  for all $s \in N$ with $s \xrightarrow{+} n$.
  We derive from Conditions~\ref{explo:v+q} and~\ref{explo:z} of \cref{def:algorithmic-exploration} that
  $q_n(\vec{z}_n + \sum_{m \rightarrow n} \vec{Q}_m) \subseteq \ReachSet{\bvass}$.
  It follows that $q_n(\vec{z}_n + \vec{Q}_n) \subseteq \ReachSet{\bvass}$.
  It remains to show that $\vec{P}_n \subseteq \vec{Q}_n$.
  Let $\vec{J}_n$ denote the set $\vec{J}$ defined in
  \cref{lem:correctness-of-acceleration-via-iterable-vectors}
  with $q(\vec{x}) := q_n(\vec{z}_n)$ and $\vec{P} := \vec{Q}_n$.
  The crucial observation now is that $\vec{E}_n$ and $\overline{\vec{C}}_n$
  are both contained in $\vec{J}_n$,
  hence,
  $\vec{I}_n \subseteq \vec{J}_n$.
  We obtain from \cref{lem:correctness-of-acceleration-via-iterable-vectors}
  that
  $\accel{\vec{I}_n}(\vec{Q}_n) \subseteq \clo{q_n}{\vec{z}_n}(\vec{Q}_n) = \vec{Q}_n$
  since $\clo{q_n}{\vec{z}_n}$ is idempotent.
  Recall that $\vec{P}_m \subseteq \vec{Q}_m \subseteq \vec{Q}_n$
  for every $m \in N$ with $m \rightarrow n$.
  We derive from Condition~\ref{explo:P} of \cref{def:algorithmic-exploration} that
  $
  \vec{P}_n
  =
  \accel{\vec{I}_n}(\sum_{m \rightarrow n} \vec{P}_m)
  \subseteq
  \accel{\vec{I}_n}(\vec{Q}_n)
  \subseteq
  \vec{Q}_n
  $.
\end{proof}

\subsection{Partial correctness of $\mathtt{Explore}$}
\label{subsec:partial-correctness-Explore}

Consider an execution $\sigma$ of $\mathtt{Explore}(\bvass)$,
where $\bvass$ is a $2$-BVASS.
We let $\left(N^j, \rightarrow^j, \lambda^j, \Redundant^j, \WorkList^j\right)_{j \in J}$ denote
the successive values of the variables $N$, $\rightarrow$, $\lambda$, $\Redundant$ and $\WorkList$
at line~\ref{line:while-start} of \cref{algo:main},
just before the evaluation of the \textbf{while}-loop condition.
The index set $J$ is $\setN$ if the execution $\sigma$ does not terminate
(i.e., $\WorkList^j \neq \emptyset$ for all $j \in \setN$).
Otherwise,
$J = \{0, \ldots, \kappa\}$ where $\kappa \in \setN$ is the number of iterations of the \textbf{while}-loop
(i.e., $\WorkList^\kappa = \emptyset$ and $\WorkList^j \neq \emptyset$ for all $j \in \{0, \ldots, \kappa - 1\}$).
It is understood that,
in both cases,
the values of the variables after the \textbf{foreach}-loop at lines~\ref{line:init-start}--\ref{line:init-end}
are $N^0$, $\rightarrow^0$, $\lambda^0$, $\Redundant^0$ and $\WorkList^0$.
Moreover,
if $\sigma$ terminates then it returns
$(N^\kappa, \rightarrow^\kappa, \lambda^\kappa)$ at line~\ref{line:return}.
For every $j \in J$,
we let $\explo^j$ denote the labeled graph $(N^j, \rightarrow^j, \lambda^j)$, and
we let $\hat{\explo}^j$ denote the restriction of $\explo^j$ to
the set of processed nodes $\hat{N}^j = N^j \setminus \WorkList^j$.
The following lemma is easily derived from \cref{algo:main}
(the detailed proof is tedious but straightforward).

\begin{restatable}{lemma}{lemAlgoExplorations}
  \label{lem:algo-explorations}
  Both $\explo^j$ and $\hat{\explo}^j$ are non-redundant explorations of $\bvass$,
  for every $j \in J$.
  Moreover,
  $\hat{\explo}^j$ is algorithmic and it holds that
  $\post[\bvass](\ReachSet{\hat{\explo}^j}) \subseteq \ReachSet{\explo^j}$.
\end{restatable}

Notice that
$N^j \subseteq N^{j+1}$,
$\hat{N}^j \subset \hat{N}^{j+1}$ and
${\rightarrow}^j \subseteq {\rightarrow}^{j+1}$,
for every $j \in J$ with $(j+1) \in J$.
We introduce the ``limit'' values
$N^\sigma$, $\hat{N}^\sigma$ and $\rightarrow^\sigma$ of the corresponding sequences,
defined naturally by
$N^\sigma = \bigcup_{j \in J} N^j$,
$\hat{N}^\sigma = \bigcup_{j \in J} \hat{N}^j$ and
$\rightarrow^\sigma = \bigcup_{j \in J} \rightarrow^j$.
We also define
$\lambda^\sigma$ as the function with domain $N^\sigma$ that maps each $n \in N^\sigma$ to
the ultimate value of the sequence $\left(\lambda^j(n)\right)_{j \in J, n \in N^j}$.
The latter sequence is non-empty and ultimately constant.
We let $\explo^\sigma$ denote the labeled graph $(N^\sigma, \rightarrow^\sigma, \lambda^\sigma)$, and
we let $\hat{\explo}^\sigma$ denote the restriction of $\explo^\sigma$ to $\hat{N}^\sigma$.

\begin{restatable}{lemma}{lemAlgoPartialCorrectness}
  \label{lem:algo-partial-correctness}
  For every $2$-BVASS $\bvass$ and
  every execution $\sigma$ of $\mathtt{Explore}(\bvass)$,
  $\hat{\explo}^\sigma$ is a non-redundant and algorithmic exploration of $\bvass$.
  If $\sigma$ terminates then
  $\hat{\explo}^\sigma$ is finite and complete, and
  $\sigma$ returns $\hat{\explo}^\sigma$.
  Otherwise,
  $\hat{\explo}^\sigma$ is infinite.
\end{restatable}
\begin{proof}[Proof sketch.]
  Let $\sigma$ be an execution of $\mathtt{Explore}(\bvass)$,
  where $\bvass$ is a $2$-BVASS.
  It is routinely checked that,
  for every $j \in J$,
  \begin{enumerate}
  \item
    for every $m, n \in N^\sigma$,
    if $n \in N^j$ and $m \rightarrow^\sigma n$ then $m \in N^j$ and $m \rightarrow^j n$, and
  \item
    the restriction of $\hat{\explo}^\sigma$ to $\hat{N}^j$ coincides with $\hat{\explo}^j$.
  \end{enumerate}
  The first observation entails that the edge relation $\rightarrow^\sigma$ of $\explo^\sigma$ is well-founded,
  hence,
  both $\explo^\sigma$ and $\hat{\explo}^\sigma$ are explorations of $\bvass$.
  The second observation,
  combined with \cref{lem:algo-explorations},
  entails that $\hat{\explo}^\sigma$ is a non-redundant and algorithmic exploration of $\bvass$.
  If $\sigma$ terminates then it returns $\explo^\kappa$.
  Moreover,
  we have $\explo^\kappa = \hat{\explo}^\kappa = \hat{\explo}^\sigma$ in that case,
  hence,
  $\hat{\explo}^\sigma$ is finite and complete by \cref{lem:algo-explorations}.
  If $\sigma$ does not terminate then
  $\hat{N}^\sigma$ is infinite since
  $\hat{N}^j \subset \hat{N}^{j+1}$ for all $j \in \setN$.
\end{proof}

Partial correctness of our algorithm $\mathtt{Explore}$ follows from \cref{lem:algo-partial-correctness}.
To prove termination,
we will show in the next sections that
every non-redundant and algorithmic exploration of a $2$-BVASS is finite,
under an additional technical condition called spannability and discussed in next section.
Termination of our algorithm will then be ensured by \cref{lem:algo-partial-correctness}.

\section{Core Witnesses of Infinity for Spannable Graphs}\label{sec:branches}
\newcommand{\fst}{\operatorname{fst}}
\newcommand{\lst}{\operatorname{lst}}

A graph $(N,\rightarrow)$ is called a \emph{forest} if the edge relation $\rightarrow$ is well-founded (hence acyclic) and the set $\{m\in N \mid m\rightarrow n\}$ contains at most one node for every node $n\in N$. A forest is said to be \emph{finitely-branching} if the set of source nodes is finite and the set $\{n\in N \mid m\rightarrow n\}$ is finite for every node $m\in N$. We say that a graph $\mathcal{G}=(N,\rightarrow)$ is \emph{spannable} if there exists a subrelation $\rightarrow_{\mathcal{F}}$ of $\rightarrow$ such that $\mathcal{F}=(N,\rightarrow_{\mathcal{F}})$ is a finitely-branching forest. In that case $\mathcal{F}$ is called a \emph{spanning forest} of $\mathcal{G}$. We prove in appendix (see \cref{lem:caragoodforbranch}) that the exploration $\hat{\explo}^\sigma$ built by $\mathtt{Explore}(\bvass)$ (see \cref{subsec:partial-correctness-Explore}) is spannable.

\smallskip

A \emph{branch} $\beta$ of a forest $(N,\rightarrow_{\mathcal{F}})$ is an infinite sequence $\beta=(\beta_n)_{n\in\setN}$ of nodes such that $\beta_0$ is a source node of $\mathcal{F}$ and such that $\beta_{i}\rightarrow_{\mathcal{F}}\beta_{i+1}$ for every $i\geq 0$. Thanks to the Koenig's lemma, we know that any infinite finitely-branching forest admits a branch. Such a branch can be seen as a witness of infinity of $\mathcal{F}$. We extend this notion of witnesses to infinite spannable graphs in a non-trivial way. Naturally, denoting by $\mathcal{F}$ a spanning forest of $\mathcal{G}$, any branch of $\mathcal{F}$ is a kind of witness of infinity of $\mathcal{G}$. However such a witness depends somehow on the choice of $\mathcal{F}$ and does not take into account the structure of $\mathcal{G}$.

Our ``core'' witness of infinity of an infinite spannable graph is defined thanks to the notion of \emph{primary} and \emph{directed} graphs. A graph $\mathcal{G}=(N,\rightarrow)$ is said to be \emph{primary} if $N\setminus \des{n}$ is finite for every node $n\in N$. A graph $\mathcal{G}=(N,\rightarrow)$ is said to be \emph{directed} if for every $n,m\in N$, there exists $s\in N$ such that $n\xrightarrow{*}s$ and $m\xrightarrow{*}s$. The following lemma will be useful to extract from an infinite spannable graph a ``core'' subset of nodes that explain its infinity.
\begin{restatable}{lemma}{lemWellExtraction}\label{lem:well-extraction}
  For every infinite spannable graph $\mathcal{G}=(N,\rightarrow)$, there exists an infinite ancestor-closed set $X\subseteq N$ such that the restriction of $\mathcal{G}$ to $X$ is primary and directed.
\end{restatable}
\begin{proof}[Proof sketch.]
  Let $\mathcal{F}$ be a spanning forest of $\mathcal{G}$. We associate to each branch $\beta=(\beta_n)_{n\in\setN}$ of $\mathcal{F}$, the set $\anc{\beta}=\bigcup_{n\in\setN}\anc{\beta_n}$ of ancestors of $\beta$ with respect to the graph $\mathcal{G}$. We introduce the preorder (reflexive and transitive) $\sqsubseteq$ over the branches defined by $\alpha\sqsubseteq \beta$ if $\anc{\alpha}\subseteq \anc{\beta}$. A branch $\beta$ is said to be \emph{minimal} (for the relation $\sqsubseteq$) if for every branch $\alpha$ such that $\alpha\sqsubseteq \beta$, we have $\beta\sqsubseteq \alpha$. Notice that a branch $\beta$ is minimal for $\sqsubseteq$ if, and only if, $\anc{\beta}$ is minimal for the inclusion relation. In appendix, we prove that a minimal branch exists by contradiction. Intuitively, if there does not exist a minimal branch then any branch admits a strictly smaller one (we prove that we can select an eagerly smaller one). Since the number of sources of $\mathcal{F}$ is finite, and the set $\{n\in N \mid m\rightarrow n\}$ is finite for every node $m\in N$, we can extract from this infinite sequence of branches, a subsequence that ``converges'' to another branch. We prove that this branch is necessarily minimal providing a contradiction. It follows that there exists a minimal branch $\beta$. In appendix, we show that $X=\anc{\beta}$ satisfies the lemma.
\end{proof}

\begin{remark}\label{rem:primary2wqo}
  If $(N,\rightarrow)$ is a primary graph then $\xrightarrow{*}$ is a well-quasi-order (wqo). In fact, let us consider an infinite sequence $(n_i)_{i\in\setN}$ of nodes $n_i\in N$. Since $\mathcal{G}$ is primary, the set $N_0=N\setminus \des{n_0}$ is finite. If $n_i\in N_0$ for every $i\geq 1$ then there exists $i<j$ such that $n_i=n_j$ and in particular $n_i\xrightarrow{*}n_j$. Otherwise there exists $j\geq 1$ such that $n_j\not\in N_0$ and in that case $n_0\xrightarrow{*}n_j$. So, in any case, we have proved that there exists $i<j$ such that $n_i\xrightarrow{*}n_j$. Therefore $\xrightarrow{*}$ is a wqo.
\end{remark}

\section{Termination}\label{sec:termination}
The termination of \cref{algo:main} is obtained by contradiction. We assume that the algorithm is not terminating and from an infinite execution we derive an infinite exploration. Such an exploration is spannable, algorithmic and infinite. Thanks to \cref{lem:well-extraction} we can extract a sub-exploration $(N,\rightarrow,\lambda)$ that is also directed and primary. We prove that the sequence of periodic sets $\vec{P}_n$ for this exploration stabilizes. From~\cref{cor:HPcone} it is sufficient to prove that the sequence of cones $\con{\vec{P}_n}$ eventually stabilizes. Since the exploration is directed, it is sufficient to prove that the cone spanned by $\vec{P}_N=\bigcup_{n\in N}\vec{P}_n$ is finitely-generated.

\smallskip

This result is obtained by interpreting geometrically the acceleration step $\vec{P}_n:=\accel{\vec{I}_n}(\vec{P}_n)$ performed at line~\ref{line:acceleration}, as the \emph{so-called} $\vec{v}$-stabilization of $\con{\vec{P}_n}$ by vectors $\vec{v}\in\vec{I}_n$. More formally, the \emph{$\vec{v}$-stabilization} of a cone $\vec{C}\subseteq\setQ_{\geq 0}^2$ for some vector $\vec{v}\in\setZ^2$ is the cone $(\vec{C}+\setQ_{\geq 0}\vec{v})\cap\setQ_{\geq 0}^2$. A cone $\vec{C}\subseteq\setQ_{\geq 0}^2$ is said to be \emph{$\vec{v}$-stable} when it is equal to its $\vec{v}$-stabilization.

\begin{example}
  The cone $\vec{C}=\con{\{(1,2),(2,1)\}}$ is $(-1,-1)$-stable. It is not $(-3,-1)$-stable since its $(-3,-1)$-stabilization is $\vec{C}+\setQ_{\geq 0}(0,1)$.
  \lipicsEnd
\end{example}

The $\vec{v}$-stabilization of a cone $\vec{C}\subseteq \setQ_{\geq 0}^2$ spanned by a finite set of vectors $\vec{G}\subseteq \setN^2$ is the cone spanned by $\vec{G}\cup\{\vec{v}\}$ when $\vec{v}\geq (0, 0)$, and the cone spanned by $\vec{G}\cup\vec{U}$ where $\vec{U}$ is a subset of the set of \emph{axis} $\{(1,0),(0,1)\}$ when $\vec{v}\not\geq(0, 0)$. Since at some step of the algorithm, no new axis are added, it follows that the cone $\con{\vec{P}_n}$ is $\vec{v}$-stable for every vector $\vec{v}\in\vec{I}_n\setminus\setQ_{\geq 0}^2$. In fact a stronger stabilization property occurs: there exists a node $n_0$ such that for every $n,m\in\des{n_0}$, the cone $\con{\vec{P}_n}$ is $\vec{v}$-stable for every $\vec{v}\in\vec{I}_m\setminus\setQ_{\geq 0}^2$. In order to capture this special node $n_0$, we introduce in \cref{sec:modes} the notion of \emph{modes} of a cone and prove that modes of non-decreasing sequence of cones eventually stabilizes.

\smallskip

Thanks to the notion of modes, it will be clear that $\con{\vec{P}_n}$ is $\vec{v}$-stable for every vector $\vec{v}\in\vec{I}_n\setminus\setQ_{\geq 0}^2$ but also for every vector $\vec{v}\in\vec{I}_m\setminus\setQ_{\geq 0}^2$ where $m\in\des{n_0}$. This $\vec{v}$-stabilization with respect to vectors that can be discovered later by the algorithm will be useful with the decomposition of elementary vectors of $\vec{E}_n$ introduced in \cref{sec:elem-cons} as finite sums of elementary vectors $\vec{E}_m$ where $m$ ranges over a finite set independent of $n$.

\smallskip

Thanks to the notion of modes and the decomposition of elementary cycles, we prove in \cref{sec:finitegen} that the cones $\con{\vec{P}_N}$ is finitely-generated. We conclude our proof by contradiction in \cref{sec:wrap-up}.

\subsection{Modes}\label{sec:modes}

The $\vec{v}$-stabilization of a cone $\vec{C}\subseteq\setQ_{\geq 0}^2$ does not change the cone when $\vec{v}$ satisfies some conditions depending on the set of axis of $\vec{C}$ and a natural number $h\in\setN$ that provides a lower-bound on some negative components of $\vec{v}$. By \emph{axis} of a cone $\vec{C}\subseteq \setQ_{\geq 0}^2$ we mean a vector in $\vec{C}\cap\{(1,0),(0,1)\}$. The \emph{$h$-mode} of $\vec{C}$ is the set of vectors $\vec{v}\in\setZ^2\setminus \setN^2$ such that $\vec{C}$ is $\vec{v}$-stable, and such that the following two conditions hold:
\begin{itemize}
\item $\vec{v}(1)\geq -h$ if $(1,0)$ is not an axis of $\vec{C}$.
\item $\vec{v}(2)\geq -h$ if $(0,1)$ is not an axis of $\vec{C}$.
\end{itemize}

\begin{restatable}{lemma}{lemFiniteswithmodes}\label{lem:finiteswithmodes}
  The $h$-mode of any non-decreasing sequence of cones $\vec{C}_0\subseteq \vec{C}_1\subseteq \cdots \subseteq \setQ_{\geq 0}^2$ eventually stabilizes.
\end{restatable}
\begin{proof}[Proof sketch.]
  By considering a suffix of the sequence, by discarding the trivial cases, and by symmetry (on the set of axis), we can assume that $(1,0)\in\vec{C}_n\not\subseteq\setQ_{\geq 0}(1,0)$ for every $n\in \setN$. We observe that the $h$-mode $\vec{M}_n$ of $\vec{C}_n$ can be decomposed into $\vec{Q}_n\cup (\setN\times\{-h,\ldots,-1\})$ where $\vec{Q}_n=\{\vec{v}\in \vec{M}_n\mid \vec{v}(1)<0\wedge \vec{v}(2)\in\{-h,\ldots,0\}\}$. By observing that $(\vec{Q}_n)_{n\in\setN}$ is a non-increasing sequence of finite sets, we deduce that it eventually stabilizes.
\end{proof}

\subsection{Cycle Decomposition}\label{sec:elem-cons}

In this section, we prove that vectors in $\vec{E}_n$ can be decomposed as finite sums of vectors in $\vec{E}_m$, vectors in $\vec{P}_s$ and some consecutive vectors where $m$ and $s$ ranges over finite sets independent of $n$. The finiteness of those sets follows from the co-finiteness of $\des{n_0}$ and $\des{n_1}$ since the exploration is primary.
\begin{lemma}\label{lem:elem_cons}
  For every primary and directed algorithmic exploration $(N,\rightarrow,\lambda)$ and for every node $n_0\in N$ such that $\sum_{s\rightarrow n_0}\vec{P}_s\not=\{(0, 0)\}$, there exists a node $n_1\in\des{n_0}$ such that for every node $n\in\des{n_0}$, we have:
  \begin{equation*}
  \vec{E}_n\subseteq \sum_{m\in\des{n_0}\setminus (\des{n_1}\setminus\{n_1\})}\per{\vec{E}_m}+\sum_{s\not\in \des{n_0}}\vec{P}_s+\sum_{m\mid n_0\xrightarrow{*}m\xrightarrow{+}n}\per{\vec{C}_m}
  \end{equation*}
\end{lemma}

We define a node $n_1$ satisfying the previous lemma as follows. We fix a primary and directed exploration $(N,\rightarrow,\lambda)$ and a node $n_0\in N$ such that $\sum_{s\rightarrow n_0}\vec{P}_s\not=\{(0, 0)\}$. We introduce an abstraction function $\alpha:\setN^2\rightarrow \{0,\ldots,c\}^2$ defined by $\alpha(\vec{x})(i)=\vec{x}(i)$ if $\vec{x}(i)<c$ and $\alpha(\vec{x})(i)=c$ if $\vec{x}(i)\geq c$, for every $\vec{x}\in\setN^2$ and every $i\in\{1,2\}$, where $c$ is the constant of iteration of $\mathcal{B}$. We introduce a partial order $\sqsubseteq$ on the set of nodes $\des{n_0}$ defined by $s\sqsubseteq n$ if $s\xrightarrow{*}n$ and $q_s(\alpha(\vec{z}_s))=q_n(\alpha(\vec{z}_n))$. Since $\xrightarrow{*}$ is a well-quasi-order on the set of nodes, and the equality is a well-quasi-order on the finite set $Q\times\{0,\ldots,c\}^2$, the partial order $\sqsubseteq$ is also a well-quasi-order as the intersection of two well-quasi-orders. In particular the set $N_{\min}$ defined as the set of minimal nodes $n\in \des{n_0}$ for $\sqsubseteq$ is finite and for every $n\in \des{n_0}$, there exists $s\in N_{\min}$ such that $s\sqsubseteq n$. Since the exploration is primary, the set $N\setminus\des{n_0}$ is finite. Moreover, as the exploration is directed and $N_{\min}\cup (N\setminus\des{n_0})$ is finite, there exists a node $n_0'\in N$ such that $n\xrightarrow{*}n_0'$ for every $n\in N_{\min}\cup (N\setminus\des{n_0})$. Let $Q_c$ be the set of states $q\in Q$ such that there exists a node $n\in \des{n_0'}$ satisfying $\vec{z}_n\geq (c,c)$ and $q_n=q$. For each $q\in Q_c$, we pick such a node $n_q$. Since the set $\{n_q \mid q\in Q_c\}$ is finite and the exploration is directed, there exists $n_1\in\des{n_0'}$ such that $n_q\xrightarrow{*}n_1$ for every $q\in Q_c$. This node $n_1$ satisfies Lemma~\ref{lem:elem_cons}. Intuitively, this property is obtained by decomposing recursively either paths in the algorithmic exploration from a node $n\in\des{n_0}$ to a node $m$ such that $q_n=q_m$ following intermediate nodes $n$ satisfying $\vec{z}_n\geq (c,c)$ into elementary and consecutive cycles, and by replacing any elementary cycle using a transition coming from an instantiated node $n\in\des{n_0}$ by the elementary cycle obtained by using a node $s\in N_{\min}$ satisfying $s\sqsubseteq n$ rather than $n$ for the instantiation. Since the effect of the original cycle is equal to the sum of the effect of the new  one with the effect of a path from $s$ to $n$, by recursively decomposing that path we prove that $n_1$ satisfies Lemma~\ref{lem:elem_cons}. Such a decomposition and the proof that $n_1$ satisfies \cref{lem:elem_cons} are fully detailed in appendix.

\subsection{Finitely-generated Cone}\label{sec:finitegen}
In this section, we prove the following lemma.
\begin{lemma}\label{lem:Pfinite}
  For every primary and directed algorithmic exploration, the cone spanned by $\vec{P}_N=\sum_{n\in N}\vec{P}_n$ is finitely-generated.
\end{lemma}

Let $\mathcal{G}$ be a primary and directed algorithmic exploration and let $\vec{P}_N$ be the periodic set $\sum_{n\in N}\vec{P}_n$. We introduce the set of axis $\vec{U}=\con{\vec{P}_N}\cap\{(1,0),(0,1)\}$, the constant of iteration $c$, and the set $\vec{Z}=\{\vec{z}_n \mid n\in N\wedge \vec{z}_n\not\geq (c,c)\}$. If $\vec{U}=\{(1,0),(0,1)\}$ we are done since in that case $\con{\vec{P}_N}=\setQ_{\geq 0}^2$. So, we can assume that $\vec{U}$ contains at most one vector.

\begin{lemma}\label{lem:cprime}
   There exists $h\in\setN$ such that $\vec{Z}\subseteq \{0,\ldots,h\}^2+\vec{U}^*$.
\end{lemma}
\begin{proof}
  Recall from \cref{sec:branches} that $\xrightarrow{*}$ is a wqo. Given $q\in Q$, $d\in\{0,\ldots,c-1\}$, and $i\in \{1,2\}$, we introduce the set $N_{q,i,d}$ of nodes $n\in N$ such that $q_n=q$ and $\vec{z}_n(i)=d$. Since $\xrightarrow{*}$ is a wqo, the set $M_{q,i,d}$ of minimal elements of $N_{q,d,i}$ for this partial order is finite, and for every $n\in N_{q,i,d}$, there exists $m\in M_{q,i,d}$ such that $m\xrightarrow{*}n$. We pick $h\in\setN$ satisfying $h\geq \vec{z}_m(\bar{i})$ for every $q\in Q$, $i\in\{1,2\}$, $d\in \{0,\ldots,c-1\}$ and $m\in M_{q,i,d}$.
 
  Observe that for every $\vec{z}\in\vec{Z}$, there exists $n\in N$ such that $\vec{z}=\vec{z}_n\not\geq (c,c)$. It follows that there exists $i\in\{1,2\}$ such that $d=\vec{z}_n(i)$ is in $\{0,\ldots,c-1\}$. Let $q=q_n$ and notice that $n\in N_{q,i,d}$. It follows that there exists $m\in M_{q,i,d}$ such that $m\xrightarrow{*}n$. Observe that $\vec{z}_m(i)=d=\vec{z}_n(i)$ and $q_m=q=q_n$. Moreover $\vec{z}_m(\bar{i})\leq h$. If $\vec{z}_n(\bar{i})\leq h$ then $\vec{z}\in\{0,\ldots,h\}^2$. If $\vec{z}_n(\bar{i})>h$, then $\vec{z}_n-\vec{z}_m\in\setQ_{\geq 0}\vec{u}$ for some $\vec{u}\in\{(1,0),(0,1)\}$. It follows that $\vec{u}$ is in $\con{\vec{P}_n}$ since $\vec{z}_n-\vec{z}_m$ is a consecutive vector in $\vec{C}_n$. Hence it is also in $\vec{U}$. As $\vec{z}_m\in \{0,\ldots,h\}^2$ we deduce that $\vec{z}\in  \{0,\ldots,h\}^2+\vec{U}^*$. We have proved the lemma.
\end{proof}

In the sequel, $h$ denotes a fix natural number satisfying the previous lemma and $h\geq c$. We put $\vec{P}_n'=\sum_{s\rightarrow n}\vec{P}_s$ for every node $n\in N$. If $\vec{P}_n'=\{(0, 0)\}$ for every $n\in N$ then $\con{\vec{P}_N}=\{(0, 0)\}$ and we are done also in that case. So, we can assume that there exists a node $n\in N$ such that $\vec{P}_n'\not=\{(0, 0)\}$.

\smallskip

Let us prove that there exists a node $n_0$ such that $\vec{U}$ is the set of axis of $\con{\vec{P}_{n_0}}$ and $\vec{P}_{n_0}'\not=\{(0, 0)\}$. If $\vec{U}=\emptyset$, it is sufficient to consider a node $n_0\in N$ such that $\vec{P}_{n_0}'\not=\{(0, 0)\}$. If $\vec{U}\not=\emptyset$, it is sufficient to consider a node $n_0$ such that the unique vector of $\vec{U}$ is in $\con{\vec{P}_{n_0}'}$. By replacing $n_0$ by a descendant of $n_0$, thanks to \cref{lem:finiteswithmodes}, we can assume that the $h$-mode of $\con{\vec{P}_n}$ is equal to the $h$-mode of $\con{\vec{P}_{n_0}}$ for every $n\in\des{n_0}$.

\smallskip

In the sequel, given a set $\vec{X}\subseteq \setQ^2$, we introduce $\vec{X}^+= \vec{X}\cap\setQ_{\geq 0}^2$ and $\vec{X}^-=\vec{X}\setminus\setQ_{\geq 0}^2$.
\begin{restatable}{lemma}{lemConsFinite}\label{lem:cons-finite}
    The set $\con{\vec{U}}+\sum_{m\in \des{n_0}}\con{\vec{C}_m^+}$ is a finitely-generated cone.
\end{restatable}
\begin{proof}[Proof sketch.]
  We just observe that $\vec{C}_m\subseteq \vec{Z}-\vec{Z}$ and $\vec{Z}\subseteq \{0,\ldots,h\}^2+\vec{U}^*$.
\end{proof}

Now, let $n_1\in\des{n_0}$ satisfying \cref{lem:elem_cons}.
\begin{restatable}{lemma}{lemPformula}\label{lem:Pformula}
  The set $\con{\vec{P}_N}$ can be decomposed as follows:
  \begin{align*}
     \sum_{m\in\des{n_0}\setminus (\des{n_1}\setminus\{n_1\})}\con{\vec{E}_m^+}+\sum_{s\not\in \des{n_0}}\con{\vec{P}_s}
    +
    \con{\vec{U}}+
    \sum_{m\in \des{n_0}}\con{\vec{C}_m^+}
  \end{align*}
\end{restatable}
\begin{proof}[Proof sketch.]
  Let $\vec{K}$ be the cone given above. Clearly $\vec{E}_m^+$ where $m\in\des{n_0}\setminus (\des{n_1}\setminus\{n_1\})$, $\vec{P}_s$ where $s\not\in \des{n_0}$, $\vec{U}$ and $\vec{C}_m^+$ where $m\in \des{n_0}$ are included in $\con{\vec{P}_N}$. It follows that $\vec{K}\subseteq \con{\vec{P_N}}$. We prove the converse inclusion by induction on the well-foundedness of the relation $\rightarrow$, showing that $\vec{P}_n\subseteq \vec{K}$ for every $n\in N$. The crucial observation is the fact that even if in the right-hand side we discard vectors in $\vec{E}_m^-$ and $\vec{C}_m^-$, since the $h$-mode of $\con{\vec{P}_n}$ does not depend on the node $n\in\des{n_0}$, we can apply the decomposition of $\vec{E}_n$ given by \cref{lem:elem_cons} to prove the inclusion  $\vec{P}_n\subseteq \vec{K}$.
\end{proof}

 \cref{lem:cons-finite} and \cref{lem:Pformula} show that $\con{\vec{P_N}}$ is a finitely-generated cone.

\subsection{Wrap-Up}\label{sec:wrap-up}
In this section,
we prove the termination of our algorithm by assembling the results from the previous sections.
We start with the following observation.

\begin{restatable}{lemma}{lemZdiffStable}
  \label{lem:zdiff-stable}
  Let $\explo = (N, \rightarrow, \lambda)$ be an algorithmic exploration of a $2$-BVASS $\bvass$ and
  let $s, n \in N$.
  If $s \xrightarrow{*} n$, $q_s = q_n$ and $\vec{P}_s = \vec{P}_n$ then
  $\con{\vec{P}_n}$ is $(\vec{z}_n - \vec{z}_s)$-stable.
\end{restatable}
\begin{proof}[Proof sketch.]
   The decomposition mentioned in \cref{sec:elem-cons} provides a way to decompose a path from $s$ to $n$ in the algorithmic exploration into elementary cycles and consecutive cycles, with cutting points from intermediate nodes $m$ satisfying $\vec{z}_m\not\geq (c,c)$. With such a decomposition, we deduce that $\vec{z}_n-\vec{z}_s$ is a sum of vectors in $\vec{P}_n$ and vectors $\vec{v}\in\vec{I}_m$ where $m$ ranges over the intermediate nodes of the considered path. Since $\con{\vec{P}_m}$ is $\vec{v}$-stable for each $\vec{v}\in\vec{I}_m$ and $\con{\vec{P}_m}=\con{\vec{P}_n}$, we deduce that $\con{\vec{P}_n}$ is $\vec{v}$-stable for all those vectors $\vec{v}\in\vec{I}_m$. We deduce that $\con{\vec{P}_n}$ is $(\vec{z}_n - \vec{z}_s)$-stable.
\end{proof}

\begin{lemma}
  \label{lem:finiteness-of-nice-explorations}
  Every non-redundant, algorithmic, primary and directed exploration of a $2$-BVASS is finite.
\end{lemma}
\begin{proof}
  Let $\explo = (N, \rightarrow, \lambda)$ be a non-redundant, algorithmic,
  primary and directed exploration of a $2$-BVASS $\bvass = (Q, \Delta)$.
  By \cref{lem:Pfinite},
  the set $\vec{P} = \sum_{n\in N} \vec{P}_n$ is a finitely-generated periodic set.
  This entails that $\vec{P} = \sum_{n \in N'} \vec{P}_n$ for some finite set $N' \subseteq N$.
  As $\explo$ is directed,
  there exists a node $n_0 \in N$ such that $n' \xrightarrow{*} n_0$ for every $n' \in N'$.
  This entails that $\vec{P}_n = \vec{P}$ for all $n \in \des{n_0}$.
  Assume,
  by contradiction,
  that $N$ is infinite.
  Since $\explo$ is primary,
  the set $\des{n_0}$ is infinite and the binary relation $\xrightarrow{*}$ is a wqo on $N$.
  By Dickson's Lemma,
  the binary relation $\preceq$ on $N$ defined by
  by $m \preceq n$ if $q_m = q_n$ and $\vec{z}_m \leq \vec{z}_n$,
  is also a wqo on $N$.
  We derive that there exists an infinite sequence
  $n_1, n_2, n_3, \ldots$ of nodes in $\des{n_0}$ such that
  $n_1 \xrightarrow{+} n_2 \xrightarrow{+} n_3 \xrightarrow{+} \cdots$,
  $q_{n_1} = q_{n_2} = q_{n_3} = \cdots$, and
  $\vec{z}_{n_1} \leq \vec{z}_{n_2} \leq \vec{z}_{n_3} \leq \cdots$.
  Note that $\vec{P}_{n_i} = \vec{P}$ for all $i \geq 1$.
  We deduce from \cref{lem:zdiff-stable} that
  $\con{\vec{P}}$ is $(\vec{z}_{n_i} - \vec{z}_{n_1})$-stable for all $i \geq 1$.
  Since the vector $\vec{z}_{n_i} - \vec{z}_{n_1}$ is in $\setN^2$, we get that it is in $\con{\vec{P}}$.
  It follows that
  $\vec{z}_{n_i} \in \vec{z}_{n_1} + \con{\vec{P}}$ for all $i \geq 1$.
  So the set $\vec{B} = \{\vec{z}_{n_i} \mid i \geq 1\}$ is contained in
  $\vec{z}_{n_1} + \con{\vec{P}}$.
  We obtain from \cite{DBLP:journals/tcs/HopcroftP79} (see statement \cref{lem:HPcone}) that
  $\vec{B} + \vec{P} = \vec{B}' + \vec{P}$ for some finite subset $\vec{B}' \subseteq \vec{B}$.
  This entails that $\vec{z}_{n_j} \in (\vec{z}_{n_i} + \vec{P})$
  for some $i < j$.
  Observe that
  $q_{n_j}(\vec{z}_{n_j} + \vec{P}_{n_j}) \subseteq q_{n_i}(\vec{z}_{n_i} + \vec{P}_{n_i})$.
  Moreover, we have $n_i \xrightarrow{+} n_j$ since $i < j$.
  So the node $n_j$ is redundant,
  but $n_j$ is not a leaf since $n_j \xrightarrow{+} n_{j+1}$.
  This contradicts our assumption that $\explo$ is non-redundant.
\end{proof}

\begin{corollary}
  \label{cor:finiteness-of-nice-explorations}
  Every non-redundant, algorithmic and spannable exploration of a $2$-BVASS is finite.
\end{corollary}
\begin{proof}
  This corollary follows from \cref{lem:well-extraction}, \cref{lem:finiteness-of-nice-explorations} and
  the following observation.
  Given an exploration $\explo = (N, \rightarrow, \lambda)$ of a $2$-BVASS $\bvass$ and
  an ancestor-closed subset $X$ of $N$,
  the restriction of $\explo$ to $X$ is also an exploration of $\bvass$.
  Moreover,
  if $\explo$ is non-redundant (resp., algorithmic) then the restriction of $\explo$ to $X$
  is also non-redundant (resp., algorithmic).
\end{proof}

As indicated in \cref{sec:branches},
for every $2$-BVASS $\bvass$ and
every execution $\sigma$ of $\mathtt{Explore}(\bvass)$,
the constructed exploration $\hat{\explo}^\sigma$ is spannable.
We derive from
\cref{lem:soundness-of-algorithmic-explorations,lem:algo-partial-correctness,cor:finiteness-of-nice-explorations}
that every execution of $\mathtt{Explore}(\bvass)$ terminates and
returns a sound and complete finite exploration of $\bvass$.
The reachability set of $\bvass$ is then easily obtained from this exploration.
We obtain the following theorem.

\begin{theorem}
  For every $2$-BVASS $\bvass$,
  the reachability set $\ReachSet{\bvass}$ of $\bvass$ is semilinear and
  a presentation of $\ReachSet{\bvass}$ is computable from $\bvass$.
\end{theorem}

\section{Conclusion}\label{sec:conclusion}
In this paper,
we have shown that the reachability set of a $2$-BVASS admits a computable semilinear presentation.
This entails that the reachability problem for $2$-BVASS is decidable.
Our approach,
which is inspired from Hopcroft and Pansiot's algorithm for classical $2$-VASS~\cite{DBLP:journals/tcs/HopcroftP79},
does not provide any upper bound on the complexity of this problem.
The decidability status of the reachability problem for $d$-BVASS remains open in arbitrary dimension.

\bibliographystyle{plainurl}
\bibliography{biblio.bib}

\begin{thebibliography}{10}

\bibitem{DBLP:conf/fsttcs/AtigG11}
Mohamed~Faouzi Atig and Pierre Ganty.
\newblock Approximating petri net reachability along context-free traces.
\newblock In Supratik Chakraborty and Amit Kumar, editors, {\em {IARCS} Annual
  Conference on Foundations of Software Technology and Theoretical Computer
  Science, {FSTTCS} 2011, December 12-14, 2011, Mumbai, India}, volume~13 of
  {\em LIPIcs}, pages 152--163. Schloss Dagstuhl - Leibniz-Zentrum f{\"{u}}r
  Informatik, 2011.
\newblock \href {https://doi.org/10.4230/LIPICS.FSTTCS.2011.152}
  {\path{doi:10.4230/LIPICS.FSTTCS.2011.152}}.

\bibitem{DBLP:journals/jacm/BlondinEFGHLMT21}
Michael Blondin, Matthias Englert, Alain Finkel, Stefan G{\"{o}}ller, Christoph
  Haase, Ranko Lazic, Pierre McKenzie, and Patrick Totzke.
\newblock The reachability problem for two-dimensional vector addition systems
  with states.
\newblock {\em J. {ACM}}, 68(5):34:1--34:43, 2021.
\newblock \href {https://doi.org/10.1145/3464794} {\path{doi:10.1145/3464794}}.

\bibitem{DBLP:conf/pods/BojanczykDMSS06}
Mikolaj Bojanczyk, Claire David, Anca Muscholl, Thomas Schwentick, and Luc
  Segoufin.
\newblock Two-variable logic on data trees and {XML} reasoning.
\newblock In Stijn Vansummeren, editor, {\em Proceedings of the Twenty-Fifth
  {ACM} {SIGACT-SIGMOD-SIGART} Symposium on Principles of Database Systems,
  June 26-28, 2006, Chicago, Illinois, {USA}}, pages 10--19. {ACM}, 2006.
\newblock \href {https://doi.org/10.1145/1142351.1142354}
  {\path{doi:10.1145/1142351.1142354}}.

\bibitem{DBLP:journals/toplas/BouajjaniE13}
Ahmed Bouajjani and Michael Emmi.
\newblock Analysis of recursively parallel programs.
\newblock {\em {ACM} Trans. Program. Lang. Syst.}, 35(3):10:1--10:49, 2013.
\newblock \href {https://doi.org/10.1145/2518188} {\path{doi:10.1145/2518188}}.

\bibitem{DBLP:conf/focs/0001CMOSW24}
Dmitry Chistikov, Wojciech Czerwinski, Filip Mazowiecki, Lukasz Orlikowski,
  Henry Sinclair{-}Banks, and Karol Wegrzycki.
\newblock The tractability border of reachability in simple vector addition
  systems with states.
\newblock In {\em 65th {IEEE} Annual Symposium on Foundations of Computer
  Science, {FOCS} 2024, Chicago, IL, USA, October 27-30, 2024}, pages
  1332--1354. {IEEE}, 2024.
\newblock \href {https://doi.org/10.1109/FOCS61266.2024.00086}
  {\path{doi:10.1109/FOCS61266.2024.00086}}.

\bibitem{DBLP:conf/lics/ClementeLLM17}
Lorenzo Clemente, Slawomir Lasota, Ranko Lazic, and Filip Mazowiecki.
\newblock Timed pushdown automata and branching vector addition systems.
\newblock In {\em 32nd Annual {ACM/IEEE} Symposium on Logic in Computer
  Science, {LICS} 2017, Reykjavik, Iceland, June 20-23, 2017}, pages 1--12.
  {IEEE} Computer Society, 2017.
\newblock \href {https://doi.org/10.1109/LICS.2017.8005083}
  {\path{doi:10.1109/LICS.2017.8005083}}.

\bibitem{DBLP:conf/esop/Cotton-BarrattM17}
Conrad Cotton{-}Barratt, Andrzej~S. Murawski, and C.{-}H.~Luke Ong.
\newblock {ML} and extended branching {VASS}.
\newblock In Hongseok Yang, editor, {\em Programming Languages and Systems -
  26th European Symposium on Programming, {ESOP} 2017, Held as Part of the
  European Joint Conferences on Theory and Practice of Software, {ETAPS} 2017,
  Uppsala, Sweden, April 22-29, 2017, Proceedings}, volume 10201 of {\em
  Lecture Notes in Computer Science}, pages 314--340. Springer, 2017.
\newblock \href {https://doi.org/10.1007/978-3-662-54434-1\_12}
  {\path{doi:10.1007/978-3-662-54434-1\_12}}.

\bibitem{DBLP:journals/jacm/CzerwinskiLLLM21}
Wojciech Czerwinski, Slawomir Lasota, Ranko Lazic, J{\'{e}}r{\^{o}}me Leroux,
  and Filip Mazowiecki.
\newblock The reachability problem for {P}etri nets is not elementary.
\newblock {\em J. {ACM}}, 68(1):7:1--7:28, 2021.
\newblock \href {https://doi.org/10.1145/3422822} {\path{doi:10.1145/3422822}}.

\bibitem{DA87}
R.~David and H.~Alla.
\newblock Continuous petri nets.
\newblock In {\em Proc. of the 8th European Workshop on Application and Theory
  of Petri Nets, Zaragoza, Spain, 1987}, 1987.

\bibitem{GGS04}
Philippe de~Groote, Bruno Guillaume, and Sylvain Salvati.
\newblock Vector addition tree automata.
\newblock In {\em 19th {IEEE} Symposium on Logic in Computer Science {(LICS}
  2004), 14-17 July 2004, Turku, Finland, Proceedings}, pages 64--73. {IEEE}
  Computer Society, 2004.
\newblock \href {https://doi.org/10.1109/LICS.2004.1319601}
  {\path{doi:10.1109/LICS.2004.1319601}}.

\bibitem{DBLP:journals/jcss/DemriJLL13}
St{\'{e}}phane Demri, Marcin Jurdzinski, Oded Lachish, and Ranko Lazic.
\newblock The covering and boundedness problems for branching vector addition
  systems.
\newblock {\em J. Comput. Syst. Sci.}, 79(1):23--38, 2013.
\newblock \href {https://doi.org/10.1016/J.JCSS.2012.04.002}
  {\path{doi:10.1016/J.JCSS.2012.04.002}}.

\bibitem{DBLP:conf/icalp/FigueiraLLMS17}
Diego Figueira, Ranko Lazic, J{\'{e}}r{\^{o}}me Leroux, Filip Mazowiecki, and
  Gr{\'{e}}goire Sutre.
\newblock Polynomial-space completeness of reachability for succinct branching
  {VASS} in dimension one.
\newblock In Ioannis Chatzigiannakis, Piotr Indyk, Fabian Kuhn, and Anca
  Muscholl, editors, {\em 44th International Colloquium on Automata, Languages,
  and Programming, {ICALP} 2017, July 10-14, 2017, Warsaw, Poland}, volume~80
  of {\em LIPIcs}, pages 119:1--119:14. Schloss Dagstuhl - Leibniz-Zentrum
  f{\"{u}}r Informatik, 2017.
\newblock \href {https://doi.org/10.4230/LIPICS.ICALP.2017.119}
  {\path{doi:10.4230/LIPICS.ICALP.2017.119}}.

\bibitem{DBLP:conf/fsttcs/FinkelLS18}
Alain Finkel, J{\'{e}}r{\^{o}}me Leroux, and Gr{\'{e}}goire Sutre.
\newblock Reachability for two-counter machines with one test and one reset.
\newblock In Sumit Ganguly and Paritosh~K. Pandya, editors, {\em 38th {IARCS}
  Annual Conference on Foundations of Software Technology and Theoretical
  Computer Science, {FSTTCS} 2018, December 11-13, 2018, Ahmedabad, India},
  volume 122 of {\em LIPIcs}, pages 31:1--31:14. Schloss Dagstuhl -
  Leibniz-Zentrum f{\"{u}}r Informatik, 2018.
\newblock \href {https://doi.org/10.4230/LIPICS.FSTTCS.2018.31}
  {\path{doi:10.4230/LIPICS.FSTTCS.2018.31}}.

\bibitem{DBLP:journals/fuin/FracaH15}
Est{\'{\i}}baliz Fraca and Serge Haddad.
\newblock Complexity analysis of continuous petri nets.
\newblock {\em Fundam. Informaticae}, 137(1):1--28, 2015.
\newblock \href {https://doi.org/10.3233/FI-2015-1168}
  {\path{doi:10.3233/FI-2015-1168}}.

\bibitem{DBLP:conf/icalp/GanardiMPSZ22}
Moses Ganardi, Rupak Majumdar, Andreas Pavlogiannis, Lia Sch{\"{u}}tze, and
  Georg Zetzsche.
\newblock Reachability in bidirected pushdown {VASS}.
\newblock In Mikolaj Bojanczyk, Emanuela Merelli, and David~P. Woodruff,
  editors, {\em 49th International Colloquium on Automata, Languages, and
  Programming, {ICALP} 2022, July 4-8, 2022, Paris, France}, volume 229 of {\em
  LIPIcs}, pages 124:1--124:20. Schloss Dagstuhl - Leibniz-Zentrum f{\"{u}}r
  Informatik, 2022.
\newblock \href {https://doi.org/10.4230/LIPICS.ICALP.2022.124}
  {\path{doi:10.4230/LIPICS.ICALP.2022.124}}.

\bibitem{GS-PACIF66}
Seymour Ginsburg and Edwin~H. Spanier.
\newblock Semigroups, {P}resburger formulas and languages.
\newblock {\em Pacific Journal of Mathematics}, 16(2):285--296, 1966.
\newblock \href {https://doi.org/10.2140/pjm.1966.16.285}
  {\path{doi:10.2140/pjm.1966.16.285}}.

\bibitem{GollerHLT16}
Stefan G{\"{o}}ller, Christoph Haase, Ranko Lazi\'c, and Patrick Totzke.
\newblock A polynomial-time algorithm for reachability in branching {VASS} in
  dimension one.
\newblock In {\em {ICALP}}, volume~55 of {\em LIPIcs}, pages 105:1--105:13.
  Schloss Dagstuhl, 2016.
\newblock \href {https://doi.org/10.4230/LIPIcs.ICALP.2016.105}
  {\path{doi:10.4230/LIPIcs.ICALP.2016.105}}.

\bibitem{DBLP:journals/siglog/Haase18}
Christoph Haase.
\newblock A survival guide to presburger arithmetic.
\newblock {\em {ACM} {SIGLOG} News}, 5(3):67--82, 2018.
\newblock \href {https://doi.org/10.1145/3242953.3242964}
  {\path{doi:10.1145/3242953.3242964}}.

\bibitem{DBLP:journals/tcs/HopcroftP79}
John~E. Hopcroft and Jean{-}Jacques Pansiot.
\newblock On the reachability problem for 5-dimensional vector addition
  systems.
\newblock {\em Theor. Comput. Sci.}, 8:135--159, 1979.
\newblock \href {https://doi.org/10.1016/0304-3975(79)90041-0}
  {\path{doi:10.1016/0304-3975(79)90041-0}}.

\bibitem{DBLP:journals/corr/JacquemardSD16}
Florent Jacquemard, Luc Segoufin, and Jer{\'{e}}mie Dimino.
\newblock Fo2({\textless}, +1, {\textasciitilde}) on data trees, data tree
  automata and branching vector addition systems.
\newblock {\em Log. Methods Comput. Sci.}, 12(2), 2016.
\newblock \href {https://doi.org/10.2168/LMCS-12(2:3)2016}
  {\path{doi:10.2168/LMCS-12(2:3)2016}}.

\bibitem{DBLP:journals/jcss/KarpM69}
Richard~M. Karp and Raymond~E. Miller.
\newblock {Parallel Program Schemata}.
\newblock {\em J. Comput. Syst. Sci.}, 3(2):147--195, 1969.
\newblock \href {https://doi.org/10.1016/S0022-0000(69)80011-5}
  {\path{doi:10.1016/S0022-0000(69)80011-5}}.

\bibitem{Lazic2013}
Ranko Lazic.
\newblock The reachability problem for vector addition systems with a stack is
  not elementary.
\newblock {\em CoRR}, abs/1310.1767, 2013.
\newblock \href {https://doi.org/10.48550/arXiv.1310.1767}
  {\path{doi:10.48550/arXiv.1310.1767}}.

\bibitem{DBLP:journals/fuin/LazicNORW08}
Ranko Lazic, Thomas~Christopher Newcomb, Jo{\"{e}}l Ouaknine, A.~W. Roscoe, and
  James Worrell.
\newblock Nets with tokens which carry data.
\newblock {\em Fundam. Informaticae}, 88(3):251--274, 2008.
\newblock URL:
  \url{http://content.iospress.com/articles/fundamenta-informaticae/fi88-3-03}.

\bibitem{LazicS15}
Ranko Lazi\'c and Sylvain Schmitz.
\newblock Nonelementary complexities for branching {VASS}, {MELL}, and
  extensions.
\newblock {\em {ACM} Trans. Comput. Log.}, 16(3):20:1--20:30, 2015.
\newblock \href {https://doi.org/10.1145/2733375} {\path{doi:10.1145/2733375}}.

\bibitem{DBLP:conf/focs/Leroux21}
J{\'{e}}r{\^{o}}me Leroux.
\newblock The reachability problem for {P}etri nets is not primitive recursive.
\newblock In {\em 62nd {IEEE} Annual Symposium on Foundations of Computer
  Science, {FOCS} 2021, Denver, CO, USA, February 7-10, 2022}, pages
  1241--1252. {IEEE}, 2021.
\newblock \href {https://doi.org/10.1109/FOCS52979.2021.00121}
  {\path{doi:10.1109/FOCS52979.2021.00121}}.

\bibitem{DBLP:conf/lics/LerouxS19}
J{\'{e}}r{\^{o}}me Leroux and Sylvain Schmitz.
\newblock Reachability in vector addition systems is primitive-recursive in
  fixed dimension.
\newblock In {\em 34th Annual {ACM/IEEE} Symposium on Logic in Computer
  Science, {LICS} 2019, Vancouver, BC, Canada, June 24-27, 2019}, pages 1--13.
  {IEEE}, 2019.
\newblock \href {https://doi.org/10.1109/LICS.2019.8785796}
  {\path{doi:10.1109/LICS.2019.8785796}}.

\bibitem{DBLP:conf/icalp/LerouxST15}
J{\'{e}}r{\^{o}}me Leroux, Gr{\'{e}}goire Sutre, and Patrick Totzke.
\newblock On the coverability problem for pushdown vector addition systems in
  one dimension.
\newblock In Magn{\'{u}}s~M. Halld{\'{o}}rsson, Kazuo Iwama, Naoki Kobayashi,
  and Bettina Speckmann, editors, {\em Automata, Languages, and Programming -
  42nd International Colloquium, {ICALP} 2015, Kyoto, Japan, July 6-10, 2015,
  Proceedings, Part {II}}, volume 9135 of {\em Lecture Notes in Computer
  Science}, pages 324--336. Springer, 2015.
\newblock \href {https://doi.org/10.1007/978-3-662-47666-6\_26}
  {\path{doi:10.1007/978-3-662-47666-6\_26}}.

\bibitem{DBLP:conf/fossacs/Lugiez03}
Denis Lugiez.
\newblock Counting and equality constraints for multitree automata.
\newblock In Andrew~D. Gordon, editor, {\em Foundations of Software Science and
  Computational Structures, 6th International Conference, {FOSSACS} 2003 Held
  as Part of the Joint European Conference on Theory and Practice of Software,
  {ETAPS} 2003, Warsaw, Poland, April 7-11, 2003, Proceedings}, volume 2620 of
  {\em Lecture Notes in Computer Science}, pages 328--342. Springer, 2003.
\newblock \href {https://doi.org/10.1007/3-540-36576-1\_21}
  {\path{doi:10.1007/3-540-36576-1\_21}}.

\bibitem{Mayr84}
Ernst~W. Mayr.
\newblock An algorithm for the general petri net reachability problem.
\newblock {\em {SIAM} J. Comput.}, 13(3):441--460, 1984.
\newblock \href {https://doi.org/10.1137/0213029} {\path{doi:10.1137/0213029}}.

\bibitem{DBLP:conf/concur/MazowieckiP19}
Filip Mazowiecki and Michal Pilipczuk.
\newblock Reachability for bounded branching {VASS}.
\newblock In Wan~J. Fokkink and Rob van Glabbeek, editors, {\em 30th
  International Conference on Concurrency Theory, {CONCUR} 2019, August 27-30,
  2019, Amsterdam, the Netherlands}, volume 140 of {\em LIPIcs}, pages
  28:1--28:13. Schloss Dagstuhl - Leibniz-Zentrum f{\"{u}}r Informatik, 2019.
\newblock \href {https://doi.org/10.4230/LIPICS.CONCUR.2019.28}
  {\path{doi:10.4230/LIPICS.CONCUR.2019.28}}.

\bibitem{DBLP:conf/csl/Ohsaki01}
Hitoshi Ohsaki.
\newblock Beyond regularity: Equational tree automata for associative and
  commutative theories.
\newblock In Laurent Fribourg, editor, {\em Computer Science Logic, 15th
  International Workshop, {CSL} 2001. 10th Annual Conference of the EACSL,
  Paris, France, September 10-13, 2001, Proceedings}, volume 2142 of {\em
  Lecture Notes in Computer Science}, pages 539--553. Springer, 2001.
\newblock \href {https://doi.org/10.1007/3-540-44802-0\_38}
  {\path{doi:10.1007/3-540-44802-0\_38}}.

\bibitem{DBLP:conf/acl/Rambow94}
Owen Rambow.
\newblock Multiset-valued linear index grammars: Imposing dominance constraints
  on derivations.
\newblock In James Pustejovsky, editor, {\em 32nd Annual Meeting of the
  Association for Computational Linguistics, 27-30 June 1994, New Mexico State
  University, Las Cruces, New Mexico, USA, Proceedings}, pages 263--270. Morgan
  Kaufmann Publishers / {ACL}, 1994.
\newblock \href {https://doi.org/10.3115/981732.981768}
  {\path{doi:10.3115/981732.981768}}.

\bibitem{DBLP:conf/acl/Schmitz10}
Sylvain Schmitz.
\newblock On the computational complexity of dominance links in grammatical
  formalisms.
\newblock In Jan Hajic, Sandra Carberry, and Stephen Clark, editors, {\em {ACL}
  2010, Proceedings of the 48th Annual Meeting of the Association for
  Computational Linguistics, July 11-16, 2010, Uppsala, Sweden}, pages
  514--524. The Association for Computer Linguistics, 2010.
\newblock URL: \url{https://aclanthology.org/P10-1053/}.

\bibitem{DBLP:journals/dmtcs/VermaG05}
Kumar~Neeraj Verma and Jean Goubault{-}Larrecq.
\newblock Karp-miller trees for a branching extension of {VASS}.
\newblock {\em Discret. Math. Theor. Comput. Sci.}, 7(1):217--230, 2005.
\newblock \href {https://doi.org/10.46298/DMTCS.350}
  {\path{doi:10.46298/DMTCS.350}}.

\end{thebibliography}

\clearpage
\appendix

\section{Additional Results on Cones and Periodic Sets}
\subsection{General Dimension}

Any intersection of periodic sets is a periodic set. We recall that the intersection of two finitely-generated periodic sets is a finitely-generated periodic set~\cite{GS-PACIF66}. Given a set $\vec{A}\subseteq\setQ^d$, we also introduce the set $\perP{\vec{A}}$ of finite sums $\vec{a}_1+\cdots+\vec{a}_k$ where $k\in\setN$, and $\vec{a}_1,\ldots,\vec{a}_k\in\vec{A}$ are such that $\sum_{\ell=1}^j\vec{a}_\ell\geq \vec{0}$ for every $j\in\{0,\ldots,k\}$.

\begin{lemma}\label{lem:perpp}
   $\perP{\vec{A}}$ is periodic.
\end{lemma}
\begin{proof}
   $\vec{0}\in \perP{\vec{A}}$ is clear. Let us prove that $\vec{u}+\vec{v}\in \perP{\vec{A}}$ for every $\vec{u},\vec{v}\in \perP{\vec{A}}$. So, assume that $\vec{u}=\vec{x}_1+\cdots+\vec{x}_a$ for a sequence $\vec{x}_1,\ldots,\vec{x}_a\in\vec{A}$ such that $\sum_{\ell=1}^j\vec{x}_\ell\geq \vec{0}$ for every $j\in\{0,\ldots,a\}$, and symmetrically $\vec{v}=\vec{y}_1+\cdots+\vec{y}_b$ for a sequence $\vec{y}_1,\ldots,\vec{y}_b\in\vec{A}$ such that $\sum_{\ell=1}^j\vec{y}_\ell\geq \vec{0}$ for every $j\in\{0,\ldots,b\}$. Let us consider the sequence $\vec{a}_1,\ldots,\vec{a}_k$ defined as $\vec{x}_1,\ldots,\vec{x}_a,\vec{y}_1,\ldots,\vec{y}_b$ and just notice that $\vec{u}+\vec{v}=\vec{a}_1+\cdots+\vec{a}_k$ and $\sum_{\ell=1}^j\vec{a}_\ell\geq \vec{0}$ for every $j\in\{0,\ldots,k\}$ by considering the case $j\leq a$ and its complement separately. It follows that $\vec{u}+\vec{v}\in \perP{\vec{A}}$.
\end{proof}

If $\vec{P}_1$ and $\vec{P}_2$ are two periodic sets then $\con{\vec{P}_1}\cap\con{\vec{P}_2}=\con{\vec{P}_1\cap\vec{P}_2}$. It follows that the intersection of two finitely-generated cones is a finitely-generated cone. Given a set $\vec{A}\subseteq\setQ^d$, we also introduce the set $\conP{\vec{A}}=\setQ_{\geq 0}\perP{\vec{A}}$. \cref{lem:perpp} shows that $\conP{\vec{A}}$ is a cone. Notice that $\conP{\vec{A}}$ is the set of finite sums $\lambda_1\vec{a}_1+\cdots+\lambda_k\vec{a}_k$ where $k\in\setN$, $\vec{a}_1,\ldots,\vec{a}_k\in\vec{A}$, and $\lambda_1,\ldots,\lambda_k\in\setQ_{\geq 0}$ are such that $\sum_{\ell=1}^j\lambda_\ell\vec{a}_\ell\geq \vec{0}$ for every $j\in\{0,\ldots,k\}$.

\begin{remark}
   The periodic set $\perP{\vec{A}}$ can be seen as the reachability set of a vector addition system where $\vec{A}$ is its infinite set of actions, while $\conP{\vec{A}}$ is the same set but under the continuous semantics~\cite{DA87,DBLP:journals/fuin/FracaH15}.
\end{remark}

\medskip

We recall the following lemma.
\begin{lemma}[\cite{DBLP:journals/tcs/HopcroftP79}]\label{lem:HPcone}
  Let $\vec{B} \subseteq \setZ^d$ and let $\vec{P} \subseteq \setZ^d$ be a finitely-generated periodic set.
  If $\vec{B} + \vec{P} \subseteq \vec{x} + \con{\vec{P}}$ for some $\vec{x} \in \setZ^d$ then
  $\vec{B} + \vec{P} = \vec{B}' + \vec{P}$ for some finite subset $\vec{B}' \subseteq \vec{B}$.
\end{lemma}

\medskip

Given a vector $\vec{v}\in\setQ^d$, the \emph{$\vec{v}$-stabilization} of a cone $\vec{C}\subseteq \setQ_{\geq 0}^d$ is the cone $(\vec{C}+\setQ_{\geq 0}\vec{v})\cap\setQ_{\geq 0}^d$, also denoted as $\stab{\vec{v}}{\vec{C}}$. A cone $\vec{C}\subseteq \setQ_{\geq 0}^d$ is said to be \emph{$\vec{v}$-stable} for some vector $\vec{v}\in\setQ^d$ if $\stab{\vec{v}}{\vec{C}}=\vec{C}$. The \emph{stabililizer} of a cone $\vec{C}\subseteq \setQ_{\geq 0}^d$ is the set $\stab{\star}{\vec{C}}$ of vectors $\vec{v}\in\setQ^d$ such that $\vec{C}$ is $\vec{v}$-stable.

\begin{example}
  \label{exa:stabilizer-dim-two}
  We have
  $\stab{\star}{\vec{C}} = (\setQ^2 \setminus \setQ_{\geq 0}^2) \cup \{(0, 0)\}$ if $\vec{C} = \{(0, 0)\}$,
  $\stab{\star}{\vec{C}} = \vec{C}$ if $\vec{C} = \con{\{(2, 1), (1, 2)\}}$, and
  $\stab{\star}{\vec{C}} = \{(x, y) \in \setQ^2 \mid 2y \geq x\}$ if $\vec{C} = \con{\{(2, 1), (0, 1)\}}$.
  \lipicsEnd
\end{example}

\begin{lemma}\label{lem:stabcone}
  The stabilizer $\vec{S}$ of every cone $\vec{C} \subseteq \setQ_{\geq 0}^d$ satisfies $\vec{C} = \vec{S} \cap \setQ_{\geq 0}^d$.
\end{lemma}
\begin{proof}
  Observe that $\stab{\vec{v}}{\vec{C}} = \vec{C}$ for every $\vec{v} \in \vec{C}$.
  It follows that $\vec{C} \subseteq \vec{S} \cap \setQ_{\geq 0}^d$.
  Let $\vec{v} \in \vec{S}$ such that $\vec{v} \geq \vec{0}$.
  Note that $\stab{\vec{v}}{\vec{C}} = \vec{C}$ since $\vec{v} \in \vec{S}$.
  As $\vec{C}$ contains $\vec{0}$,
  we derive from the definition of $\stab{\vec{v}}{\vec{C}}$ that
  $\vec{v} \in \stab{\vec{v}}{\vec{C}} = \vec{C}$.
  This concludes the proof that $\vec{C} = \vec{S} \cap \setQ_{\geq 0}^d$.
\end{proof}

\begin{lemma}\label{lem:conp-stabilizer}
  It holds that $\vec{A} \subseteq \stab{\star}{\conP{\vec{A}}}$ for every set $\vec{A} \subseteq \setQ^d$.
\end{lemma}
\begin{proof}
  Let $\vec{v}\in \conP{\vec{A}}$, $\lambda\in\setQ_{\geq 0}$, and $\vec{a}\in\vec{A}$ such that the vector $\vec{w}$ defined as $\vec{v}+\lambda\vec{a}$ is in $\setQ_{\geq 0}^d$, and let us prove that $\vec{w}\in \conP{\vec{A}}$. As $\vec{v}\in\conP{\vec{A}}$, there exists $k\in\setN$, a sequence $\vec{a}_1,\ldots,\vec{a}_k\in\vec{A}$, and a sequence $\lambda_1,\ldots,\lambda_k\in\setQ_{\geq 0}$ such that $\vec{x}=\lambda_1\vec{a}_1+\cdots+\lambda_k\vec{a}_k$ and $\sum_{\ell=1}^j\lambda_\ell\vec{a}_\ell\geq \vec{0}$ for every $j\in\{0,\ldots,k\}$. Let us consider $\vec{a}_{k+1}=\vec{a}$ and $\lambda_{k+1}=\lambda$ and observe that $\vec{w}=\lambda_1\vec{a}_1+\cdots+\lambda_{k+1}\vec{a}_{k+1}$ and $\sum_{\ell=1}^j\lambda_\ell\vec{a}_\ell\geq \vec{0}$ for every $j\in\{0,\ldots,k+1\}$. It follows that $\vec{w}$ is in $\conP{\vec{A}}$. We have proved that for every $\vec{a}\in\vec{A}$, the set $\conP{\vec{A}}$ is $\vec{a}$-stable. Therefore $\vec{A}\subseteq \stab{\star}{\conP{\vec{A}}}$.
\end{proof}

\subsection{Dimension Two}\label{sec:dim2}
Cones and stabilizers are much simpler in dimension two than in the general case
(see, e.g., \cref{exa:stabilizer-dim-two}).
\begin{lemma}
  \label{fact:axis-in-stab}
  Consider a cone $\vec{C} \subseteq \setQ_{\geq 0}^2$ with $\vec{C} \neq \{(0, 0)\}$.
  For every vector $\vec{v}\in\setQ^2$,
  \begin{itemize}
  \item
    if $\vec{v}(1) > 0 \geq \vec{v}(2)$ then $(1, 0) \in \stab{\vec{v}}{\vec{C}}$, and
  \item
    if $\vec{v}(2) > 0 \geq \vec{v}(1)$ then $(0, 1) \in \stab{\vec{v}}{\vec{C}}$.
  \end{itemize}
\end{lemma}
\begin{proof}
  We only prove the first item, as the second item can be obtained by symmetry.
  Pick $\vec{c} \in \vec{C}$ such that $\vec{c} \neq (0, 0)$.
  Assume that $\vec{v}(1) > 0 \geq \vec{v}(2)$.
  If $\vec{v}(2) = 0$ then $\vec{v} \in \stab{\vec{v}}{\vec{C}}$ since $\vec{v} \geq (0, 0)$.
  Otherwise,
  we have $\vec{v}(2) < 0$,
  so the vector $\vec{u}$ defined by
  $\vec{u} = \vec{c} + \lambda \vec{v}$,
  where $\lambda = -\frac{\vec{c}(2)}{\vec{v}(2)} \geq 0$,
  verifies $\vec{u}(1) \geq 0 = \vec{u}(2)$ and $\vec{u} \in \stab{\vec{v}}{\vec{C}}$.
  Moreover,
  $\vec{u}(1) > 0$ since $\vec{c} \neq (0, 0)$.
  In both cases,
  $\stab{\vec{v}}{\vec{C}}$ intersects $\setQ_{> 0} (1, 0)$,
  hence,
  $(1, 0) \in \stab{\vec{v}}{\vec{C}}$.
\end{proof}

The following lemma shows that the stabilizer of any cone $\vec{C}\subseteq \setQ_{\geq 0}^2$ is a cone except when $\vec{C}=\{(0,0)\}$.
\begin{lemma}
  \label{lem:useless-stabilizers}
  The stabilizer of a cone $\vec{C}\subseteq \setQ_{\geq 0}^2$ such that $\vec{C}\not=\{(0,0)\}$ is a cone.
\end{lemma}
\begin{proof}
  Let $\vec{S}$ be the stabilizer of a cone $\vec{C}\subseteq \setQ_{\geq 0}^2$ such that $\vec{C}\not=\{(0,0)\}$. It is routinely checked that $(0, 0) \in \vec{S}$ and $\setQ_{\geq 0} \vec{S} \subseteq \vec{S}$.
  We still need to prove that $\vec{S} + \vec{S} \subseteq \vec{S}$.
  Let $\vec{u}, \vec{v} \in \vec{S}$ and
  let $\vec{x} \in \stab{\vec{u} + \vec{v}}{\vec{C}}$.
  Note that $\vec{x} \geq (0, 0)$ and
  let us show that $\vec{x} \in \vec{C}$.
  We have $\vec{x} = \vec{c} + \lambda (\vec{u} + \vec{v})$
  for some $\vec{c} \in \vec{C}$ and $\lambda \in \setQ_{\geq 0}$.
  If $\vec{c} + \lambda \vec{u} \geq (0, 0)$ then
  $\vec{c} + \lambda \vec{u}$ is in $\vec{C}$ since $\vec{C}$ is $\vec{u}$-stable.
  Since $\vec{c} + \lambda \vec{u} + \lambda \vec{v} = \vec{x} \geq (0, 0)$,
  we get that $\vec{x}$ is in $\vec{C}$ since $\vec{C}$ is $\vec{v}$-stable.
  Analogously,
  if $\vec{c} + \lambda \vec{v} \geq (0, 0)$ then we obtain that
  $\vec{c} + \lambda \vec{v}$ is in $\vec{C}$,
  hence,
  $\vec{x}$ is in $\vec{C}$.
  Assume that
  $\vec{c} + \lambda \vec{u} \not\geq (0, 0)$ and
  $\vec{c} + \lambda \vec{v} \not\geq (0, 0)$.
  There exists $i \in \{1, 2\}$ such that $(\vec{c} + \lambda \vec{u})(i) < 0$.
  Since $\vec{c} + \lambda \vec{u} + \lambda \vec{v} \geq (0, 0)$,
  we derive\footnote{%
    The proof of this fact is as follows.
    Observe that $\vec{c} \geq (0, 0)$ since $\vec{c} \in \vec{C}$.
    Recall also that $\lambda \geq 0$.
    We have $\vec{u}(i) < 0$ because $(\vec{c} + \lambda \vec{u})(i) < 0$.
    It follows that $\vec{v}(i) > 0$ since $\vec{c} + \lambda \vec{u} + \lambda \vec{v} \geq (0, 0)$.
    As $\vec{c} + \lambda \vec{v} \not\geq (0, 0)$,
    we deduce that $(\vec{c} + \lambda \vec{v})(\bar{i}) < 0$.
    Hence,
    $\vec{v}(\bar{i}) < 0$ and moreover
    $\vec{u}(\bar{i}) > 0$ since $\vec{c} + \lambda \vec{u} + \lambda \vec{v} \geq (0, 0)$.
  }
  that
  $\vec{u}(i) < 0 \wedge \vec{u}(\bar{i}) > 0$
  and
  $\vec{v}(i) > 0 \wedge \vec{v}(\bar{i}) < 0$,
  where $\bar{i} = 3 - i$.
  Recall that $\vec{C} = \stab{\vec{u}}{\vec{C}} = \stab{\vec{v}}{\vec{C}}$.
  It follows from \cref{fact:axis-in-stab} that
  $\vec{C}$ contains $\{(0, 1), (1, 0)\}$,
  hence,
  $\vec{C}$ is the cone $\setQ_{\geq 0}^2$.
  This entails that $\vec{x}$ is in $\vec{C}$.
  We have shown that $\stab{\vec{u} + \vec{v}}{\vec{C}} \subseteq \vec{C}$,
  which means that $\vec{u} + \vec{v}$ is in $\vec{S}$.
\end{proof}

\begin{lemma}
  \label{lem:cone-accel}
  For every set $\vec{A} \subseteq \setQ^2$,
  it holds that
  \begin{itemize}
  \item
    $\conP{\vec{A}} = \{(0, 0)\}$ if $\vec{A}$ is disjoint from $\setQ_{\geq 0}^2 \setminus \{(0, 0)\}$, and
  \item
    $\conP{\vec{A}} = \con{\vec{A}} \cap \setQ_{\geq 0}^2$ otherwise.
  \end{itemize}
\end{lemma}
\begin{proof}
  Let $\vec{A} \subseteq \setQ^2$ and let $\vec{C} = \conP{\vec{A}}$.
  If $\vec{A}$ is disjoint from $\setQ_{\geq 0}^2 \setminus \{(0, 0)\}$ then clearly
  $\vec{C} = \{(0, 0)\}$.
  Otherwise,
  $\vec{A}$ intersects $\setQ_{\geq 0}^2 \setminus \{(0, 0)\}$.
  It follows that $\vec{C} \neq \{(0, 0)\}$ since $(\vec{A} \cap \setQ_{\geq 0}^2) \subseteq \vec{C}$.
  Let us show that $\vec{C} = \con{\vec{A}} \cap \setQ_{\geq 0}^2$.
  We observe that
  $\vec{C}$ is the set of finite sums $\sum_{j=1}^k \lambda_j\vec{a}_j$,
  where $k \in \setN$,
  $\lambda_1, \ldots, \lambda_k \in \setQ_{\geq 0}$ and
  $\vec{a}_1, \ldots, \vec{a}_k \in \vec{A}$,
  such that $\sum_{j=1}^h \lambda_j\vec{a}_j \geq (0, 0)$ for every $h \in \{0, \ldots, k\}$.
  So the inclusion $\vec{C} \subseteq \con{\vec{A}} \cap \setQ_{\geq 0}^2$ is straightforward.
  It remains to prove the converse inclusion.
  It is routinely checked that
  $\vec{C}$ is $\vec{a}$-stable for every $\vec{a} \in \vec{A}$.
  So $\vec{A} \subseteq \stab{\star}{\vec{C}}$.
  According to \cref{lem:stabcone} and \cref{lem:useless-stabilizers},
  $\vec{C} = \stab{\star}{\vec{C}} \cap \setQ_{\geq 0}^2$ and $\stab{\star}{\vec{C}}$ is a cone.
  We derive from $\vec{A} \subseteq \stab{\star}{\vec{C}}$ that
  $\con{\vec{A}} \subseteq \stab{\star}{\vec{C}}$,
  hence,
  $\con{\vec{A}} \cap \setQ_{\geq 0}^2 \subseteq \vec{C}$.
\end{proof}

\begin{remark}
   The previous lemma cannot be extended to higher dimension. In fact, let $\vec{C}=\setQ_{\geq 0}(0,0,1)$ and observe that the two vectors $(-1,2,0)$ and $(2,-1,0)$ are in the stabilizer of $\vec{C}$ but the sum of those two vectors, i.e. $(1,1,0)$ is not in the stabilizer of $\vec{C}$.
\end{remark}

By observing that $\conP{\vec{A}}$ is a finitely-generated cone when $\vec{A}$ is a finite subset of $\setZ^2$, we deduce the following lemma.
\begin{lemma}
  \label{lem:accel-finitely-generated-dim-2}
  $\perP{\vec{A}}$ is a finitely-generated periodic set for every finite set $\vec{A}\subseteq \setZ^2$. 
\end{lemma}
\begin{proof}  
  Let $\vec{C} = \con{\perP{\vec{A}}}$. Notice that $\vec{C}=\conP{\vec{A}}$.
  We derive from \cref{lem:cone-accel} that
  $\vec{C}$ is $\{(0, 0)\}$ or $\con{\vec{A}} \cap \setQ_{\geq 0}^2$.
  As $\vec{A}$ is finite,
  we get in both cases that $\vec{C}$ is a finitely-generated cone as the intersection of two finitely-generated cones. It follows from \cref{cor:HPcone} that $\perP{\vec{A}}$ is a finitely-generated periodic set.
\end{proof}

\begin{remark}
  \cref{lem:accel-finitely-generated-dim-2} cannot be extended when $d\geq 3$ since there exists finite sets $\vec{A}\subseteq \setZ^3$ such that the periodic set $\perP{\vec{A}}$ is not finitely-generated. For instance, let $\vec{A}=\{(1,1,1),(-1,2,0),(-1,0,0),(2,-1,0),(0,-1,0)\}$, and observe that $\perP{\vec{A}}=\{(0,0,0)\} \cup (\setN^2\times\setN_{>0})$.
\end{remark}

\begin{corollary}\label{cor:accel-finitely-generated-dim-2}
   For every finite set $\vec{Z}\subseteq \setZ^2$, we can effectively compute a finite set $\vec{X}\subseteq \setN^2$ such that $\perP{\vec{Z}}=\per{\vec{X}}$.
\end{corollary}
\begin{proof}
  Let $\vec{P}=\perP{\vec{Z}}$. We introduce the sequence $(\vec{X}_n)_{n\in\setN}$ defined by $\vec{X}_0=\{(0,0)\}$ and by induction by $\vec{X}_{n+1}=\vec{X}_n\cup (\vec{X}_n+\vec{Z})\cap\setN^2$. Notice that $\vec{X}_n\subseteq \vec{P}$ by induction on $n\in\setN$.

 Let us prove that $\per{\vec{X}_n}=\vec{P}$ for some $n\in\setN$. \cref{lem:accel-finitely-generated-dim-2} shows that $\vec{P}$ is a finitely-generated periodic set. It follows that there exists a finite set $\vec{Q}\subseteq \vec{P}$ such that $\vec{P}=\per{\vec{Q}}$. From $\vec{Q}\subseteq\vec{P}=\bigcup_{n\in\setN}\vec{X}_n$ we deduce that there exists $n\in\setN$ such that $\vec{Q}\subseteq \vec{X}_n$. From $\vec{P}=\per{\vec{Q}}\subseteq \per{\vec{X}_n}\subseteq \vec{P}$ we get $\per{\vec{X}_n}=\vec{P}$.

  We are going to prove that there exists $n\in\setN$ such that $(\per{\vec{X}_n}+\vec{Z})\cap\setN^2\subseteq \per{\vec{X}_n}$ and for any such a $n\in\setN$, we have $\vec{P}=\per{\vec{X}_n}$. The corollary will be then derived by just observing that the inclusion $(\per{\vec{X}_{n}}+\vec{Z})\cap\setN^2\subseteq \per{\vec{X}_n}$ can be decided with classical algorithms on semilinear sets~\cite{GS-PACIF66,DBLP:journals/siglog/Haase18}.

  Notice that it is sufficient to prove that for every $n\in\setN$, we have $\per{\vec{X}_n}=\vec{P}$ iff $(\per{\vec{X}_n}+\vec{Z})\cap\setN^2\subseteq \per{\vec{X}_n}$. Let $n\in\setN$. Since $(\vec{P}+\vec{Z})\cap \setN^2\subseteq \vec{P}$, it follows that $\per{\vec{X}_n}=\vec{P}$ implies $(\per{\vec{X}_n}+\vec{Z})\cap\setN^2\subseteq \per{\vec{X}_n}$. Conversely, assume that $(\per{\vec{X}_n}+\vec{Z})\cap\setN^2\subseteq \per{\vec{X}_n}$ and let us prove $\per{\vec{X}_n}=\vec{P}$. Since $\vec{X}_n\subseteq\vec{P}$, it is sufficient to prove that $\vec{P}\subseteq \per{\vec{X}_n}$. Let $\vec{p}\in\vec{P}$. There exists a sequence $\vec{a}_1,\ldots,\vec{a}_k\in\vec{Z}$ such that $\vec{p}=\sum_{j=1}^k\vec{a}_j$ and such that $\sum_{j=1}^\ell\vec{a}_j\geq (0,0)$ for every $\ell\in\{1,\ldots,k\}$. We introduce $\vec{p}_\ell=\sum_{j=1}^\ell\vec{a}_j$. Let us prove by induction on $\ell\in \{0,\ldots,k\}$ that $\vec{p}_\ell\in \per{\vec{X}_n}$. For $\ell=0$ the membership is trivial. Let us assume that $\vec{p}_\ell\in \per{\vec{X}_n}$ for some $\ell\in \{0,\ldots,k-1\}$ and let us prove that $\vec{p}_{\ell+1}\in \per{\vec{X}_n}$. Notice that $\vec{p}_{\ell+1}=\vec{p}_\ell+\vec{a}_{\ell+1}$. It follows that $\vec{p}_{\ell+1}\in (\per{\vec{X}_n}+\vec{Z})\cap\setN^2$. We deduce that $\vec{p}_{\ell+1}\in\per{\vec{X}_n}$ and the induction is proved. It follows that $\vec{p}=\vec{p}_k\in\per{\vec{X}_n}$. We have proved that $\per{\vec{X}_n}=\vec{P}$.
\end{proof}

\lemAccel*
\begin{proof}
   Just observe that $\accel{\vec{I}}(\per{\vec{G}})=\perP{\vec{I}\cup\vec{G}}$ and apply \cref{cor:accel-finitely-generated-dim-2}.
\end{proof}

\section{Proofs of \cref{sec:correctness}}
\subsection{Proofs of \cref{subsec:soundness-algorithmic-explorations}}

\lemCorrectnessAccelerationIterableVectors*
\begin{proof}
  Consider a configuration $q(\vec{x})$ of $\bvass$ and a periodic subset $\vec{P}$ of $\setN^2$.
  For short,
  we introduce the sets
  $\vec{F}, \vec{T} \subseteq \setZ^2$ and $\vec{Q} \subseteq \setN^2$
  defined as follows:
  \begin{align*}
    \vec{F}
    & = \{\vec{y} - \vec{x} \mid \vec{y} \in \setN^2 \text{ and } \exists \vec{u} \in \vec{P} : q(\vec{x} + \vec{u}) \step[*] q(\vec{y})\}
    \\
    \vec{T}
    & = \{\vec{x} - \vec{y} \mid \vec{y} \in \setN^2 \text{ and } \exists \vec{u} \in \vec{P} : q(\vec{y} + \vec{u}) \step[*] q(\vec{x})\}
    \\
    \vec{Q}
    & = \clo{q}{\vec{x}}(\vec{P})
  \end{align*}
  It is readily seen that $\vec{J} = \vec{F} \cup \vec{T}$ and that $\vec{Q} = (\vec{F} \cap \setN^2)$.
  To prove the lemma,
  we need to show that $\accel{\vec{J}}(\vec{P}) = \vec{Q}$.
  The equality $\vec{Q} = (\vec{F} \cap \setN^2)$ entails that $\accel{\vec{J}}(\vec{P}) \supseteq \vec{Q}$.
  Let us show the converse inclusion $\accel{\vec{J}}(\vec{P}) \subseteq \vec{Q}$.
  According to the definition of $\accel{\vec{J}}(\vec{P})$,
  it is enough to prove that $\vec{Q}$ contains both $\vec{P}$ and $(\vec{Q} + \vec{J}) \cap \setN^2$.
  We have $\vec{P} \subseteq \vec{Q}$ since $\clo{q}{\vec{x}}$ is extensive.
  Let $\vec{v} \in \vec{Q}$ and $\vec{a} \in \vec{J}$ such that $\vec{v} + \vec{a} \geq (0, 0)$.
  We show that $(\vec{v} + \vec{a}) \in \vec{Q}$.
  Since $\vec{v} \in \vec{Q}$,
  there exists $\vec{u} \in \vec{P}$ such that
  $q(\vec{x} + \vec{u}) \step[*] q(\vec{x} + \vec{v})$.
  Since $\vec{a} \in \vec{J}$,
  we have $\vec{a} \in \vec{F}$ or $\vec{a} \in \vec{T}$.
  We consider these two cases separately.
  \begin{itemize}
  \item
    If $\vec{a} \in \vec{F}$ then
    $\vec{a} = \vec{y} - \vec{x}$ for some $\vec{y} \in \setN^2$ and $\vec{u}' \in \vec{P}$ such that $q(\vec{x} + \vec{u}') \step[*] q(\vec{y})$.
    As $\vec{u}' \geq (0, 0)$ and $\vec{v} \geq (0, 0)$,
    we get by diagonality\footnotemark[\getrefnumber{footnote:diagonal}] of $\step[*]$ that
    $q(\vec{x} + \vec{u} + \vec{u}')
    \step[*]
    q(\vec{x} + \vec{v} + \vec{u}')
    \step[*]
    q(\vec{y} + \vec{v})
    =
    q(\vec{x} + \vec{v} + \vec{a})$.
  \item
    If $\vec{a} \in \vec{T}$ then
    $\vec{a} = \vec{x} - \vec{y}$ for some $\vec{y} \in \setN^2$ and $\vec{u}' \in \vec{P}$ such that $q(\vec{y} + \vec{u}') \step[*] q(\vec{x})$.
    As $\vec{u}' \geq (0, 0)$ and $\vec{v} + \vec{a} \geq (0, 0)$,
    we get by diagonality\footnotemark[\getrefnumber{footnote:diagonal}] of $\step[*]$ that
    $q(\vec{x} + \vec{u} + \vec{u}')
    \step[*]
    q(\vec{x} + \vec{v} + \vec{u}')
    =
    q(\vec{y} + \vec{u}' + \vec{v} + \vec{a})
    \step[*]
    q(\vec{x} + \vec{v} + \vec{a})$.
  \end{itemize}
  We obtain, in both cases, that
  $q(\vec{x} + \vec{u} + \vec{u}') \step[*] q(\vec{x} + \vec{v} + \vec{a})$
  for some $\vec{u}' \in \vec{P}$.
  Note that $\vec{u} + \vec{u}'$ is in $\vec{P}$ since $\vec{P}$ is periodic,
  and recall that $\vec{v} + \vec{a} \geq (0, 0)$ by assumption.
  It follows that
  $(\vec{v} + \vec{a}) \in \vec{Q}$.
  We have thus shown that
  $(\vec{Q} + \vec{J}) \cap \setN^2$ is contained in $\vec{Q}$,
  which concludes the proof of the lemma.
\end{proof}

\begin{lemma}
  \label{lem:local-post}
  Let $\explo = (N, \rightarrow, \lambda)$ be an algorithmic exploration of a $2$-BVASS $\bvass$,
  let $\left(\vec{Q}_n\right)_{n \in N}$ be a family of periodic subsets of $\setN^2$, and
  let $n \in N$.
  Assume that $\vec{P}_m \subseteq \vec{Q}_m$ and $q_m(\vec{z}_m + \vec{Q}_m) \subseteq \ReachSet{\bvass}$
  for every $m \in N$ with $m \rightarrow n$.
  Then
  $q_n(\vec{z}_n + \sum_{m \rightarrow n} \vec{Q}_m) \subseteq \ReachSet{\bvass}$.
  Moreover,
  $q_n(\vec{z}_n + \sum_{m \rightarrow n} \vec{Q}_m) \subseteq \relimg{\step}{q_t(\vec{z}_t + \vec{Q}_t)}$
  for every $t \in N$ with $t \rightarrow n$.
\end{lemma}
\begin{proof}
  Let $\vec{x} \in (\vec{z}_n + \sum_{m \rightarrow n} \vec{Q}_m)$.
  By assumption,
  it holds for every $m \in N$ with $m \rightarrow n$ that
  $\vec{P}_m \subseteq \vec{Q}_m$ and that $\vec{Q}_m$ is periodic.
  We derive from Condition~\ref{explo:z} of \cref{def:algorithmic-exploration} that
  $\vec{x} = \vec{a}_n + \sum_{m \rightarrow n} (\vec{z}_m + \vec{y}_m)$
  for some family $\left(\vec{y}_m\right)_{m \rightarrow n}$ of vectors $\vec{y}_m \in \vec{Q}_m$.
  By assumption,
  $q_m(\vec{z}_m + \vec{y}_m) \in \ReachSet{\bvass}$ for every $m \in N$ with $m \rightarrow n$.
  Furthermore,
  according to Condition~\ref{explo:v+q} of \cref{def:algorithmic-exploration},
  the multiset $S = \multiset{q_m \mid m \rightarrow n}$ is a set and
  verifies $(S, \vec{a}_n, q_n) \in \Delta$.
  It follows that $q_n(\vec{x}) \in \post[\bvass](\ReachSet{\bvass})$,
  hence,
  $q_n(\vec{x}) \in \ReachSet{\bvass}$.
  Moreover,
  for every $t \in N$ with $t \rightarrow n$,
  we have $q_t(\vec{z}_t + \vec{y}_t) \step q_n(\vec{x})$.
  We have shown that
  $q_n(\vec{z}_n + \sum_{m \rightarrow n} \vec{Q}_m) \subseteq \ReachSet{\bvass}$
  and that
  $q_n(\vec{z}_n + \sum_{m \rightarrow n} \vec{Q}_m) \subseteq \relimg{\step}{q_t(\vec{z}_t + \vec{Q}_t)}$
  for every $t \in N$ with $t \rightarrow n$.
\end{proof}

\begin{corollary}
  \label{cor:valid-explorations-post+}
  Let $\explo = (N, \rightarrow, \lambda)$ be an algorithmic exploration of a $2$-BVASS $\bvass$ and
  let $n \in N$.
  Assume that $q_s(\vec{z}_s + \vec{P}_s) \subseteq \ReachSet{\bvass}$
  for every $s \in N$ with $s \xrightarrow{+} n$.
  Then
  $q_n(\vec{z}_n)$ is in $\relimg{\step[*]}{q_s(\vec{z}_s + \vec{P}_t)}$
  for every $s, t \in N$ with $s \xrightarrow{*} t \rightarrow n$.
\end{corollary}
\begin{proof}
  Let $s, t \in N$ with $s \xrightarrow{*} t \rightarrow n$.
  There exist $k \geq 0$ and $n_0, \ldots, n_{k+1} \in N$ such that
  $s = n_0 \rightarrow n_1 \cdots \rightarrow n_k = t$ and $n_{k+1} = n$.
  By \cref{lem:local-post},
  applied with $\left(\vec{Q}_n\right)_{n \in N}$ equal to $\left(\vec{P}_n\right)_{n \in N}$,
  there exist
  $\vec{u}_0 \in \vec{P}_{n_0}, \ldots, \vec{u}_k \in \vec{P}_{n_k}$
  such that
  $q_{n_i}(\vec{z}_{n_i} + \vec{u}_i) \step q_{n_{i+1}}(\vec{z}_{n_{i+1}})$
  for every $i \in \{0, \ldots, k\}$.
  We get by diagonality\footnotemark[\getrefnumber{footnote:diagonal}] of $\step[*]$ that
  $q_s(\vec{z}_s + \vec{u}_0 + \cdots + \vec{u}_k) \step[*] q_n(\vec{z}_n)$.
  Since $\explo$ is an algorithmic exploration,
  $\vec{P}_t$ is periodic and
  $\vec{P}_{n_0} \subseteq \cdots \subseteq \vec{P}_{n_k} = \vec{P}_t$.
  Hence,
  $(\vec{u}_0 + \cdots + \vec{u}_k) \in \vec{P}_t$.
  We have shown that
  $q_n(\vec{z}_n)$ is in $\relimg{\step[*]}{q_s(\vec{z}_s + \vec{P}_t)}$.
\end{proof}

\lemSoundnessAlgorithmicExplorations*
\begin{proof}
  Let $\explo = (N, \rightarrow, \lambda)$ be an algorithmic exploration of $2$-BVASS $\bvass$.
  We introduce the family $\left(\vec{Q}_n\right)_{n \in N}$ of subsets of $\setN^2$
  defined,
  by well-founded recursion over $\rightarrow$,
  by
  $
  \vec{Q}_n = \clo{q_n}{\vec{z}_n}(\sum_{m \rightarrow n} \vec{Q}_m)
  $
  for every node $n \in N$.
  Recall that,
  for every configuration $q(\vec{x})$ of $\bvass$,
  the function $\clo{q}{\vec{x}}$ preserves periodicity and is extensive.
  It follows that,
  for every $m, n \in N$,
  the set $\vec{Q}_n$ is periodic and
  $m \rightarrow n$ implies $\vec{Q}_m \subseteq \vec{Q}_n$.

  \smallskip

  We show by well-founded induction over $\rightarrow$ that,
  for all $n \in N$,
  we have
  $\vec{P}_n \subseteq \vec{Q}_n$ and $q_n(\vec{z}_n + \vec{Q}_n) \subseteq \ReachSet{\bvass}$.
  Let $n \in N$ and
  assume that
  $\vec{P}_s \subseteq \vec{Q}_s$ and $q_s(\vec{z}_s + \vec{Q}_s) \subseteq \ReachSet{\bvass}$
  for all $s \in N$ with $s \xrightarrow{+} n$.
  We obtain from \cref{lem:local-post} that
  $q_n(\vec{z}_n + \sum_{m \rightarrow n} \vec{Q}_m) \subseteq \ReachSet{\bvass}$.
  Recall that $\vec{Q}_n = \clo{q_n}{\vec{z}_n}(\sum_{m \rightarrow n} \vec{Q}_m)$.
  It follows that
  $q_n(\vec{z}_n + \vec{Q}_n)
  \subseteq
  \relimg{\step[*]}{q_n(\vec{z}_n + \sum_{m \rightarrow n} \vec{Q}_m)}
  \subseteq
  \ReachSet{\bvass}$.

  \smallskip

  It remains to show that $\vec{P}_n \subseteq \vec{Q}_n$.
  Let $\vec{J}_n$ denote the set $\vec{J}$ defined in
  \cref{lem:correctness-of-acceleration-via-iterable-vectors}
  with $q(\vec{x}) := q_n(\vec{z}_n)$ and $\vec{P} := \vec{Q}_n$.
  The crucial observation now is that $\vec{E}_n$ and $\overline{\vec{C}}_n$
  are both contained in $\vec{J}_n$.

\begin{claim}
  \label{claim:I_n}
  It holds that
  $\vec{E}_n \subseteq \vec{J}_n$
  and
  $\overline{\vec{C}}_n \subseteq \vec{J}_n$.
\end{claim}
\begin{claimproof}
  As in the proof of \cref{lem:correctness-of-acceleration-via-iterable-vectors},
  we introduce for short the sets
  $\vec{F}_n, \vec{T}_n \subseteq \setZ^2$
  defined as follows:
  \begin{align*}
    \vec{F}_n
    & = \{\vec{y} - \vec{z}_n \mid \vec{y} \in \setN^2 \text{ and } \exists \vec{u} \in \vec{Q}_n : q_n(\vec{z}_n + \vec{u}) \step[*] q_n(\vec{y})\}
    \\
    \vec{T}_n
    & = \{\vec{z}_n - \vec{y} \mid \vec{y} \in \setN^2 \text{ and } \exists \vec{u} \in \vec{Q}_n : q_n(\vec{y} + \vec{u}) \step[*] q_n(\vec{z}_n)\}
  \end{align*}
  It is readily seen that $\vec{J}_n = \vec{F}_n \cup \vec{T}_n$.
  We prove a stronger statement than the claim,
  namely that
  $\vec{E}_n \subseteq \vec{F}_n$
  and
  $\overline{\vec{C}}_n \subseteq \vec{T}_n$.

  \smallskip

  We first show that $\vec{E}_n \subseteq \vec{F}_n$.
  Let $\vec{v} \in \vec{E}_n$.
  According to the definition of $\vec{E}_n$,
  there exist an elementary cycle $\theta$ of $\mathcal{V}_n$ with displacement $\vec{v}$
  and a node $s \in N$ with $s \xrightarrow{*} n$ such that
  $q_s$ is the start of $\theta$,
  $\vec{z}_s \geq (c, \ldots, c)$ and $s \neq n$ implies $\vec{v} \geq (0, 0)$.
  Since $\mathcal{V}_n$ is the instantiation of $\bvass$ with a finite subset of $\ReachSet{\bvass}$,
  we derive from \cref{fact:iteration-cycles-of-instantiation} that
  $q_s(\vec{z}_s) \step[*] q_s(\vec{z}_s + \vec{v})$.
  It follows that $\vec{v} \in \vec{F}_n$ if $s = n$.
  Otherwise,
  we have $\vec{v} \geq (0, 0)$,
  hence,
  $\vec{v} \in \clo{q_s}{\vec{z}_s}(\{(0, 0)\})$.
  Observe that
  $\clo{q_s}{\vec{z}_s}(\{(0, 0)\}) \subseteq \vec{Q}_s \subseteq \vec{Q}_n \subseteq \vec{F}_n$.
  This entails that $\vec{v} \in \vec{F}_n$.

  \smallskip

  We now show that $\overline{\vec{C}}_n \subseteq \vec{T}_n$.
  Let $\vec{v} \in \overline{\vec{C}}_n$.
  According to the definition of $\overline{\vec{C}}_n$,
  there exists a node $s \in N$ with $s \xrightarrow{*} n$ such that
  $q_s = q_n$ and $\vec{v} = \vec{z}_n - \vec{z}_s$.
  If $s = n$ then $\vec{v} = (0, 0) \in \vec{T}_n$.
  Otherwise,
  there exists $t \in N$ such that $s \xrightarrow{*} t \rightarrow n$.
  We obtain from \cref{cor:valid-explorations-post+},
  whose requirement is ensured by the induction hypothesis,
  that
  $q_n(\vec{z}_n)$ is in $\relimg{\step[*]}{q_s(\vec{z}_s + \vec{P}_t)}$.
  Recall that $\vec{P}_t \subseteq \vec{Q}_t$
  by the induction hypothesis.
  As $\vec{Q}_t \subseteq \vec{Q}_n$,
  we get that
  $q_s(\vec{z}_s + \vec{u}) \step[*] q_n(\vec{z}_n)$ for some $\vec{u} \in \vec{Q}_n$.
  Since $q_s = q_n$ and $\vec{v} = \vec{z}_n - \vec{z}_s$,
  we obtain that
  $\vec{v} \in \vec{T}_n$.
\end{claimproof}

  By definition,
  $\vec{I}_n \subseteq (\vec{E}_n \cup \overline{\vec{C}}_n)$,
  hence,
  $\vec{I}_n \subseteq \vec{J}_n$ according to \cref{claim:I_n}.
  It follows that
  $\accel{\vec{I}_n}(\vec{Q}_n) \subseteq \accel{\vec{J}_n}(\vec{Q}_n)$.
  We derive from \cref{lem:correctness-of-acceleration-via-iterable-vectors}
  that
  $\accel{\vec{I}_n}(\vec{Q}_n) \subseteq \clo{q_n}{\vec{z}_n}(\vec{Q}_n) = \vec{Q}_n$
  since $\clo{q_n}{\vec{z}_n}$ is idempotent.
  Recall that $\vec{P}_m \subseteq \vec{Q}_m \subseteq \vec{Q}_n$
  for every $m \in N$ with $m \rightarrow n$.
  We obtain from Condition~\ref{explo:P} of \cref{def:algorithmic-exploration} that
  $
  \vec{P}_n
  =
  \accel{\vec{I}_n}(\sum_{m \rightarrow n} \vec{P}_m)
  \subseteq
  \accel{\vec{I}_n}(\vec{Q}_n)
  \subseteq
  \vec{Q}_n
  $.
  This concludes the proof that
  for all $n \in N$,
  we have
  $\vec{P}_n \subseteq \vec{Q}_n$ and $q_n(\vec{z}_n + \vec{Q}_n) \subseteq \ReachSet{\bvass}$.
  This property immediately entails that $\explo$ is sound.
\end{proof}

\subsection{Proofs of \cref{subsec:partial-correctness-Explore}}
We introduce a notion of weakly-algorithmic explorations to simplify the presentation of the following lemma.
An exploration $\explo = (N, \rightarrow, \lambda)$ of a $2$-BVASS $\bvass = (Q, \Delta)$
is \emph{weakly-algorithmic} if it satisfies
Conditions~\ref{explo:v+q} and~\ref{explo:z} of \cref{def:algorithmic-exploration}.

\begin{lemma}
  \label{lem:algo-properties}
  For every $j \in J$,
  $\explo^j$ is a non-redundant weakly-algorithmic exploration of $\bvass$.
  Moreover,
  it holds that $\Redundant^j \subseteq N^j$, $\WorkList^j \subseteq N^j$ and
  the following properties are satisfied:
  \begin{itemize}
  \item
    $\Redundant^j \cap \WorkList^j = \emptyset$ and
    $\Redundant^j$ is the set of redundant nodes of $\explo^j$
    that are not in $\WorkList^j$,
  \item
    For all $n \in (\Redundant^j \cup \WorkList^j)$,
    $n$ is a leaf of $\explo^j$, and
  \item
    For all $n \in N^j$,
    $\vec{P}_n^j = \sum_{m \rightarrow^j n} \vec{P}_m^j$ if $n \in \WorkList^j$, and
    $\vec{P}_n^j = \accel{\vec{I}_n^j}(\sum_{m \rightarrow^j n} \vec{P}_m^j)$ otherwise.
  \end{itemize}
\end{lemma}
\begin{proof}
  We show the lemma by induction on $j \in J$.
  For the base case,
  observe that
  the values of the variables after the \textbf{foreach}-loop at lines~\ref{line:init-start}--\ref{line:init-end}
  satisfy
  $N^0 = \WorkList^0$, $\Redundant^0 = {\rightarrow}^0 = \emptyset$, and
  for every $n \in N^0$,
  there exists $(S, \vec{a}, q) \in \Delta$ with $S = \emptyset$ and $\vec{a} \geq (0, 0)$ such that
  $\lambda^0(n) = (\vec{a}, q, \vec{a}, \{(0, 0)\})$.
  We derive that $\explo^0$ is a non-redundant weakly-algorithmic exploration of $\bvass$.
  Moreover,
  $\Redundant^0 \subseteq N^0$, $\WorkList^0 \subseteq N^0$,
  and the three lemma properties are satisfied.
  To prove the induction step,
  suppose that the lemma holds for a given $j \in J$.
  We also assume that $(j+1) \in J$ as we are done otherwise.
  According to the body of the \textbf{while}-loop
  (see lines~\ref{line:while-start}--\ref{line:while-end}),
  there exists $n \in \WorkList^j$ such that $n \not\in \WorkList^{j+1}$.
  Put differently,
  $n$ is the node that is processed at the $(j+1)^{\mathrm{th}}$ iteration of the \textbf{while}-loop.
  Let $N^\sharp$, $\rightarrow^\sharp$, $\lambda^\sharp$, $\Redundant^\sharp$ and $\WorkList^\sharp$ denote
  the values of the variables at line~\ref{line:cover-check},
  just before the evaluation of the \textbf{if}-statement condition, and
  let $\explo^\sharp$ denote the labeled graph $(N^\sharp, \rightarrow^\sharp, \lambda^\sharp)$.
  Notice that $N^\sharp = N^j$, ${\rightarrow}^\sharp = {\rightarrow}^j$,
  and that the functions $\lambda^\sharp$ and $\lambda^j$ are identical except that
  $\vec{P}_n^\sharp = \accel{\vec{I}_n^j}(\vec{P}_n^j)$.
  So $\explo^\sharp$ is a non-redundant weakly-algorithmic exploration of $\bvass$ as well.
  Notice also that
  $\Redundant^\sharp = \Redundant^j$ and $\WorkList^\sharp = \WorkList^j \setminus \{n\}$.
  There are two cases to consider,
  depending on whether the \textbf{if}-statement condition at line~\ref{line:cover-check} holds or not.
  \begin{itemize}
  \item
    If $n$ is redundant in $\explo^\sharp$ then
    line~\ref{line:redundant-node} is executed and we get that the equalities
    $\explo^{j+1} = \explo^\sharp$,
    $\Redundant^{j+1} = \Redundant^\sharp \uplus \{n\}$, and
    $\WorkList^{j+1} = \WorkList^\sharp$
    are satisfied.\footnote{%
      \label{footnote:uplus}
      To improve readability,
      we use the symbol $\uplus$ to denote the union of two disjoint sets.
    }
    It follows,
    by a routine check,
    that the lemma holds for $j+1$ in this case.
  \item
    Otherwise,
    the \textbf{foreach}-loop at lines~\ref{line:expansion-start}--\ref{line:expansion-end} is executed.
    We derive that
    there exists a set $N'$ and
    a binary relation ${\rightarrow}' \subseteq (N^\sharp \setminus (\Redundant^\sharp \cup \WorkList^\sharp)) \times N'$ such that
    $N^{j+1} = N^\sharp \uplus N'$,
    ${\rightarrow}^{j+1} = {\rightarrow}^\sharp \uplus {\rightarrow}'$,
    $\Redundant^{j+1} = \Redundant^\sharp$, and
    $\WorkList^{j+1} = \WorkList^\sharp \uplus N'$.
    The functions $\lambda^{j+1}$ and $\lambda^\sharp$ coincide on $N^\sharp$.
    Moreover,
    for every $n' \in N'$,
    there exists a subset $M$ of $N^\sharp$, a transition rule $(S, \vec{a}, q) \in \Delta$ and
    a vector $\vec{b} \in \setN^2$ such that
    $M \cap (\Redundant^\sharp \cup \WorkList^\sharp) = \emptyset$,
    $S = \multiset{q_m^\sharp \mid m \in M}$,
    $\vec{b} \in (\vec{a} + \sum_{m \in M} (\vec{z}_m^\sharp + \vec{P}_m^\sharp))$,
    $M = \{m \in N^\sharp \mid m \rightarrow' n'\}$, and
    $\lambda^{j+1}(n') = (\vec{a}, q, \vec{b}, \sum_{m \in M} \vec{P}_m^\sharp)$.
    It again follows,
    by a routine check,
    that the lemma holds for $j+1$ in this case.
  \end{itemize}
  A crucial observation in both cases is that
  $\vec{I}_m^j = \vec{I}_m^\sharp = \vec{I}_m^{j+1}$ for every $m \in N^j$.
\end{proof}

We recall that,
for every $j \in J$,
$\hat{\explo}^j$ is the restriction of $\explo^j$ to
the set of processed nodes $\hat{N}^j = N^j \setminus \WorkList^j$.

\begin{lemma}
  \label{lem:algo-completeness}
  It holds that
  $\post[\bvass](\ReachSet{\hat{\explo}^j}) \subseteq \ReachSet{\explo^j}$
  for every $j \in J$.
\end{lemma}
\begin{proof}
  For short,
  let us write
  $C^j = \ReachSet{\hat{\explo}^j}$ and
  $U^j = N^j \setminus (\Redundant^j \cup \WorkList^j)$,
  for each $j \in J$.
  Recall that
  $C^j = \bigcup_{n \in (N^j \setminus \WorkList^j)} q_n^j(\vec{z}_n^j + \vec{P}_n^j)$.
  By \cref{lem:algo-properties},
  $\Redundant^j$ only contains redundant nodes of $\explo^j$,
  hence,
  $C^j = \bigcup_{n \in U^j} q_n^j(\vec{z}_n^j + \vec{P}_n^j)$.
  Let us show by induction on $j \in J$ that $\post[\bvass](C^j) \subseteq \ReachSet{\explo^j}$.
  The base case follows from the \textbf{foreach}-loop at lines~\ref{line:init-start}--\ref{line:init-end}.
  Indeed,
  for every $(S, \vec{a}, q) \in \Delta$ with $S = \emptyset$ and $\vec{a} \geq (0, 0)$,
  there exists $n \in N^0$ such that
  $\lambda^0(n) = (\vec{a}, q, \vec{a}, \{(0, 0)\})$.
  Moreover,
  $C^0 = \emptyset$ since $N^0 = \WorkList^0$.
  So
  $\post[\bvass](C^0)
  =
  \post[\bvass](\emptyset)
  =
  \{q(\vec{a}) \mid (\emptyset, \vec{a}, q) \in \Delta \text{ and } \vec{a} \geq (0, 0)\}$
  is contained in $\ReachSet{\explo^0}$.
  To prove the induction step,
  suppose that $\post[\bvass](C^j) \subseteq \ReachSet{\explo^j}$ for a given $j \in J$.
  We also assume that $(j+1) \in J$ as we are done otherwise.
  As in the proof of \cref{lem:algo-properties},
  let $n \in \WorkList^j$ such that $n \not\in \WorkList^{j+1}$, and
  let $N^\sharp$, $\rightarrow^\sharp$, $\lambda^\sharp$, $\Redundant^\sharp$ and $\WorkList^\sharp$ denote
  the values of the variables at line~\ref{line:cover-check}.
  Recall that $N^\sharp = N^j$ and $\WorkList^\sharp = \WorkList^j \setminus \{n\}$.
  We will use the following fact,
  which easily follows from the proof of \cref{lem:algo-properties}.
  The exploration $\explo^\sharp$ of $\bvass$ is defined by
  $\explo^\sharp = (N^\sharp, \rightarrow^\sharp, \lambda^\sharp)$,
  as in that proof.
  \begin{fact}
    \label{fact:proof-lem:algo-completeness}
    We have $\ReachSet{\explo^j} \subseteq \ReachSet{\explo^\sharp} \subseteq \ReachSet{\explo^{j+1}}$.
    Moreover,
    if $U^j \neq U^{j+1}$ then
    $n$ is not redundant in $\explo^\sharp$ and the equalities
    $U^{j+1} = U^j \uplus \{n\} = N^\sharp \setminus (\Redundant^\sharp \cup \WorkList^\sharp)$
    are satisfied.\footnotemark[\getrefnumber{footnote:uplus}]
  \end{fact}
  We also recall from the proof of \cref{lem:algo-properties} that
  the functions $\lambda^\sharp$ and $\lambda^j$ are identical except that
  $\vec{P}_n^\sharp = \accel{\vec{I}_n^j}(\vec{P}_n^j)$, and that
  the functions $\lambda^{j+1}$ and $\lambda^\sharp$ coincide on $N^\sharp$.
  Let $q(\vec{y}) \in \post[\bvass](C^{j+1})$ and
  let us show that $q(\vec{y}) \in \ReachSet{\explo^{j+1}}$.
  By definition of $\post[\bvass]$,
  there exists
  a transition rule $\delta = (S, \vec{a}, q)$ and
  a subset $D$ of $C^{j+1}$
  such that
  $S = \multiset{r \mid r(\vec{z}) \in D}$
  and
  $\vec{y} = \vec{a} + \sum_{r(\vec{z}) \in D} \vec{z}$.
  If $D \subseteq C^j$ then $q(\vec{y}) \in \post[\bvass](C^j)$,
  and we are done since
  $\post[\bvass](C^j) \subseteq \ReachSet{\explo^j} \subseteq \ReachSet{\explo^{j+1}}$.
  Assume, for the remainder of the proof,
  that $D \not\subseteq C^j$.
  As the multiset $\multiset{r \mid r(\vec{z}) \in D}$ is a set,
  each configuration in $D$ comes from a distinct node in
  $U^{j+1}$.
  We derive that there exists a subset $M$ of $U^{j+1}$,
  with $M \not\subseteq U^j$,
  such that
  $S = \multiset{q_m^{j+1} \mid m \in M}$ and
  $\vec{y} \in (\vec{a} + \sum_{m \in M} (\vec{z}_m^{j+1} + \vec{P}_m^{j+1}))$.
  Since $M \subseteq U^{j+1}$ and $M \not\subseteq U^j$,
  we derive from \cref{fact:proof-lem:algo-completeness} that
  $n$ is not redundant in $\explo^\sharp$,
  $n \in M$, and
  $M$ is contained in $N^\sharp \setminus (\Redundant^\sharp \cup \WorkList^\sharp)$.
  Observe also that $S = \multiset{q_m^\sharp \mid m \in M}$.
  So the \textbf{foreach}-loop at lines~\ref{line:expansion-start}--\ref{line:expansion-end} is executed and
  the lines~\ref{line:z'-P'}--\ref{line:children-end} are executed for $M$ and $(S, \vec{a}, q)$.
  We have
  $(\vec{B} + \sum_{m \in M} \vec{P}_m^\sharp) = (\vec{a} + \sum_{m \in M} (\vec{z}_m^\sharp + \vec{P}_m^\sharp)) \cap \setN^2$,
  where $\vec{B}$ is the subset of $\setN^2$ taken at line~\ref{line:basis-computation}.
  Hence,
  $\vec{y} \in (\vec{b} + \sum_{m \in M} \vec{P}_m^\sharp)$
  for some $\vec{b} \in \vec{B}$.
  The \textbf{foreach}-loop at lines~\ref{line:children-start}--\ref{line:children-end} guarantees that
  there exists $n' \in N^{j+1}$ with
  $\lambda^{j+1}(n') = (\vec{a}, q, \vec{b}, \sum_{m \in M} \vec{P}_m^\sharp)$.
  We obtain that $q(\vec{y}) \in \ReachSet{\explo^{j+1}}$.
\end{proof}

\lemAlgoExplorations*
\begin{proof}
  The lemma is an immediate consequence of \cref{lem:algo-properties,lem:algo-completeness}.
\end{proof}

\lemAlgoPartialCorrectness*
\begin{proof}
  Let $\sigma$ be an execution of $\mathtt{Explore}(\bvass)$,
  where $\bvass$ is a $2$-BVASS.
  We first show that $\hat{\explo}^\sigma$ is a non-redundant and algorithmic exploration of $\bvass$.
  This will derive from the following claim.

  \begin{claim}
    \label{claim:limit-exploration}
    For every $j \in J$, it holds that
    \begin{enumerate}
    \item
      for every $m, n \in N^\sigma$,
      if $n \in N^j$ and $m \rightarrow^\sigma n$ then $m \in N^j$ and $m \rightarrow^j n$, and
    \item
      the restriction of $\hat{\explo}^\sigma$ to $\hat{N}^j$ coincides with $\hat{\explo}^j$.
    \end{enumerate}
  \end{claim}
  \begin{claimproof}
    Let $j \in J$.
    Recall from the proof of \cref{lem:algo-properties} that
    if $(j+1) \in J$ then
    ${\rightarrow}^{j+1} \setminus {\rightarrow}^j$ is contained in
    $N^j \times (N^{j+1} \setminus N^j)$.
    An immediate induction on $\ell$ shows that,
    for every $\ell \in J$,
    if $\ell > j$ then
    ${\rightarrow}^\ell \setminus {\rightarrow}^j$ is contained in
    $N^{\ell - 1} \times (N^\ell \setminus N^j)$.
    Now let $m, n \in N^\sigma$ such that $n \in N^j$ and $m \rightarrow^\sigma n$.
    We have $m \rightarrow^\ell n$ for some $\ell \in J$.
    Assume by contradiction that $(m, n)$ is in ${\rightarrow}^\ell \setminus {\rightarrow}^j$.
    We have $\ell > j$ since ${\rightarrow}^\ell \subseteq {\rightarrow}^j$ otherwise.
    Hence,
    $(m, n)$ is in $N^{\ell - 1} \times (N^\ell \setminus N^j)$,
    which contradicts our assumption that $n \in N^j$.
    So $m \rightarrow^j n$ holds,
    hence,
    $m \in N^j$.
    This concludes the proof of the first assertion of the claim.

    \smallskip

    To prove the second assertion of the claim,
    we show that the node-labeled graphs $\explo^j$ and $\explo^\sigma$ have the same restriction on $\hat{N}^j = N^j \setminus \WorkList^j$.
    As an immediate consequence of the first assertion,
    we get that
    the restriction of the graph $(N^\sigma, \rightarrow^\sigma)$ on $N^j$ coincides with $(N^j, \rightarrow^j)$.
    Let us now prove that $\lambda^j$ and $\lambda^\sigma$ coincide on $\hat{N}^j$.
    Recall from the proof of \cref{lem:algo-properties} that
    if $(j+1) \in J$ then
    $\WorkList^{j+1} \subseteq \WorkList^j \cup (N^{j+1} \setminus N^j)$ and
    the functions $\lambda^j$ and $\lambda^{j+1}$ coincide on $\hat{N}^j$.
    An immediate induction on $\ell$ shows that,
    for every $\ell \in J$,
    if $\ell > j$ then
    $\WorkList^\ell \subseteq \WorkList^j \cup (N^\ell \setminus N^j)$ and
    the functions $\lambda^j$ and $\lambda^\ell$ coincide on $\hat{N}^j$.
    We obtain that
    $\lambda^j$ and $\lambda^\sigma$ coincide on $\hat{N}^j$,
    which concludes the proof of the second assertion of the claim.
  \end{claimproof}

  The first property of \cref{claim:limit-exploration} entails that
  the edge relation $\rightarrow^\sigma$ of $\explo^\sigma$ is well-founded.
  Indeed,
  if $n_0, n_1, \ldots$ is an infinite sequence of nodes in $N^\sigma$ such that
  $\cdots \rightarrow^\sigma n_1 \rightarrow^\sigma n_0$,
  then $n_0 \in N^j$ for some $j \in J$ and we get that
  $\cdots \rightarrow^j n_1 \rightarrow^j n_0$,
  which is impossible since $\rightarrow^j$ is well-founded.
  Note that the well-foundedness of $\rightarrow^\sigma$ entails that $\explo^\sigma$ is acyclic.
  We obtain that $\explo^\sigma$ is an exploration of $\bvass$,
  hence,
  so is $\hat{\explo}^\sigma$.
  The second property of \cref{claim:limit-exploration},
  combined with \cref{lem:algo-explorations},
  entails that $\hat{\explo}^\sigma$ is a non-redundant and algorithmic exploration of $\bvass$.

  \smallskip

  If $\sigma$ terminates then it returns $\explo^\kappa$.
  Moreover,
  we have $\explo^\kappa = \hat{\explo}^\kappa = \explo^\sigma = \hat{\explo}^\sigma$
  in that case since $\hat{N}^\kappa = N^\kappa$.
  We derive from \cref{lem:algo-explorations} that
  $\post[\bvass](\ReachSet{\explo^\kappa}) \subseteq \ReachSet{\explo^\kappa}$,
  hence,
  $\explo^\kappa$ is complete.
  It follows that $\hat{\explo}^\sigma$ is finite and complete.
  If $\sigma$ does not terminate then
  $\hat{N}^\sigma$ is infinite since
  $\hat{N}^j \subset \hat{N}^{j+1}$ for all $j \in \setN$.
\end{proof}

\section{Proofs of \cref{sec:branches}}
We show that explorations $\explo^\sigma$ built by $\mathtt{Explore}(\bvass)$ (see \cref{subsec:partial-correctness-Explore}) admit a \emph{spanning forest} by introducing the class of well-connected graphs defined as follows. A graph $\mathcal{G}=(N,\rightarrow)$ is said to be \emph{well-connected} if there exists a sequence $(W_i)_{i\in\setN}$ of finite subsets $W_i\subseteq N$ such that the following conditions hold for every $n\in N$:
\begin{itemize}
\item $\{i\in\setN \mid n\in W_i\}$ is a non-empty interval (finite or infinite).
\item If $n\in W_{i+1}\setminus W_i$ for some $i$ then there exists $m\in W_i\setminus W_{i+1}$ such that $m\rightarrow n$.
\end{itemize}
Step-by-step algorithms computing a graph thanks to worklists will produce in general well-connected graphs. Intuitively, $W_i$ corresponds to the set of nodes in the worklist at step $i$ of the computation. The exploration $\explo^\sigma$ built by $\mathtt{Explore}(\bvass)$ (see \cref{subsec:partial-correctness-Explore}) is clearly well-connected, and so is $\hat{\explo}^\sigma$ as it is the restriction of $\explo^\sigma$ to an ancestor-closed subset of nodes.

\begin{lemma}\label{lem:caragoodforbranch}
  A graph is well-connected iff it is spannable.
\end{lemma}
\begin{proof}
  Assume first that $\mathcal{G}=(N,\rightarrow)$ is a graph that admits a spanning forest $\mathcal{F}=(N,\rightarrow_{\mathcal{F}})$. We introduce the set $W_0$ defined as the set of source nodes of $\mathcal{F}$. We also define inductively on $i\in\setN$ the set $W_{i+1}$ of nodes $n\in N\setminus (\bigcup_{j=0}^iW_j)$ such that there exists a node $m\in W_i$ satisfying $m\rightarrow_{\mathcal{F}}n$. By induction on $i$ we deduce that $W_i$ is finite for every $i$. Observe also that for every $n\in N$, the set $\{i\in\setN \mid n\in W_i\}$ is a singleton. Moreover, for every $n\in N$, if $n\in W_{i+1}$, there exists $m\in W_i$ such that $m\rightarrow n$. Since $W_i$ and $W_{i+1}$ are disjoint, we deduce that $W_{i+1}=W_{i+1}\setminus W_i$ and $W_i=W_i\setminus W_{i+1}$. We have proved that $\mathcal{G}$ is well-connected.

  \newcommand{\predname}{\operatorname{pred}}
  \newcommand{\pred}[1]{\predname(#1)}
  Now, let us assume that $\mathcal{G}=(N,\rightarrow)$ is well-connected. We denote by $(W_i)_{i\in\setN}$ a sequence of finite sets of nodes witnessing that property. Since $\{i\in\setN \mid n\in W_i\}$ is non-empty, it contains a minimal element denoted by $\tau(n)$. Notice that $\tau(n)=0$ iff $n\in W_0$. For every $n\in N\setminus W_0$, we have $n\in W_{\tau(n)}\setminus W_{\tau(n)-1}$. It follows that there exists $m\in W_{\tau(n)-1}\setminus W_{\tau(n)}$ such that $m\rightarrow n$. It follow that we can define a function $\predname:N\setminus W_0\rightarrow N$ such that for every $n\in N\setminus W_0$, the node $m=\pred{n}$ is such that $m\rightarrow n$ and $m\in  W_{\tau(n)-1}\setminus W_{\tau(n)}$. We introduce the graph $\mathcal{F}=(N,\rightarrow_{\mathcal{F}})$ defined by $m\rightarrow_{\mathcal{F}} n$ if $n\in N\setminus W_0$ and $m=\pred{n}$. Let us prove that this graph is a spanning forest of $\mathcal{G}$. Since $\rightarrow_{\mathcal{F}}$ is included in $\rightarrow$, it is sufficient to prove that $\mathcal{F}$ is a finitely-branching forest. Notice that the source nodes of $\mathcal{F}$ is $W_0$ which is a finite set. By induction on $i$, we have that for every $n\in N$ such that $\tau(n)=i$, we have $n\in \des[\mathcal{F}]{W_0}$. Hence $N=\des[\mathcal{F}]{W_0}$. As $m\rightarrow_{\mathcal{F}}n$ implies $m\in  W_{\tau(n)-1}\setminus W_{\tau(n)}$, it follows that $\tau(m)<\tau(n)$. In particular $\rightarrow_{\mathcal{F}}$ is anti-symmetric. Notice that $\{m\in N \mid m\rightarrow_{\mathcal{F}} n\}$ contains at most one node. Finally, let us prove that $\{n\in N \mid m\rightarrow_{\mathcal{F}}n\}$ is finite for every $m\in N$. If the set is empty, we are done. Otherwise, there exists $n\in N$ such that $m\rightarrow_{\mathcal{F}}n$. It follows that $m\in W_{\tau(n)-1}\setminus W_{\tau(n)}$. In particular, there exists a maximal $i_m\in\setN$ such that $m\in W_{i_m}$. Moreover, $\tau(n)-1=i_m$, i.e. $\tau(n)=i_m+1$. Hence $n\in W_{i_m+1}$. We have proved that $\{n\in N \mid m\rightarrow_{\mathcal{F}}n\}$ is included in the finite set $W_{i_m+1}$. Hence, in any case this set is finite. It follows that $\mathcal{F}$ is a finitely-branching forest.
\end{proof}

Now, let $\mathcal{F}$ be a spanning forest of a graph $\mathcal{G}$. A branch $\beta$ of $\mathcal{F}$ is said to be \emph{primary} if the set $\anc{\beta}\setminus \des{n}$ is finite for every node $n\in \beta$.
\begin{lemma}\label{lem:primary}
  A branch of $\mathcal{F}$ is minimal for $\sqsubseteq$ if, and only if, it is primary.
\end{lemma}
\begin{proof}
  Assume first that $\beta$ is a branch that is not primary and let us prove that $\beta$ is not minimal for $\sqsubseteq$. There exists a node $n\in \beta$ such that $N'=\anc{\beta}\setminus \des{n}$ is infinite. We introduce the finitely-branching forest $\mathcal{F}'=(N',\rightarrow')$ where $\rightarrow'$ is the restriction of $\rightarrow_{\mathcal{F}}$ on $N'$. From König's lemma it contains a branch $\alpha$. Since $\alpha\subseteq \anc{\beta}\setminus \des{n}$ we deduce that $\alpha\sqsubseteq \beta$ and as $n\not\in \anc{\alpha}$, we deduce that $\neg(\beta\sqsubseteq \alpha)$. Therefore $\beta$ is not minimal for $\sqsubseteq$.

  Now, assume that $\beta$ is a primary branch and let us prove that it is minimal for $\sqsubseteq$. Let us consider a branch $\alpha$ such that $\alpha\sqsubseteq \beta$. Let $n\in\beta$. Since $\beta$ is primary, the set $\anc{\beta}\setminus \des{n}$ is finite. Hence, there exists $k\in\setN$ such that $\alpha_k$ is not in that set. Since $\alpha_k\in \anc{\beta}$ we deduce that $\alpha_k\in\des{n}$. Hence $n\in\anc{\alpha}$. We have proved that $\beta\subseteq \anc{\alpha}$ which entails $\beta\sqsubseteq \alpha$. Hence $\beta$ is minimal for $\sqsubseteq$.
\end{proof}

\medskip

The following lemma prove the existence of a minimal branch.
\begin{lemma}\label{lem:primaryexists}
  There exists a minimal branch. 
\end{lemma}
\begin{proof}
  Assume by contradiction that there does not exist a minimal branch. Given $k\in\setN\cup\{\omega\}$, we introduce the binary relation $\sqsubseteq_k$ over the branches defined by $\alpha\sqsubseteq_k \beta$ if $\alpha\sqsubseteq \beta$ and $k=\sup\{\ell\in\setN \mid \beta_\ell\in \anc{\alpha}\}$. For every branch $\beta$, there exists a branch $\alpha$ such that $\alpha\sqsubseteq\beta$ and $\neg(\beta\sqsubseteq\alpha)$. It follows that for every branch $\beta$, there exists $k\in\setN$ and a branch $\alpha$ such that $\alpha\sqsubseteq_k \beta$. We say that $\alpha$ is \emph{eagerly-smaller} than $\beta$ if this natural number $k$ is the minimal possible one satisfying the previous property.

  From the König's lemma, there exists a branch $\alpha^0$. From the previous paragraph, we can build an infinite sequence $\alpha^0,\alpha^1,\ldots$ of branches such that $\alpha^{\ell+1}$ is eagerly-smaller than $\alpha^\ell$ for every $\ell\in\setN$. We denote by $k_\ell$ the natural number in $\setN$ such that $\alpha^{\ell+1}\sqsubseteq_{k_\ell} \alpha^\ell$.
  
  We introduce the sequence $(N_i)_{i\in\setN}$ of finite sets of nodes at \emph{depth} $i$ in the spanning forest $\mathcal{F}$. Those sets are formally defined by $N_0$ is the set of sources of $\mathcal{F}$, and inductively by $N_{i+1}=\{m\in N \mid \exists n\in N_i,\,n\rightarrow_{\mathcal{F}} m\}$ for every $i\in\setN$. Observe that $\alpha^\ell_i\in N_i$ for every $\ell\in\setN$. We denote by $N^\infty_i$ the set of nodes $n\in N_i$ such that $\{\ell\in\setN \mid \alpha^\ell_i=n\}$ is infinite. We let $N^\infty=\bigcup_{i\in \setN}N^\infty_i$ and we introduce $\rightarrow^\infty$ defined as the restriction of $\rightarrow_{\mathcal{F}}$ to $N^\infty$. Since $(N^\infty,\rightarrow^\infty)$ is an infinite finitely-branching forest, König's lemma shows that it contains a branch $\beta$.
  
   Let us prove that $\beta\sqsubseteq \alpha^\ell$ for every $\ell$. For every $i\in\setN$, since $\beta_i\in N^\infty_i$, there exists $\ell'>\ell$ such that $\alpha^{\ell'}_i=\beta_i$. Hence $\beta_i\in \anc{\alpha^{\ell'}}$. As $\alpha^{\ell'}\sqsubseteq \alpha^\ell$, we deduce that $\beta_i\in \anc{\alpha^{\ell}}$ and we have proved that $\beta\sqsubseteq \alpha^\ell$.

   Let us prove that $\lim_{\ell\rightarrow+\infty}k_\ell=+\infty$. Let $k\in\setN$. As $N_k$ is finite, we deduce that there exists $\ell_k\in\setN$ such that $\beta^\ell_k\in N^\infty_k$ for every $\ell\geq \ell_k$. Now, let $\ell\geq \ell_k$. There exists $\ell'\geq \ell+1$ such that $\alpha^{\ell'}_k=\alpha^\ell_k$. From $\alpha^{\ell'} \sqsubseteq\alpha^{\ell+1}$ we deduce that $\alpha^\ell_k\in \anc{\alpha^{\ell+1}}$. By definition of $k_\ell$, we deduce that $k_\ell> k$. We have proved that for every $k\in\setN$, there exists $\ell_k\in\setN$ such that $k_\ell> k$ for every $\ell\geq \ell_k$. Hence, $\lim_{\ell\rightarrow+\infty}k_\ell=+\infty$.
  
   Since $\beta$ is a branch and there does not exists minimal branches, it follows that there exits a branch $\alpha$ such that $\alpha\sqsubseteq_k \beta$ for some $k\in\setN$. Since $\lim_{\ell\rightarrow+\infty}k_\ell=+\infty$ and $\alpha_k^\ell=\beta_k$ for infinitely many $\ell$, there exists $\ell\in\setN$ such that $k_\ell> k$ and such that $\alpha_k^\ell=\beta_k$. By minimality of $k_\ell$, we deduce that $\alpha_{k+1}^\ell\in \anc{\alpha}$. As $\alpha\sqsubseteq\beta \sqsubseteq \alpha^\ell$, we get $\alpha_{k+1}^\ell\in \anc{\beta}$, and we get a contradiction with $\alpha\sqsubseteq_k \beta$.
   
   It follows that there exists a primary branch.
\end{proof}

\lemWellExtraction*
\begin{proof}
  \cref{lem:primaryexists} shows that there exists a minimal branch $\beta$. From \cref{lem:primary}, the branch is primary. Let $X=\anc{\beta}$. Notice that $X$ is an infinite ancestor-closed set of nodes. Let $\mathcal{X}=(X,\rightarrow_X)$ denote the restriction of $\mathcal{G}$ to $X$. Recall that ${\rightarrow}_X = {\rightarrow} \cap (X \times X)$ by definition. Let us prove that $\mathcal{X}$ is primary and directed. Since $X$ is the set of ancestors of a branch, it follows that $\mathcal{X}$ is directed. To prove that $\mathcal{X}$ is primary, let $x\in X$ and let us prove that $X\setminus\des[\mathcal{X}]{x}$ is finite. Since $x\in X$, there exists $n\in\beta$ such that $x\xrightarrow{*}n$. Now, let $y\in X\setminus \des[\mathcal{X}]{x}$ and assume by contradiction that $y\in \des{n}$. In follows that $n\xrightarrow{*}y$ and in particular $x\xrightarrow{*}y$. Since $y\in X$ and $X$ is ancestor-closed, we deduce that $x\xrightarrow{*}_X y$. Therefore $y\in \des[\mathcal{X}]{x}$ and we get a contradiction. Hence $y\in N\setminus \des{n}$ and we have proved the inclusion $X\setminus \des[\mathcal{X}]{x}\subseteq N\setminus \des{n}$. As $\beta$ is primary we deduce that $N\setminus \des{n}$ is finite. Hence $X\setminus \des[\mathcal{X}]{x}$ is finite. Therefore $\mathcal{X}$ is primary. We have proved \cref{lem:well-extraction}.
\end{proof}

\section{Proofs of \cref{sec:termination}}
\subsection{Proofs of \cref{sec:modes}}
We first prove some additional results about the sets $\stab{\star}{\vec{C}}\setminus\setQ_{\geq 0}^2$ where $\vec{C}\subseteq\setQ_{\geq 0}^2$ is a cone.

\begin{lemma}\label{lem:inclusionstab}
 For every cones $\vec{X}, \vec{Y}\subseteq \setQ_{\geq 0}^2$ with the same set of axis, the inclusion $\vec{X}\subseteq \vec{Y}$ implies the inclusion $\stab{\star}{\vec{Y}}\setminus\setQ_{\geq 0}^2\subseteq\stab{\star}{\vec{X}}\setminus\setQ_{\geq 0}^2$.
\end{lemma}
\begin{proof}
   We denote by $\vec{U}$ the set of axis of $\vec{X},\vec{Y}$.
   Let $\vec{v}\in \stab{\star}{\vec{Y}}\setminus\setQ_{\geq 0}^2$. Let $\vec{x}\in\vec{X}$ and $\lambda\in\setQ_{\geq 0}$ such that $\vec{x}+\lambda\vec{v}\geq (0,0)$, and let us prove that $\vec{x}+\lambda\vec{v}\in\vec{X}$. If $\lambda=0$ the property is immediate. So, we can assume that $\lambda>0$. Since $\vec{v}\not\geq (0,0)$, there exists a maximal $\mu\in\setQ_{\geq 0}$ such that $\vec{x}+\mu\vec{v}\geq (0,0)$. We put $\vec{x}'=\vec{x}+\mu\vec{v}$. As $\vec{x}\in\vec{X}$ and $\vec{X}\subseteq \vec{Y}$, we deduce that $\vec{x}'\in\stab{\vec{v}}{\vec{Y}}=\vec{Y}$. Moreover, by maximality of $\mu$, it follows that at least one component of $\vec{x}'$ is zero. Hence $\vec{x}'\in\setQ_{\geq 0}\vec{U}$. In particular $\vec{x}'\in\vec{X}$. By maximality of $\mu$ we have $\lambda\leq \mu$. As $\lambda>0$, it follows that $\mu>0$. A simple computation shows that $(1-\frac{\lambda}{\mu})\vec{x}+\frac{\lambda}{\mu}\vec{x}'=\vec{x}+\lambda\vec{v}$. As $\vec{X}$ is a cone, we deduce that $\vec{x}+\lambda\vec{v}\in\vec{X}$. We have proved that $\vec{v}\in\stab{\star}{\vec{X}}$. We have proved the lemma.
\end{proof}

A cone $\vec{C} \subseteq \setQ_{\geq 0}^2$ is said to be \emph{degenerated} if there exists $\vec{c}\in\setQ_{\geq 0}^2$ such that $\vec{C}=\setQ_{\geq 0}\vec{c}$.
\begin{fact}
  \label{fact:nonaxestab}
   For every non degenerated cone $\vec{C} \subseteq \setQ_{\geq 0}^2$ with an empty set of axis, we have $\stab{\star}{\vec{C}}\setminus\setQ_{\geq 0}^2=\emptyset$.
\end{fact}
\begin{proof}
Let $\vec{S}$ be the stabilizer of $\vec{C}$. Since $\vec{C}$ is non degenerated, there exists $\vec{x}\in\vec{C}\setminus\{(0,0)\}$ and $\vec{y}\in \vec{C}\setminus \setQ_{\geq 0}\vec{x}$. As the set of axis of $\vec{C}$ is empty, it follows that $\vec{x},\vec{y}\in\setQ_{>0}^2$. Assume by contradiction that there exists $\vec{v}\in\vec{S}\setminus\setQ_{\geq 0}^2$. There exists maximal $\lambda,\mu\in\setQ_{\geq 0}$ such that $\vec{x}+\lambda\vec{v}\geq (0,0)$ and $\vec{y}+\mu\vec{v}\geq (0,0)$. Since $\vec{x},\vec{y}\in\setQ_{>0}^2$, it follows that $\lambda,\mu>0$. Observe that $\vec{x}+\lambda\vec{v}$ and $\vec{y}+\mu\vec{v}$ are in $\stab{\vec{v}}{\vec{C}}=\vec{C}$ and have a component equal to zero. As the set of axis $\vec{C}$ is empty, it follows that $\vec{x}+\lambda\vec{v}$ and $\vec{y}+\mu\vec{v}$ are equal to the zero vectors. Hence $\vec{y}=\frac{\mu}{\lambda}\vec{x}$ and we get a contradiction.
\end{proof}

\lemFiniteswithmodes*
\begin{proof}
  Let us consider a non-decreasing sequence $(\vec{C}_n)_{n\in\setN}$ of cones $\vec{C}_n\subseteq \setQ_{\geq 0}^2$. Let $\vec{M}_n$ be the $h$-mode of $\vec{C}_n$ and let us prove that there exists $N\in\setN$ such that $\vec{M}_n=\vec{M}_N$ for every $n\geq N$. Observe that if $\vec{C}_n$ is degenerated for every $n\in\setN$ then the property is immediate since in that case there exists $N\in\setN$ such that $\vec{C}_n=\vec{C}_N$ for every $n\geq N$. So, by considering a suffix of the sequence $(\vec{C}_n)_{n\in\setN}$, we can assume w.l.o.g. that $\vec{C}_n$ is non degenerated for every $n\in\setN$. By also considering a suffix of the sequence $(\vec{C}_n)_{n\in\setN}$, we can also assume w.l.o.g. that the set of axis of $\vec{C}_n$ does not depend on $n$. We denote by $\vec{U}$ this set.

  Observe that if $\vec{U}=\emptyset$, \cref{fact:nonaxestab} shows that $\stab{\star}{\vec{C}_n}=\vec{C}_n$. In particular $\vec{M}_n=\emptyset$ for every $n\in\setN$ and we are done. If $\vec{U}=\{(1,0),(0,1)\}$ then $\vec{C}_n=\setQ_{\geq 0}^2$ and in that case $\vec{M}_n=\setZ^2\setminus \setN^2$ for every $n\in\setN$ and we are done. So, we can assume that $\vec{U}$ contains exactly one vector. By symmetry, we can also assume that $\vec{U}=\{(1,0)\}$.

 Let $\vec{Q}_n=\{\vec{v}\in \vec{M}_n\mid \vec{v}(1)<0\wedge \vec{v}(2)\in\{-h,\ldots,0\}\}$ and let us prove that $\vec{M}_n=\vec{Q}_n\cup (\setN\times\{-h,\ldots,-1\})$. It is clear that $\vec{Q}_n\subseteq \vec{M}_n$. Let $\vec{v}\in \setN\times\{-h,\ldots,0\}$, and let $\vec{c}\in\vec{C}$ and $\lambda\in\setQ_{\geq 0}$ such that $\vec{c}+\lambda\vec{v}\geq (0,0)$. Let $\mu\in\setQ_{\geq 0}$ be maximal such that $\vec{c}+\mu\vec{v}\geq (0,0)$. As $\vec{v}(1)\geq 0$ and $\vec{v}(2)<0$, we deduce that $\vec{c}+\mu\vec{v}\in\setQ_{\geq 0}(1,0)$. In particular this vector is in $\vec{C}_n$. By maximality of $\mu$, it follows that $\mu\geq \lambda$. As $\vec{c}$ and $\vec{c}+\mu\vec{v}$ are both in the cone $\vec{C}_n$, and $0\leq \lambda\leq \mu$, it follows that $\vec{c}+\lambda\vec{v}\in\vec{C}_n$. We have proved that $\stab{\vec{v}}{\vec{C}}\subseteq\vec{C}$. Hence $\vec{v}\in\vec{M}_n$. We have proved the inclusion $\vec{Q}_n\cup (\setN\times\{-h,\ldots,0\})\subseteq \vec{M}_n$. For the converse inclusion, let $\vec{v}\in\vec{M}_n$. Observe that if $\vec{v}\not\in\vec{Q}_n$ then $\vec{v}(1)\geq 0$ or $\vec{v}(2)>0$. If $\vec{v}(2)> 0$ then $\vec{v}(1)<0$ and from \cref{fact:axis-in-stab} we derive a contradiction with $(0,1)\not\in \vec{U}$. Hence $\vec{v}(2)\leq 0$ and $\vec{v}(1)\geq 0$. Since $\vec{v}\not\geq (0,0)$, it follows that $\vec{v}(2)<0$. We have proved that $\vec{v}\in\setN\times\{-h,\ldots,-1\}$. We have proved the equality $\vec{M}_n=\vec{Q}_n\cup (\setN\times\{-h,\ldots,-1\})$.
 
  Let us prove that $\vec{Q}_n$ is finite. Since $\vec{C}_0$ is non-degenerated, it is not included in $\setQ_{\geq 0}(1,0)$. Hence, there exists $\vec{c}\in\vec{C}_0$ such that $\vec{c}(2)>0$. We are going to prove that $\vec{Q}_n\subseteq \{\vec{v}\in\setZ^2 \mid -\frac{\vec{c}(1)}{\vec{c}(2)}h\leq \vec{v}(1)<0\wedge -h\leq \vec{v}(2)\leq 0\}$. Let $\vec{v}\in\vec{Q}_n$. There exists a maximal $\lambda\in\setQ_{\geq 0}$ such that $\vec{c}+\lambda\vec{v}\geq (0,0)$. Since this vector is in $\stab{\vec{v}}{\vec{C}_n}=\vec{C}_n$, we deduce that $\vec{c}+\lambda\vec{v}$ is in $\vec{C}_n$. Moreover, as at least one component of that vector is zero, we deduce that this vector is in $\setQ_{\geq 0}(1,0)$, i.e. $\vec{c}(2)+\lambda\vec{v}(2)=0$. As $\vec{c}(2)>0$, we deduce that $\lambda>0$ and $\vec{v}(2)<0$. In particular $\vec{v}(2)\in\{-h,\ldots,-1\}$. Moreover, $\lambda=\vec{c}(2){-\vec{v}(2)}$. As $\vec{c}(1)+\lambda\vec{v}(1)\geq 0$, we deduce that $\vec{v}(1)\geq \frac{\vec{c}(1)}{\vec{c}(2)}\vec{v}(2)\geq -\frac{\vec{c}(1)}{\vec{c}(2)}h$. We have proved that $\vec{Q}_n$ is finite.

  From \cref{lem:inclusionstab}, we deduce that $(\vec{Q}_n)_{n\in\setN}$ is a non-increasing sequence. It follows that there exists $N\in\setN$ such that $\vec{Q}_n=\vec{Q}_N$ for every $n\geq N$. We deduce that $\vec{M}_n=\vec{M}_N$ for every $n\geq N$.
\end{proof}

\subsection{Proofs of \cref{sec:elem-cons}}
In this section, we prove that the node $n_1$ defined in \cref{sec:elem-cons} satisfies \cref{lem:elem_cons}.

\medskip

\begin{algorithm}[t]
  \DontPrintSemicolon
  \Input{%
    A sequence of states $q_1,\ldots,q_k$ and a set $J\subseteq\{1,\ldots,k\}$.
  }
  \Output{%
    A sequence $i_1<j_1<\cdots i_g<j_g$ of elements in $J$ such that $q_{i_\ell}=q_{j_\ell}$ for every $\ell\in\{1,\ldots,g\}$, and such that for every $i,j\not\in \bigcup_{\ell=1}^g\{i_\ell+1,\ldots,j_\ell\}$ such that $1\leq i<j\leq k$ and $q_i=q_j$ we have $i\not\in J$ or $j\not\in J$.
  }
  $g := 0$\;
  $j_0 := 0$\;
  \While{there exists $i,j\in J$ such that $j_g<i<j$ and $q_i=q_j$}{
    let $i\in J$ be minimal such that there exists $j\in J$ such that $j_g<i<j$ and $q_i=q_j$\;
    let $j\in J$ be maximal such that $i<j$ and $q_i=q_j$\;
    let $g:=g+1$, $i_g:=i$, and $j_g:=j$\;
  }
  \Return{$i_1,j_1,\ldots,i_g,j_g$}\;
  \caption{$\mathtt{Decompose}(q_1,\ldots,q_k,J)$}
  \label{algo:decompose}
\end{algorithm}

We first prove the following lemma that provides a decomposition of paths in the exploration by identifying nodes $n$ such that $\vec{z}_n\geq (c,c)$. 
\begin{lemma}\label{lem:elem_cons:1}
  For every $s,n\in N$ such that $s\xrightarrow{*}n$ and $q_s=q_n$, we have:
  $$\vec{z}_n-\vec{z}_s\in \sum_{m\mid s\xrightarrow{*}m\xrightarrow{*}n}\per{\vec{E}_m}+\per{\vec{C}_m}+\sum_{m\rightarrow n}\vec{P}_m$$
\end{lemma}
\begin{proof}
   We put $\vec{K}=\sum_{t\mid s\xrightarrow{*}t\xrightarrow{*}n}\per{\vec{E}_t}+\per{\vec{C}_t}$ and $\vec{P}=\sum_{m\rightarrow n}\vec{P}_m$.
   Let $\vec{v}=\vec{z}_n-\vec{z}_s$. There exists a sequence of $k$ nodes $n_0,\ldots,n_k$ such that $n_0=s$, and $n_k=n$ and such that $n_{j-1}\rightarrow n_j$ for every $j\in\{1,\ldots,k\}$. We introduce the set $M_j=\{m\in N \mid m\rightarrow n_j\}$, and the tuple $(\vec{a}_j,q_j,\vec{z}_j,\vec{P}_j)$ labeling the node $n_j$. We introduce $\vec{p}_1,\ldots, \vec{p}_k$ such that $p_j\in \sum_{m\in M_j}\vec{P}_m$ and $\vec{z}_j=\vec{z}_{j-1}+\vec{b}_j+\vec{p}_j$ where $\vec{b}_j=\vec{a}_j+\sum_{m\in M_j\setminus\{n_{j-1}\}}\vec{z}_m$. Note that $\vec{p}_j\in \vec{P}$ and $(q_{j-1},\vec{b}_j,q_j)$ is a transition of the VASS $\mathcal{V}_n$.

   \medskip

   We introduce $\theta=(q_0,\vec{b}_1,q_1)\cdots(q_{k-1},\vec{b}_k,q_k)$. Notice that $\theta$ is a cycle of $\mathcal{V}_n$. Given $i,j\in \{1,\ldots,k\}$ we denote by $\theta_{[i,j]}$ the sequence $(q_i,\vec{b}_{i+1},q_{i+1})\ldots (q_{j-1},\vec{b}_j,q_j)$ with the convention $\theta_{[i,j]}=\varepsilon$ if $i\geq j$.

  \medskip

  Let $J$ be the set of $j\in \{0,\ldots,k\}$ such that $\vec{z}_j\not\geq (c,c)$.

  \medskip

  Let us consider the sequence $i_1<j_1<\cdots i_g<j_g$ produced by the greedy \cref{algo:decompose}, and let $H=\{0,\ldots,k\}\setminus(\bigcup_{\ell=1}^g\{i_\ell+1,\ldots,j_\ell\})$.

  \medskip

  For every $\ell\in\{1,\ldots,g\}$,
  as $i_\ell,j_\ell\in J$, $i_\ell<j_\ell$ and $q_{i_\ell}=q_{j_\ell}$, we deduce that $\vec{z}_{j_\ell}-\vec{z}_{i_\ell}$ is a vector in $\vec{C}_{n_{j_\ell}}$. As $n_{j_\ell}\xrightarrow{*}n$, it follows that $\vec{C}_{n_{j_\ell}}\subseteq \vec{K}$, and in particular $\vec{z}_{j_\ell}-\vec{z}_{i_\ell}\in \vec{K}$. Observe that $\vec{z}_{j_\ell}-\vec{z}_{i_\ell}$ is also the sum of the displacement of $\theta_{[i_\ell,j_\ell]}$ and $\sum_{i_{\ell}<j\leq j_\ell}\vec{p}_j$. It follows that the displacement of $\theta_{[i_\ell,j_\ell]}$ is in $\vec{K}-\sum_{i_{\ell}<j\leq j_\ell}\vec{p}_j$.
  \medskip

  Now, let $\theta'=\theta_{[j_0,i_1]}\cdots\theta_{[j_g,i_{g+1}]}$ with the convention $j_0=0$ and $i_{g+1}=k$. As $\theta'$ is a cycle it can be decomposed into a finite sequence of elementary cycles. Observe that the source and the target of such an elementary cycle correspond to two distinct nodes $i<j$ in $H$ with $q_i=q_j$. Hence $i\not\in J$ or $j\not\in J$. It follows that the displacement of such a cycle is in $\vec{E}_{n_j}$. Since this set is included in $\vec{K}$, we deduce that the displacement of $\theta'$ is in $\vec{K}$.
 
  \medskip

  Notice that the displacement of $\theta$ is equal to the sum of the displacement of $\theta'$ and the displacements of the cycles $\theta_{[i_\ell,j_\ell]}$ where $\ell\in\{1,\ldots,g\}$, we deduce that the displacement of $\theta$ is in the following set.
  $$\vec{K}-\sum_{\ell=1}^g\sum_{i_{\ell}<j\leq j_\ell}\vec{p}_j$$
  Since the displacement of $\theta$ is also equal to $\sum_{j=1}^k\vec{b}_j=\sum_{j=1}^k(\vec{z}_j-\vec{z}_{j-1}-\vec{p}_j)=\vec{z}_k-\vec{z}_0-\sum_{j=1}^k\vec{p}_j$, we deduce the following relation (recall that $\vec{z}_n=\vec{z}_k$ and $\vec{z}_m=\vec{z}_0$) where $\vec{p}=\sum_{\ell=0}^g\sum_{j_{\ell}<j\leq i_{\ell+1}}\vec{p}_j$.
  $$\vec{z}_n-\vec{z}_m\in \vec{K}+\vec{p}$$
  We have proved the lemma since $\vec{p}\in\vec{P}$.
\end{proof}

\medskip

Given $n\in \des{n_0}$, we introduce the set $\tilde{\vec{C}}_n$ of vectors $\vec{z}_n-\vec{z}_m$ where $m\in N$ is such that $n_0\xrightarrow{*}m\xrightarrow{*}n$ and $q_m=q_n$.
\begin{corollary}\label{cor:elem_cons:1}
  For every $n\in \des{n_0}$, we have $\tilde{\vec{C}}_n\subseteq \sum_{m \mid n_0\xrightarrow{*}m\xrightarrow{*}n}\per{\vec{E}_m}+\per{\vec{C}_m}+\sum_{s\not\in\des{n_0}}\vec{P}_s$.
\end{corollary}
\begin{proof}
  From $\vec{P}_n=\accel{\vec{I}_n}(\sum_{s\rightarrow n}\vec{P}_s)$ for every node $n\in N$, we deduce that $\sum_{s\mid s\rightarrow n}\vec{P}_s\not=\{(0, 0)\}$ for every $n\in \des{n_0}$. In particular $\vec{I}_n=\vec{E}_n+\vec{C}_n$ for every $n\in\des{n_0}$. It follows that $\vec{P}_n\subseteq \per{\vec{E}_n}+\per{\vec{C}_n}+\sum_{s\rightarrow n}\vec{P}_s$ for every $n\in \des{n_0}$. By induction on the well-foundedness of $\rightarrow$, we get $\vec{P}_n\subseteq \sum_{m\mid n_0\xrightarrow{*}m\xrightarrow{*}n}\per{\vec{E}_m}+\per{\vec{C}_m}+\sum_{s\not\in\des{n_0}}\vec{P}_s$ for every node $n\in \des{n_0}$. Now, the corrollary follows from that inclusion and \cref{lem:elem_cons:1}.
\end{proof}

\begin{lemma}\label{lem:citter}
  Let $\vec{x}\geq (-c,-c)$ be a vector in $\setZ^2$, and let $(\vec{z}_j)_{j\in J}$ and $(\vec{z}'_j)_{j\in J}$ be two finite sequences of vectors in $\setN^2$ such that $\alpha(\vec{z}_j)=\alpha(\vec{z}_j')$ for every $j\in J$.
  It holds that
  $\vec{x}+\sum_{j=1}^k\vec{z}_j\geq (0, 0)$
  if, and only if,
  $\vec{x}+\sum_{j=1}^k\vec{z}_j'\geq (0, 0)$.
\end{lemma}
\begin{proof}
   Let $i\in\{1,2\}$ and assume that $\vec{x}(i)+\sum_{j\in J}\vec{z}_j(i)\geq 0$. By symmetry, it is sufficient to prove that $\vec{x}(i)+\sum_{j\in J}\vec{z}'_j(i)\geq 0$. If there exists $j\in J$ such that $\vec{z}_j(i)\geq c$ then $\vec{z}'_j(i)\geq c$. It follows that $\vec{x}(i)+\sum_{j\in J}\vec{z}'_j(i)\geq \vec{x}(i)+c\geq 0$. If for every $j\in J$ we have $\vec{z}_j(i)<c$ then $\vec{z}_j'(i)=\vec{z}_j(i)$ for every $j\in J$. In particular $\vec{x}(i)+\sum_{j\in J}\vec{z}'_j(i)\geq 0$.
\end{proof}

\begin{lemma}\label{lem:elem_cons:2}
  For every $n\in \des{n_0'}$, we have $\vec{E}_n\subseteq \bigcup_{m\mid n_0'\xrightarrow{*}m\xrightarrow{*}n_1}\vec{E}_m+\sum_{s \mid n_0\xrightarrow{*}s\xrightarrow{+}n}\per{\tilde{\vec{C}}_s}$.
\end{lemma}
\begin{proof}
  We introduce a mapping $\lambda:\des{n_0}\rightarrow N_{\min}$ that maps each node $n\in \des{n_0}$ onto a node $\lambda(n)\in N_{\min}$ such that $\lambda(m)\sqsubseteq m$. We extend $\lambda$ on the full set $N$ by $\lambda(n)=n$ if $n\not\in\des{n_0}$. We are going to prove that cycles proving vectors in $\vec{E}_n$ can be transformed by applying $\lambda$ on nodes ocurring along the cycles. More formally, let $n\in \des{n_0'}$ and let $\vec{v}'\in \vec{E}_n$. By definition, $\vec{v}'$ is the displacement of an elementary cycle $\theta'$ of $\mathcal{V}_n$ on $q_{n'}$ for some $n'\in N$ such that $n'\xrightarrow{*}n$, $\vec{z}_{n'}\geq (c,c)$, and such that $n'=n$ or $\vec{v}'\geq (0, 0)$.

\medskip
                
  Let $k$ be the length of $\theta'$, and let us introduce the sequence of states $q_0,\ldots,q_k$ and the sequence of vectors $\vec{b}_1,\ldots,\vec{b}_k$ in $\setZ^2$ such that $\theta'=t_1'\ldots t_k'$ where $t_j'=(q_{j-1},\vec{b}_j',q_j)$ for every $j\in\{1,\ldots,k\}$.
  By construction, we have $\vec{v}' = \sum_{j=1}^k \vec{b}_j'$.

\medskip

Since $t_j'$ is a transition of $\mathcal{V}_n$. There exists a rule $(S_j,\vec{v}_j,q_j)$ of $\mathcal{B}$ and a multiset of $M_j$ of nodes in $\{m \mid m\xrightarrow{+}n\}$ such that $S_j=q_{j-1}+\multiset{q_m \mid m\in M_j}$ and $\vec{b}_j'=\vec{v}_j+\sum_{m \in M_j}\vec{z}_m$.

\medskip

We put $t_j=(q_{j-1},\vec{b}_j,q_j)$ where $\vec{b}_j=\vec{v}_j+\sum_{m \in M_j}\vec{z}_{\lambda(m)}$. Notice that $t_j$ is a transition of $\mathcal{V}_{n_0'}$. We introduce $\theta=t_1\ldots t_k$ and its displacement $\vec{v}$. We have $\vec{v}=\vec{x}+\sum_{j=1}^k\sum_{m\in M_j}\vec{z}_m$ and $\vec{v}'=\vec{x}+\sum_{j=1}^k\sum_{m\in M_j}\vec{z}_{\lambda(m)}$ where $\vec{x}=\sum_{j=1}^k\vec{v}_j$. Since $k\leq |Q|$ we deduce that $\vec{x}\geq (-c,-c)$. From Lemma~\ref{lem:citter} is follows that $\vec{v}\geq (0, 0)$ iff $\vec{v}'\geq (0, 0)$.

\medskip

Let us prove that there exists a node $m$ such that $n_0'\xrightarrow{*}m\xrightarrow{*}n_1$ such that $\vec{v}\in\vec{E}_m$. Assume first that $\vec{v}'\not\geq (0, 0)$. In that case $n'=n$ and in particular $\vec{z}_n\geq (c,c)$ since $\vec{z}_{n'}\geq (c,c)$. As $n\in\des{n_0'}$, it follows that $q_n\in Q_c$. By definition of $n_1$, there exists a node $m$ such that $n_0'\xrightarrow{*}m\xrightarrow{*}n_1$, $q_m=q_n$ and $\vec{z}_m\geq (c,c)$. Since $\theta$ is a cycle of $\mathcal{V}_{n_0'}$ and $n_0'\xrightarrow{*}m$, we deduce that $\theta$ is a cycle of $\mathcal{V}_m$ as well. If follows that $\vec{v}\in \vec{E}_m$. Next, assume that $\vec{v}'\geq (0, 0)$. Then $\vec{v}\geq (0, 0)$. If $n'\not\in\des{n_0}$ then $n'\xrightarrow{*}n_0'\xrightarrow{*}n_1$ by definition of $n_0'$. We deduce that $\vec{v}\in\vec{E}_{n_1}$. If $n'\in\des{n_0}$, there exists $m\in N_{\min}$ such that $m\sqsubseteq n'$. If follows that $q_m=q_{n'}=q_n$ and from $\alpha(\vec{z}_m)=\alpha(\vec{z}_{n'})$ we get $\vec{z}_{m}\geq (c,c)$. As $m\xrightarrow{*}n_0'\xrightarrow n_1$ we deduce that $\vec{v}\in\vec{E}_{n_1}$. Therefore, in all case, there exists a node $m$ such that $n_0'\xrightarrow{*}m\xrightarrow{*}n_1$ such that $\vec{v}\in\vec{E}_m$.

\medskip

Notice that $\vec{z}_{s_j'}-\vec{z}_{s_j}$ is in $\tilde{\vec{C}}_{s_j'}$
for each $j\in \{1,\ldots,k\}$.
Since $\vec{v}'=\vec{v}+\sum_{j=1}^k(\vec{z}_{s_j'}-\vec{z}_{s_j})$, we are done.
\end{proof}

Let us introduce the periodic set $\vec{P}=\sum_{m\in\des{n_0}\setminus (\des{n_1}\setminus\{n_1\})}\per{\vec{E}_m}+\sum_{s\not\in \des{n_0}}\vec{P}_s$. We deduce the following corollary.
\begin{corollary}\label{cor:elem_cons:2}
  For every $n\in \des{n_0}$, we have:
 \begin{align*}
 \vec{E}_n
  \subseteq &
  \sum_{s \mid n_0\xrightarrow{*}s\xrightarrow{+}n}\per{\vec{E}_s}+\per{\vec{C}_s}+\vec{P}
  \end{align*}
\end{corollary}
\begin{proof}
  Notice that if $n\not\in \des{n_1}$ then $\vec{E}_n\subseteq \vec{P}$. So, we can assume that $n\in\des{n_1}$. In particular $n\in \des{n_0'}$ since $n_0'\xrightarrow{*}n_1$. The inclusion then follows from \cref{cor:elem_cons:1} and \cref{lem:elem_cons:2}.
\end{proof}

The proof of \cref{lem:elem_cons} follows from \cref{cor:elem_cons:2} by well-founded induction over the edge relation $\rightarrow$ of the exploration.

\subsection{Proofs of \cref{sec:finitegen}}
\lemConsFinite*
\begin{proof}
    Let $\vec{H}=\con{\vec{U}}+\sum_{m\in \des{n_0}}\con{\vec{C}_m^+}$. If $\vec{U}=\{(1,0),(0,1)\}$ the lemma is proved since $\vec{H}=\setQ_{\geq 0}^2$. If $\vec{U}$ is empty, \cref{lem:cprime} shows that $\bigcup_{m\in M_h}\vec{C}_m^+$ is finite since it is included in $\{0,\ldots,h\}^2$. It follows that $\vec{H}$ is a finitely-generated cone in that case. So, by symmetry, we can assume that $\vec{U}=\{(1,0)\}$. Let $\vec{X}=\{\vec{x}\in \bigcup_{m\in\des{n_0}}\vec{C}_m^+\mid \vec{x}(2)>0\}$. \cref{lem:cprime} shows that for every $\vec{x}\in \vec{X}$, we have $\vec{x}(2)\in\{1,\ldots,h\}$. It follows that there exists $\mu\in\setQ_{\geq 0}$ such that $\mu=\min_{\vec{x}\in\vec{X}}\frac{\vec{x}(1)}{\vec{x}(2)}$. Observe that $\con{\{(1,0),(\mu,1)\}}\subseteq \vec{H}$. In order to prove the converse inclusion, it is sufficient to prove that $\vec{C}_m^+\subseteq \con{\{(1,0),(\mu,1)\}}$ for every $m\in \des{n_0}$. Let $\vec{x}\in\vec{C}_m$. If $\vec{x}(2)=0$ then $\vec{x}\in \setQ_{\geq 0}(1,0)$ and we deduce that $\vec{x}\in \con{\{(1,0),(\mu,1)\}}$. If $\vec{x}(2)>0$ then $\vec{x}\in\vec{X}$ and we deduce that $\frac{\vec{x}(1)}{\vec{x}(2)}\geq \mu$. A direct computation shows that $\vec{x}=(\vec{x}(1)-\mu\vec{x}(2))(1,0)+\vec{x}(2)(\mu,1)$. Hence $\vec{x}\in \con{\{(1,0),(\mu,1)\}}$. We have proved that $\vec{C}_m^+\subseteq \con{\{(1,0),(\mu,1)\}}$. It follows that $\vec{H}=\con{\{(1,0),(\mu,1)\}}$ which is a finitely-generated cone.
\end{proof}

\lemPformula*
\begin{proof}
  Let $\vec{K}$ be the cone of the right handside. Clearly $\vec{E}_m^+$ where $m\in\des{n_0}\setminus (\des{n_1}\setminus\{n_1\})$, $\vec{P}_s$ where $s\not\in \des{n_0}$, $\vec{U}$ and $\vec{C}_m^+$ where $m\in \des{n_0}$ are included in $\con{\vec{P}_N}$. It follows that $\vec{K}\subseteq \con{\vec{P_N}}$.
  We prove the converse inclusion by induction on the well-foundedness of the relation $\rightarrow$, showing that $\vec{P}_n\subseteq \vec{K}$ for every $n\in N$. So, let $n\in N$ such that for every $m\rightarrow n$ we have $\vec{P}_m\subseteq \vec{K}$ and let us prove that $\vec{P}_n\subseteq \vec{K}$. If $n\not\in\des{n_0}$ then clearly $\vec{P}_n\subseteq \vec{K}$. So, we can assume that $n\in\des{n_0}$. Observe that $\vec{P}_n=\accel{\vec{I}_n}(\vec{P}_n')$. It follows that $\con{\vec{P}_n}=\conP{\vec{I}_n\cup\vec{P}_n'}$, and from \cref{lem:conp-stabilizer} we deduce that $\con{\vec{P}_n}$ is $\vec{x}$-stable for every $\vec{x}\in \vec{I}_n$. From $n\in\des{n_0}$ we derive $\vec{P}_n'\not=\{(0, 0)\}$ and in particular $\vec{I}_n=\vec{E}_n\cup\vec{C}_n$. By induction, it follows that $\vec{P}_n'\subseteq\vec{K}$. We introduce the following cone $\vec{R}_n$. Clearly, $\vec{R}_{n}\subseteq \vec{K}$.
  $$\vec{R}_n=\con{\vec{P}_{n}'}+\sum_{m\in\des{n_0}\setminus (\des{n_1}\setminus\{n_1\})}\con{\vec{E}_m^+}+\con{\bigcup_{s\not\in \des{n_0}}\vec{P}_s}$$
  
  We denote by $\vec{M}$ the $h$-mode of $\con{\vec{P}_{n_0}}$. Let us prove that the $h$-mode of $\vec{R}_n$ is equal to $\vec{M}$. Observe that for every node $m$, we have $\vec{I}_m^+\subseteq\vec{P}_m$. Since the exploration is directed, there exists a node $n'\in\des{n}$ such that for every $m\in\des{n_0}\setminus (\des{n_1}\setminus\{n_1\}$ we have $m\xrightarrow{*}n'$ and for every $s\not\in\des{n_0}$ we have $s\xrightarrow{*}n'$. In particular $\vec{E}_m^+\subseteq \vec{P}_{n'}$ for every $m\in\des{n_0}\setminus (\des{n_1}\setminus\{n_1\}$, and $\vec{P}_s\subseteq\vec{P}_{n'}$ for every $s\not\in\des{n_0}$. It follows that $\con{\vec{P}_{n_0}}\subseteq \vec{R}_n\subseteq \con{\vec{P}_{n'}}$. Since those three cones have the same set of axis, and the $h$-modes of $\con{\vec{P}_{n_0}}$ and $\con{\vec{P}_{n'}}$ are equal (to $\vec{M}$), we deduce from \cref{lem:inclusionstab} that the $h$-mode of $\vec{R}_n$ is also equal to $\vec{M}$.

  We introduce the set $\vec{X}_n=\bigcup_{m\in\des{n_0}\setminus (\des{n_1}\setminus\{n_1\})}\vec{E}_m\cup\bigcup_{m\mid n_0\xrightarrow{*}m\xrightarrow{+}n_1}\vec{C}_m\cup\vec{P}_n'\cup\bigcup_{s\not\in\des{n_0}}\vec{P}_s$. Let us prove that $\vec{R}_n$ is $\vec{x}$-stable for every $\vec{x}\in\vec{X}_n$. If $\vec{x}\in\vec{P}_n'$ or $\vec{x}\in \vec{P}_s$ for some $s\not\in \des{n_0}$, clearly $\vec{x}\in \vec{R}_n$ and we conclude that $\vec{R}_n$ is $\vec{x}$-stable. So, we can assume that $\vec{x}\in \vec{E}_m$ for some $m\in \des{n_0}\setminus (\des{n_1}\setminus\{n_1\})$ or $\vec{x}\in \vec{C}_m$ for some $m$ such that $n_0\xrightarrow{*}m\xrightarrow{+}n$. If $\vec{x}\geq (0, 0)$ then $\vec{x}\in \vec{R}_n$ and in particular $\vec{R}_n$ is $\vec{x}$-stable. Assume that $\vec{x}\not\geq (0, 0)$. Notice that $\con{\vec{P}_m}$ is $\vec{x}$-stable by construction. Notice moreover that $\vec{x}\geq (-h,-h)$. It follows that is in the $h$-mode of $\con{\vec{P}_m}$, i.e. $\vec{x}\in\vec{M}$. Since $\vec{M}$ is the $h$-mode of $\vec{R}_n$, we deduce that $\vec{R}_n$ is $\vec{x}$-stable. We have proved that $\vec{R}_n$ is $\vec{x}$-stable for every $\vec{x}\in\vec{X}_n$.

  As $\vec{X}_n$ is a subset of the stabilizer $\stab{\star}{\vec{R}_n}$, which is a cone from \cref{lem:useless-stabilizers}, we deduce that $\con{\vec{X}_n}\subseteq \stab{\star}{\vec{R}_n}$ and $\stab{\star}{\vec{R}_n}\cap\setQ_{\geq 0}^2=\vec{R}_n$. We get $\con{\vec{X}_n}\subseteq \vec{R}_n$. From $\vec{P}_n=\accel{\vec{I}_n}{\vec{P}_{n'}}$, we deduce that $\con{\vec{P}_n}\subseteq \con{\vec{P}_n'\cup \vec{E}_n\cup\vec{C}_n}\cap\setQ_{\geq 0}^2$. \cref{lem:elem_cons} show that $\vec{E}_n\subseteq \con{\vec{X}_n}$. We deduce that $\con{\vec{P}_n}\subseteq\con{\vec{X}_n}\cap\setQ_{\geq 0}^2$. It follows that $\con{\vec{P}_n}\subseteq\vec{R}_n$. From $\vec{R}_n\subseteq\vec{K}$, we have proved the induction step $\vec{P}_n\subseteq \vec{K}$.
\end{proof}

\subsection{Proofs of \cref{sec:wrap-up}}
\lemZdiffStable*
\begin{proof}
  For short,
  let us write $\vec{P} = \vec{P}_n$.
  Consider a node $t \in N$.
  Since $\explo$ is algorithmic,
  we have $\vec{P}_t = \perP{\vec{J}_t}$
  for some subset $\vec{J}_t$ of $\setZ^2$ with $\vec{I}_t \subseteq \vec{J}_t$.
  Hence,
  $\con{\vec{P}_t} = \conP{\vec{J}_t}$.
  It follows from \cref{lem:conp-stabilizer} that
  $\vec{J}_t \subseteq \stab{\star}{\con{\vec{P}_t}}$.
  We thus have shown that $\vec{I}_t \subseteq \stab{\star}{\con{\vec{P}_t}}$,
  for every $t \in N$.
  We now consider two cases.
  If $\vec{P} = \{(0, 0)\}$ then
  $\vec{z}_n - \vec{z}_s$ is in $\overline{\vec{C}}_n \subseteq \vec{I}_n$,
  hence,
  $\vec{z}_n - \vec{z}_s$ is in $\stab{\star}{\con{\vec{P}}}$.
  Otherwise,
  we derive from \cref{{lem:useless-stabilizers}} that
  $\stab{\star}{\con{\vec{P}}}$ is a cone.
  According to \cref{lem:elem_cons:1},
  it holds that
  $\vec{z}_n - \vec{z}_s$ is in
  $\sum_{t\mid s\xrightarrow{*}t\xrightarrow{*}n} \per{\vec{I}_t} + \vec{P}$.
  Observe that for every node $t \in N$ with $s \xrightarrow{*} t \xrightarrow{*} n$,
  we have $\vec{P} = \vec{P}_s \subseteq \vec{P}_t \subseteq \vec{P}_n = \vec{P}$,
  hence,
  $\vec{P}_t = \vec{P}$.
  So $\vec{I}_t \subseteq \stab{\star}{\con{\vec{P}}}$
  for every $t \in N$ with $s \xrightarrow{*} t \xrightarrow{*} n$.
  We also have $\vec{P} \subseteq \stab{\star}{\con{\vec{P}}}$ by \cref{lem:stabcone}.
  It follows that $\vec{z}_n - \vec{z}_s$ is in $\stab{\star}{\con{\vec{P}}}$
  since $\stab{\star}{\con{\vec{P}}}$ is a cone.
\end{proof}

\end{document}